\documentclass{llncs}
\usepackage{amssymb,amsmath,subfigure,graphicx, pxfonts}
\usepackage{array}
\usepackage{rotating}
\usepackage{lmodern,blindtext}
\usepackage{epsfig}
\usepackage{multirow}
\usepackage{qtree}
\usepackage{stmaryrd}
\usepackage{ctable}
\newtheorem{observation}[proposition]{\textbf{Observation}}

\begin{document}
\title{Dynamics of Belief: Abduction, Horn Knowledge Base And Database Updates}
\author{Radhakrishnan Delhibabu}
\institute{Informatik 5, Knowledge-Based Systems Group\\
RWTH Aachen, Germany\\
\email{delhibabu@kbsg.rwth-aachen.de}}
\maketitle
\begin{abstract}
The dynamics of belief and knowledge is one of the major components
of any autonomous system  that should be able to incorporate new
pieces of information. In order to apply the rationality result of
belief dynamics theory to various practical problems, it should be
generalized in two respects: first it should allow a certain part of
belief to be declared as immutable; and second, the belief state
need not be deductively closed. Such a generalization of belief
dynamics, referred to as base dynamics, is presented in this paper,
along with the concept of a generalized revision algorithm for
knowledge bases (Horn or Horn logic with stratified negation). We
show that knowledge base dynamics has an interesting connection with
kernel change via hitting set and abduction. In this paper, we show
how techniques from disjunctive logic programming can be used for
efficient (deductive) database updates. The key idea is to transform
the given database together with the update request into a
disjunctive (datalog) logic program and apply disjunctive techniques
(such as minimal model reasoning) to solve the original update
problem. The approach extends and integrates standard techniques for
efficient query answering and integrity checking. The generation of
a hitting set is carried out through a hyper tableaux calculus and
magic set that is focused on the goal of minimality. The present
paper provides a comparative study of view update algorithms in
rational approach. For, understand the basic concepts with
abduction, we provide an abductive framework for knowledge base
dynamics. Finally, we demonstrate how belief base dynamics can
provide an axiomatic characterization for insertion a view atom to
the database. We give a quick overview of the main operators for
belief change, in particular, belief update versus database update.

\vspace{0.2cm}

\textbf{Keyword}: AGM, Belief Revision, Belief Update, Horn
Knowledge Base Dynamics, Kernel Change, Abduction, Hyber Tableaux,
Magic Set, View update, Update Propagation.
\end{abstract}
\section{Introduction}

We live in a constantly changing world, and consequently our beliefs
have to be revised whenever there is new information. When we
encounter a new piece of information that contradicts our current
beliefs, we revise our beliefs \emph{rationally}.

In the last three decades, the field of computer science has grown
substantially beyond mere number crunching, and aspires to imitate
rational thinking of human beings. A separate branch of study,
\emph{artificial intelligence} (AI) has evolved, with a number of
researchers attempting to represent and manipulate knowledge in a
computer system. Much work has been devoted to study the statics of
the knowledge, i.e. representing and deducting from fixed knowledge,
resulting in the development of expert systems. The field of logic
programming, conceived in last seventies, has proved to be an
important tool for handling static knowledge. However, such fixed
Horn knowledge based systems can not imitate human thinking, unless
they are accomplish revising their knowledge in the light of new
information. As mentioned before, this revision has to take place
rationally. This has led to a completely new line of research, the
\textbf{dynamics of belief}.

Studies in dynamics of belief are twofold: What does it mean to
\emph{rationally} revise a belief state? How can a belief state be
represented in a computer and revised? The first question is more
philosophical theory, and a lot of works have been carried out from
epistemological perspective to formalize belief dynamics. The second
question is computation oriented, and has been addressed differently
from various perspectives of application. For example, a lot of
algorithms have been proposed in logic programming to revise a Horn
knowledge base or a database represented as a logic program; number
of algorithms are there to carry out a view update in a rational
database; algorithm to carry out diagnosis; algorithm for abduction
reasoning and so on. We need the concept of "change" in some form or
other and thus need some axiomatic characterization to ensure that
the algorithms are rational. Unfortunately, till this date, these
two tracks remain separate, with minimal sharing of concepts and
results. The primary purpose of the paper is to study these
two developments and integrate them.

When a new piece of information is added to a Horn knowledge base
(Delgrande 2008 and Delgrande \& Peppas 2011), (Papini 2000) it may become
inconsistent. Revision means modifying the Horn knowledge base in
order to maintain consistency, by keeping the new information and
removing the least possible previous information. In our case,
update means revision and contraction, that is insertion and
deletion in database perspective. Previous works (Aravindan \& Dung
1994), (Aravindan 1995) have explained connections between
contraction and knowledge base dynamics. Our Horn knowledge base
dynamics is defined in two parts: an immutable part (Horn formulae)
and updatable part (literals) (for definition and properties see the
works of Nebel 1998, Segerberg 1998, Hansson et al 2001 and Ferm{\'e} $\&$ Hansson 2001).
Knowledge bases have a set of integrity constraints. In the case of
finite knowledge bases, it is sometimes hard to see how the update
relations should be modified to accomplish certain Horn knowledge
base updates.
.

\section{Motivation}
In the general case of arbitrary formulae, the revision problem for
knowledge bases is hard to solve. So we restrict the revision
problem to \emph{Horn formulae}. The connection between belief
change and database change is an interesting one since so far the
two communities have independently considered two problems that are
very similar, and our aim is to bring out this connection.

We aim to bridge the gap between philosophical and database theory.
In such a case, Hansson's (Hansson 1997) kernel change is related to
the abductive method. Aliseda's (Aliseda 2006) book on abductive
reasoning is one of our key motivation. Wrobel's (Wrobel 1995)
definition of first-order theory revision was helpful to frame our
algorithm. On the other hand, we are dealing with the view update
problem. Keller and Minker's (Keller 1985 and Minker 1996) work is
one the motivation for the view update problem.
In Figure 1 understand the concept of view update problem in
rational way. Figure show that foundation form Belief Revision
theory, intermediate step handle to Horn knowledge base, this step
very impairment that agent have background knowledge and he/she made
decision with postulate may require to process next step. Target of
the application is connect database updates via Horn knowledge base
with abduction reasoning. All clear procedure shown in each section.

\hspace{3cm}

\begin{figure}[h]
\begin{center}
   \includegraphics[height=6cm, width=12cm,angle=0]{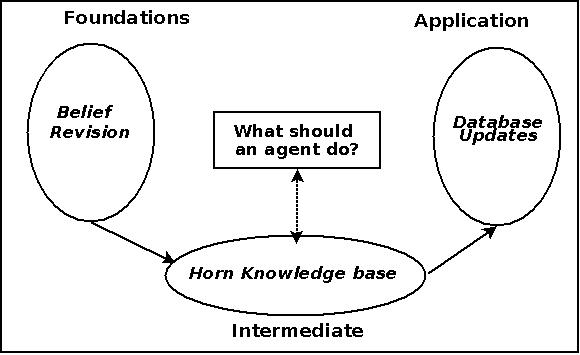}
   \caption{Layout of the paper}
   \label{Figure 1.1}
   \end{center}
\end{figure}

Following example illustrates the motivation of the paper:
\begin{example} \label{e10} Consider a database with an (immutable) rule that a staff
member is a person who is currently working in a research group
under a chair. Additional (updatable) facts are that matthias and
gerhard are group chairs, and delhibabu and aravindan are staff
members in group infor1. Our first integrity constraint (IC) is that
each research group has only one chair ie. $\forall x,y,z$ (y=z)
$\leftarrow$ group\_chair(x,y) $\wedge$ group\_chair(x,z). Second
integrity constraint is that a person can be a chair for only one
research group ie. $\forall x,y,z$ (y=z)$\leftarrow$
group\_chair(y,x) $\wedge$ group\_chair(z,x).
\end{example}

\begin{center}
\underline {Immutable part}: staff\_chair(x,y)$\leftarrow$
staff\_group(x,z),group\_chair(z,y). \vspace{0.5cm}

\underline{Updatable part}: group\_chair(infor1,matthias)$\leftarrow$ \\
\hspace{2.4cm}group\_chair(infor2,gerhard)$\leftarrow$ \\
\hspace{2.6cm}staff\_group(delhibabu,infor1)$\leftarrow$ \\
\hspace{2.6cm}staff\_group(aravindan,infor1)$\leftarrow$ \\
\end{center}
Suppose we want to update this database with the information,
staff\_chair({aravin\-dan},{gerhard}); From the immutable part, we
can deduce that this can be achieved by asserting
staff\_group(\underline{aravindan},z) $\bigwedge$
group\_chair(z,\-\underline{gerhard})

If we are restricted to definite clauses, there are three plausible
ways to do this: first case is, aravindan and gerhard belong to
infor1, i.e, staff\_group(\underline{aravind}\-\underline{an},infor1) $\bigwedge$
group\_chair(info1,\underline{gerhard}). We need to delete both
base facts group\_\-chair(infor1,matthias)$\leftarrow$ and
group\_chair(infor2,gerhard)$\leftarrow$, because our first IC as
well as second IC would be violated otherwise. In order to change
the view, we need to insert
group\_chair(i\-nfor1,gerhard)$\leftarrow$ as a base fact. Assume
that we have an algorithm that deletes the base facts
staff\_group(delhibabu,\-infor1)$\leftarrow$ from the database. But,
no rational person will agree with such an algorithm, because the
fact staff\_group(delhibabu,infor1)$\leftarrow$ is not "relevant" to
the view atom.

Second case, aravindan and gerhard belong to infor2, that is
staff\_group(\underline{aravin}\-\underline{dan},\-infor2) $\bigwedge$
group\_chair(infor2,\underline{gerhard}). Simply, insert the new
fact staff\_group\-(aravindan,\-infor2)$\leftarrow$ to change the
view. Suppose an algorithm deletes the base facts
staff\_group(ara\-vindan,infor1)$\leftarrow$ from the database, then
it can not be "rational" since these facts are not "relevant" to the
view atom.

Third case, aravindan and gerhard belong to infor3 (free assignment
of the group value), that is
staff\_group(\underline{aravindan},infor3) $\bigwedge$
group\_chair(info3,\underline{gerhard}). Suppose, we insert new base
fact group\_chair(infor3,gerhard) $\leftarrow$, our second IC does
not follow. Suppose an algorithm inserts the new base fact
staff\_group(aravin\-dan,infor2)$\leftarrow$ or
staff\_group(aravindan,infor1)$\leftarrow$ is deleted, then it can
not be "rational".

The above example highlights the need for some kind of "relevance
policy" to be adopted when a view atom is to be inserted to a
deductive database. How many such axioms and policies do we need to
characterize a "good" view update? When are we sure that our
algorithm for view update is "rational"? Clearly, there is a need
for an axiomatic characterization of view updates. By axiomatic
characterization, we mean explicitly listing all the rationality
axioms that are to be satisfied by any algorithm for view update.

The basic idea in (Behrend \& Manthey 2008), (Aravindan \&
Baumgartner 1997) is to employ the model generation property of
hyper tableaux and magic set to generate models, and read off
diagnosis from them. One specific feature of this diagnosis
algorithm is the use of semantics (by transforming the system
description and the observation using an initial model of the
correctly working system) in guiding the search for a diagnosis.
This semantical guidance by program transformation turns out to be
useful for database updates as well. More specifically we use a
(least) Herbrand model of the given database to transform it along
with the update request into a logic program in such a
way that the models of this transformed program stand for possible
updates.

We discuss two ways of transforming the given database together with
the view update (insert and delete) request into a logic
program resulting in two variants of view update algorithms. In the
first variant, a simple and straightforward transformation is
employed. Unfortunately, not all models of the transformed program
represent a rational update using this approach. The second variant
of the algorithm uses the least Herbrand model of the given database
for the transformation. In fact what we referred to as offline
preprocessing before is exactly this computation of the least
Herbrand model. This variant is very meaningful in applications
where views are materialized for efficient query answering. The
advantage of using the least Herbrand model for the transformation
is that all models of the transformed logic program (not
just the minimal ones) stand for a rational update.

When dealing with the revision of a Horn knowledge base (both
insertions and deletions), there are other ways to change a Horn
knowledge base and it has to be performed automatically also (Ferm{\'e} 1992 and Rodrigues\& Benevidas 1994).
Considering the information, change is precious and must be
preserved as much as possible. The \emph{principle of minimal
change} (G$\ddot{a}$rdenfors 1998, Dalal 1988 and Herzig \& Rifi 1999), (Schulte 1999) can provide a
reasonable strategy. On the other hand, practical implementations
have to handle contradictory, uncertain, or imprecise information,
so several problems can arise: how to define efficient change in the
style of Carlos Alchourr$\acute{o}$n, Peter G$\ddot{a}$rdenfors, and
David Makinson (AGM) (Alchourron et al. 1985b); what result has to
be chosen (Lakemeyer 1995), (Lobo \& Trajcevski 1997), (Nayak et al.
2006); and finally, according to a practical point of view, what
computational model to explore for the Horn knowledge base revision has
to be provided?

The rest of the paper is organized as follows: First we start with
preliminaries in Section 3. In Section 4, we give a quick overview
of belief changes and belief update. In Section 5, we introduce
knowledge base dynamics along with the concept of generalized
revision, and revision operator for knowledge base. Section 6
studies the relationship between knowledge base dynamics and
abduction and shows how abductive procedures could be used to
realize revision. In Section 7, we give a quick overview of belief
update versus knowledge base update. In Section 8, we discuss an
important application of knowledge base dynamics in providing an
axiomatic characterization for insertion view atoms to databases;
and nature of view update problem for incomplete to complete
information shown. We give a quick overview of belief update versus
database update in Section 9. In Section 10, we provide an abductive framework for 
Horn knowledge base dynamics in first order version. In Section 11,
we give brief overview of related works. In Section 12, we draw
conclusions with a summary of our contribution and indicate future
directions of our investigation.


\section{Preliminaries}

We consider a propositional language $\mathcal{L_P}$ defined from a
finite set of propositional variables $\mathcal{P}$ and the standard
connectives. We use $a, b, x, y,...\varphi, \phi,
\psi, ...$ for propositional formulae. Sets of formulae are denoted
by upper case Roman letters $A,B, F,K, ....$. A literal is an atom
(positive literal), or a negation of an atom (negative literal).

For any formula $\varphi$, we write $E(\varphi)$ to mean the set of
the elementary letters that occur in $\varphi$. The same notation
also applies to a set of formulae. For any set $F$ of formulae,
$L(F)$ represents the sub-language generated by $E(F)$, i.e. the set
of all formulae $\varphi$ with $E(\varphi) \subseteq E(F)$.

Horn formulae are defined (Delgrande \& Peppas 2011) as follows:
\begin{enumerate}
\item[1.] Every $a \in \Phi$ where $\Phi \in \mathcal{L_P} \cup \{ \bot \}$ , $a$ and $\neg a$ are Horn clauses.
\item[2.] $a \leftarrow a_1 \land a_2 \land ... \land a_n$ is a Horn clause, where $n \geq 0$ and
$a, a_i \in \Phi$ ($1 \leq i \leq n$).
\item[3.] Every Horn clause is a Horn formula, $a$ is called head and $a_i$
is body of the Horn formula.
\item[4.] If $\varphi$ and $\psi$ are Horn formulae, so is $\varphi\land \psi$.
\end{enumerate}

A definite Horn clause is a finite set of literals (atoms) that
contains exactly one positive literal which is called the head of
the clause. The set of negative literals of this definite Horn
clause is called the body of the clause. A Horn clause is
non-recursive, if the head literal does not occur in its body. We
usually denote a Horn clause as head$\leftarrow$body. Let
$\mathcal{L_H}$ be the set of all Horn formulae with respect to
$\mathcal{L_P}$. A formula $\phi$ is a syntactic consequence within
$\mathcal{L_P}$ of a set $\Gamma$ of formulas if there is a formal
proof in $\mathcal{L_P}$ of $\phi$ from the set $\Gamma
\vdash_{\mathcal{L_P}} \phi$.

A Horn logic with stratified negation (Jackson \& Schulte 2008) is
similar to a set of Horn formulae. An immutable part is a
function-free clause of the form $a \leftarrow a_1 \land a_2 \land
... \land a_n$, with $n\geq 1$ where $a$ is an atom denoting the
immutable part's head and $a_1 \land a_2 \land ... \land a_n$ are
literals. i.e. positive or negative atoms, representing the body of
the Horn clauses.

Formally, a finite Horn knowledge base $KB$ (Horn or Horn logic with
stratified negation) is defined as a finite set of formulae from
language $\mathcal{L_{H}}$, and divided into three parts: an
immutable theory $KB_{I}$ is a Horn formula (head$\leftarrow$body),
which is the fixed part of the knowledge; updatable theory $KB_{U}$
is a Horn clause (head$\leftarrow$); and integrity constraint
$KB_{IC}$ representing a set of clauses (Horn logic with stratified
negation) ($\leftarrow$body).

\begin{definition} [Horn Knowledge Base] \label{D1} A Horn knowledge base, KB
is a finite set of Horn formulae from language $\mathcal{L_{H}}$,
s.t $KB=KB_{I}\cup KB_{U}\cup KB_{IC}$, $KB_{I}\cap
KB_{U}=\varnothing$ and $KB_{U}\cap KB_{IC}=\varnothing$.
\end{definition}

Horn knowledge base change deals with situations in which an agent
has to modify its beliefs about the world, usually due to new or
previously unknown incoming information, also represented as
formulae of the language. Common operations of interest in Horn
knowledge base change are the expansion of an agent's current Horn
knowledge base KB by a given Horn clause $\varphi$ (usually denoted
as KB+$\varphi$), where the basic idea is to add regardless of the
consequences, and the revision of its current beliefs by $\varphi$
(denoted as KB * $\varphi$), where the intuition is to incorporate
$\varphi$ into the current beliefs in some way while ensuring
consistency of the resulting theory at the same time. Perhaps the
most basic operation in Horn knowledge base change, like belief
change, is that of contraction AGM (Alchourron et al. 1985b), which
is intended to represent situations in which an agent has to give up
$\varphi$ from its current stock of beliefs (denoted as
KB-$\varphi$).

\begin{definition} [AGM Contraction] Let KB be a Horn knowledge base, and $\alpha$ a belief
that is present in KB. Then \emph{contraction} of KB by $\alpha$,
denoted as $KB-\alpha$, is a consistent belief set that  ignore
$\alpha$
\end{definition}

\begin{definition} [Levi Identity] \label{D2} Let - be an AGM contraction
operator for KB. A way to define a revision is by using Generalized
Levi Identity:
\begin{center}
$KB*\alpha~=~(KB-\neg\alpha)\cup\alpha$
\end{center}
\end{definition}

Then, the revision can be trivially achieved by expansion, and the
axiomatic characterization could be straightforwardly obtained from
the corresponding characterizations of the traditional models
(Alchourron, CE et al 1985). The aim of our work is not to define revision
from contraction, but rather to construct and axiomatically
characterize revision operators in a direct way.

\section{Belief Changes}

Working at an abstract philosophical level, the aim of
belief dynamics is to formalize the rationality of change, without
worrying much about the syntactic representation of belief. However,
it is not possible to completely ignore belief representation, and
works on belief dynamics assume as little necessary things as
possible about the representation of the belief. In this Section
based on Konieczny's (Konieczny 2011) work, we recall the definition
of the main belief change operators and the links between them. We
focus on the classical case, where the belief represent use
propositional logic. This is a very quick presentation of belief
change theory. For a complete introduction the reader is referred to
seminal books on belief revision ((G$\ddot{a}$rdenfors 1992 \&
1998), (Hansson 1997a), (Rott 2001)) or the recent special issue of
Journal of Philosophical Logic on the 25 Years of AGM Theory (Ferme
\& Hansson 2011).

A belief base K is a finite set of propositional formulae. In order
to simplify the notations we identify the base K with the formula.
which is the conjunction of the formulae of K
\footnotemark \footnotetext{There are two major interpretations of belief bases. One of them, supported by Dalal 1998, uses belief bases as mere expressive devices; hence if Cn(B1) =
Cn(B2) then B1 and B2 represent the same belief state and yield the same
outcome under all operations of change. The other, more common approach treats inclusion in the belief base as
epistemically significant. The belief base contains those sentences that have
an epistemic standing of their own (Ferme 2011)}

Belief revision aims at changing the status of some beliefs in the
base that are contradicted by a more reliable piece of information.
Several principles are govern this revision operation:

\begin{itemize}
 \item First is the primacy of update principle: the new piece of information has to be
accepted in the belief base after the revision. This is due to the
hypothesis that the new piece of information is more reliable than
the current beliefs \footnotemark
\footnotetext{If this is not the case one should use a
non-prioritized revision operator (Hansson 1997b)}
  \item Second is the principle of coherence: the new belief base after the revision should
be a consistent belief base. Asking the beliefs to be consistent is
a natural requirement if one wants to conduct reasoning tasks from
her belief base
  \item Third is the principle of minimal change: the new belief base after the revision
should be as close as possible from the current belief base. This
important principle aims at ensuring that no unnecessary information
(noise) is added to the beliefs during the revision process, and
that no unnecessary information is lost during the process:
information/beliefs are usually costly to obtain, we do not want to
throw them away without any serious reason.
\end{itemize}

let $\psi$ and $\mu$ be two formulae denoting respectively the
belief base, and a new piece of information. Then $\psi\circ\mu$ is
a formula representing the new belief base. An operator $\circ$ is
an AGM belief revision operator if it satisfies the following
properties.

\hspace{0.5cm}

\begin{definition} [Belief revision] \label{D3}
\hspace{0.5cm}
\begin{enumerate}
\item[]\hspace{-0.6cm}(R1) $\psi\circ\mu$ implies $\mu$.
\item[]\hspace{-0.6cm}(R2) If $\psi\land\mu$ is satisfiable, then
$\psi\circ\mu\equiv \psi\land \mu$.
\item[]\hspace{-0.6cm}(R3) If $\mu$ is satisfiable, then so is $\psi\circ\mu$.
\item[]\hspace{-0.6cm}(R4) If $\psi_{1}\equiv\psi_{2}$ and $\mu_{1}\equiv\mu_{2}$,
then $\psi_{1}\circ\mu_{1}\equiv \psi_{2}\circ \mu_{2}$.
\item[]\hspace{-0.6cm}(R5) $(\psi\circ\mu)\land \phi$ implies $\psi\circ(\mu\land
\phi)$.
\item[]\hspace{-0.6cm}(R6) If $(\psi\circ\mu)\land \phi$ is satisfiable, then
$\psi\circ(\mu\land\phi)$ implies $(\psi\circ\mu)\land\phi$.
\end{enumerate}
\end{definition}

When one works with a finite propositional language the above
postulates,~proposed by Katsuno and Mendelzon (Katsuno, H \& Mendelzon, AO. 1991b), are equivalent to AGM ((Alchourron et al 1985b) and (G$\ddot{a}$rdenfors 1998))
(R1) states that the new piece of information must be believed after
the revision. (R2) says that when there is no conflict between the
new piece of information and the current belief, the revision is
just the conjunction. (R3) says that revision always a
consistent belief base, unless the new piece of information is not
consistent. (R4) is an irrelevance of syntax condition, it states
that logically equivalent bases must lead to the same result. (R5)
and (R6) give conditions on the revision by a conjunction.

AGM also defined contraction operators, that aim to remove some
piece of information from the beliefs of the agent. These
contraction operators are closely related to revision operators,
since each contraction operator can be used to define a revision
operator, through the Levy identity and conversely each revision
operator can be used to define a contraction operator through the
Harper identity ((Alchourron et al 1985b) and (G$\ddot{a}$rdenfors
1998)). So one can study indifferently revision or contraction
operators. So we focus on revision here.

Several representation theorems, that give constructive ways to
define AGM revision/ contraction operators, have been proposed, such
as partial meet contraction/revision (Alchourron et al 1985b),
epistemic entrenchments (G$\ddot{a}$rdenfors 1992) and
(G$\ddot{a}$rdenfors \& Makinson 1988), safe contraction (Alchourron
et al 1985a), etc. In (Katsuno \& Mendelzon 1991b and 1992)
(Benferhat et al. 2005), Katsuno and Mendelzon give a representation
theorem, showing that each revision operator corresponds to a
faithful assignment, that associates to each base a plausibility
preorder on interpretations (this idea can be traced back to Grove
systems of spheres (Grove 1988)).

\begin{sloppypar}
Assume a total pre-order $\leq_{\psi}$ on W (set of possible world). That is to say, KB = $min(W,\leq_{\psi})$. As usual we take $\leq_{\psi}$ to be an ordering of plausibility on the worlds, with worlds lower down in the ordering seen as
more plausible. In what follows $\simeq_\psi$ will always denote the symmetric closure of $\leq_{\psi}$, i.e., $W1 \simeq_\psi W2$ iff both $W1\leq_{\psi} W2$ and $W2\leq_{\psi} W1$. 
\end{sloppypar}

\begin{definition} [\text{(Konieczny 2011)}] \label{D4} A faithful assignment is a function mapping each
base $\psi$ to a pre-order $\leq_{\psi}$ over interpretations such
that
\begin{enumerate}
  \item if $\omega \models\psi$ and $\omega'\models \psi,$ then
  $\omega \simeq_\psi \omega'$
  \item if $\omega \models\psi$ and $\omega'\not\models \psi,$ then
  $\omega <_\psi \omega'$
  \item if $\psi\equiv\psi'$, then $\leq_{\psi}= \leq_{\psi'}$
\end{enumerate}
\end{definition}

\begin{sloppypar}
\begin{theorem} [\text{(Katsuno \& Mendelzon 1991b and 1991b)}] \label{T1} An operator $\circ$ is an AGM revision operator
(i.e. it satisfies (R1)-(R6)) if and only if there exists a faithful
assignment that maps each base $\psi$ to a total pre-order
$\leq_{\psi}$ such that $mod(\psi\circ\mu)$= min$(mod(\mu),\leq
_{\psi})$.
\end{theorem}
\end{sloppypar}

\begin{proof} Follows from the definition \ref{D4} and the result of Konieczny 2011.
\end{proof}

One of the main problems of this characterization of belief revision
is that it does not constrain the operators enough for ensuring a
good behavior when we do iteratively several revisions. So one needs
to add more postulates and to represent the beliefs of the agent
with a more complex structure than a simple belief base. In
(Darwiche \& Pearl 1997) Darwiche and Pearl proposed a convincing
extension of AGM revision. This proposal have improved by an
additional condition in ((Booth \& Meyer 2006) and (Jin \&
Thielscher 2007), (Konieczny, et al. 2010) and (Konieczny \&
Pino Prez 2008)) define improvement operators that are a
generalization of iterated revision operators.

Whereas belief revision should be used to improve the beliefs by
incorporating more reliable pieces of evidence, belief update
operators aim at maintaining the belief base up-to-date, by allowing
modifications of the base according to a reported change in the
world. This distinction between revision and update was made clear
in  Katsuno $\&$ Mendelzon 1991 and 1992, where  Katsuno $\&$ Mendelzon 1991 proposed postulates for belief update.

\begin{definition} [Belief Update] \label{D5}

An operator $\diamond$ is a (partial) update operator if it
satisfied the properties (U1)-(U8). It is a total update operator if
it satisfies the property (U1)-(U5),(U8),(U9).\\
\item[]\begin{enumerate}
\item[]\hspace{-0.6cm}(U1) $\psi\diamond\mu$ implies $\mu$.
\item[]\hspace{-0.6cm}(U2) $if \psi$ implies $\mu$, then $\psi\diamond\mu \equiv
\psi$
\item[]\hspace{-0.6cm}(U3) $if \psi$ not implies $\bot$ and $\mu$ not implies
$\bot$ then $\psi\diamond\mu$ not implies $\bot$
\item[]\hspace{-0.6cm}(U4) If $\psi_{1}\equiv\psi_{2}$ and $\mu_{1}\equiv\mu_{2}$,
then $\psi_{1}\diamond\mu_{1}\equiv \psi_{2}\diamond\mu_{2}$.
\item[]\hspace{-0.6cm}(U5) If $(\psi\diamond\mu)\land \phi$ implies
$\psi\diamond(\mu\land\phi)$
\item[]\hspace{-0.6cm}(U6) If $(\psi\diamond\mu_{1})$ implies $\mu_{2}$ and $(\psi\diamond\mu_{2})$ implies $\mu_{2}$
then $\psi\diamond\mu_{1}\equiv\psi\diamond\mu_{2}$
\item []\hspace{-0.6cm}(U7) If $\psi$ is a complete formula, then
$(\psi\diamond\mu_{1})\land(\psi\diamond\mu_{2})$ implies
$\psi\diamond(\mu_{1}\vee\mu_{2})$
\item []\hspace{-0.6cm}(U8)
$(\psi_{1}\vee\psi_{2})\diamond\mu\equiv(\psi_{1}\diamond\mu)\vee(\psi_{2}\diamond\mu)$
\item[]\hspace{-0.6cm}(U9) If $\psi$ is a complete formula and
$(\psi\diamond\mu)\land\psi$ not implies $\perp$ then
$\psi\diamond(\mu\land\psi)$ implies $(\psi\diamond\mu)\land\psi$
\end{enumerate}
\end{definition}

Most of these postulates are close to the ones of revision. The main
differences lie in postulate (U2) that is much weaker than (R2):
conversely to revision, even if the new piece of information is
consistent with the belief base, the result is generally not simply
the conjunction. This illustrates the fact that revision can be seen
as a selection process of the most plausible worlds of the current
beliefs with respect to the new piece information, whereas update is
a transition process: each world of the current beliefs have to be
translated to the closest world allowed by the new piece of
information. This world-by-world treatment is expressed by postulate
(U8).

As for revision, there is a representation theorem in terms of
faithful assignment.

\begin{definition} [\text{[25]}] \label{D6} A faithful assignment is a function mapping each
interpretation $\omega$ to a pre-order $\leq_{\omega}$ over
interpretations such that if $\omega \neq \omega'$, then $\omega
<_{\omega} \omega'$
\end{definition}

\begin{theorem} \label{T2} An update operator $\diamond$ satisfies (U1)-(U8) if
and only if there exists a faithful assignment that maps each
interpretation $\omega$ to a partial pre-order $\leq_{\omega}$ such
that $mod(\psi\diamond\mu)$=$\bigcup_{\omega\models\psi}$
min$(mod(\mu),\leq_{\omega})$.
\end{theorem}

\begin{proof} Follows from the observation and the
result of Konieczny 2011.
\end{proof}

But there is also a second theorem corresponding to total
pre-orders.

\begin{theorem} \label{T4} An update operator $\diamond$ satisfies (U1)-(U5),
(U8) and (U9) if and only if there exists a faithful assignment that
maps each interpretation $\omega$ to a total preorder
$\leq_{\omega}$ such that
$mod(\psi\diamond\mu)$=$\bigcup_{\omega\models\psi}$
min$(mod(\mu),\leq_{\omega})$.
\end{theorem}

\begin{proof} Follows from the observation and the
result of Konieczny 2011.
\end{proof}

\begin{definition}[\text{[3]}] \label{D8} Let M=(W.w) be a $K$-model and $\mu$ a formula.
A k-model $M'=(W',w')$ is called a possible resulting $k$-model
after updating M with $\mu$ if and only if the following conditions
hold:
\begin{enumerate}
  \item M'$\models \mu$;
  \item there does not exist another $k$-model $M''=(W'',w'')$ such
  that $M'' \models \mu$ and $M''<_M M'$.
\end{enumerate}
The set of all possible resulting k-models after updating M with
$\mu$ as Res(M,$\mu$).
\end{definition}

\begin{theorem} \label{T5} Knowledge update operator $\diamond$ defined in definition \ref{D8} satisfies
(U1)-(U9).
\end{theorem}

\begin{proof} Follows from the definition \ref{D8} and
the result of Baral $\&$ Zhang 2005.
\end{proof}

\begin{note}
Horn knowledge base is a subset of belief base, $KB\subseteq B$, so
everything that follows for belief base, also follows for Horn
Knowledge base.
\end{note}


\section{Horn Knowledge base dynamics}

In the AGM framework, a belief set is represented by a
deductively closed set of propositional formulae. While such sets
are infinite, they can't always be finitely representable. However,
working with deductively closed, infinite belief sets is not very
attractive from a computational point of view. The AGM approach to
belief dynamics is very attractive in its capture the rationality
of change, but it is not always easy to implement either Horn
formula based partial meet revision or generalized kernel revision.
In artificial intelligence and database applications, what is
required is to represent the knowledge using a finite Horn knowledge
base. Further, a certain part of the knowledge is treated as
immutable and should not be changed.

AGM (Alchourron et al. 1985b) proposed a formal framework in which
revision is interpreted as belief change. In this section, we focus
on the Horn knowledge base revision and propose new rationality
postulates that are adopted from AGM postulates for revision.

\begin{definition} \label{D9} Let KB be a Horn knowledge base with an immutable part
$KB_{I}$. Let $\alpha$ and $\beta$ be any two Horn clauses from
$\mathcal{L_H}$. Then, $\alpha$ and $\beta$ are said to be
\emph{KB-equivalent} iff the following condition is satisfied:
$\forall$ set of Horn clauses E $\subseteq \mathcal{L_H}$ and IC:
$KB_{I}\cup E \cup IC \vdash\alpha$ iff $KB_{I}\cup E \cup IC\vdash\beta$.
\end{definition}

These postulates stem from three main principles: the new item of
information has to appear in the revised Horn knowledge base, the
revised base has to be consistent and revision operation has to
change the least possible beliefs. Now we consider the revision of a
Horn clause $\alpha$ with respect to KB, written as $KB*\alpha$. The rationality
postulates for revising $\alpha$ from KB can be formulated.

\vspace{1cm}
\newpage

\begin{definition} [Rationality postulates for Horn knowledge base
revision] \label{10}

\hspace{0.3cm}

\begin{enumerate}
\item[]\hspace{-0.6cm}(KB*1)\hspace{0.2cm}  \emph{Closure:} $KB*\alpha$ is a Horn knowledge base.
\item[]\hspace{-0.6cm}(KB*2)\hspace{0.2cm}  \emph{Weak Success:} if $\alpha$ is consistent with $KB_{I}\cup KB_{IC}$ then
$\alpha \subseteq KB*\alpha$.
\item[]\hspace{-0.6cm}(KB*3.1)  \emph{Inclusion:} $KB*\alpha\subseteq
Cn(KB\cup\alpha)$.
\item[]\hspace{-0.6cm}(KB*3.2)  \emph{Immutable-inclusion:} $KB_{I}\subseteq
Cn(KB*\alpha)$.
\item[]\hspace{-0.6cm}(KB*4.1)  \emph{Vacuity 1:} if $\alpha$ is
inconsistent with $KB_{I}\cup KB_{IC}$ then $KB*\alpha=KB$.
\item[]\hspace{-0.6cm}(KB*4.2)  \emph{Vacuity 2:} if $KB\cup \alpha \nvdash \perp$ then $KB*\alpha$ = $KB \cup
\alpha$.
\item[]\hspace{-0.6cm}(KB*5)\hspace{0.3cm}   \emph{Consistency:} if $\alpha$ is consistent with $KB_{I}\cup KB_{IC}$
 then $KB*\alpha$ consistent with $KB_{I}\cup KB_{IC}$.
\item[]\hspace{-0.6cm}(KB*6)  \hspace{0.2cm} \emph{Preservation:} If $\alpha$ and $\beta$ are
KB-equivalent, then $KB*\alpha \leftrightarrow KB*\beta$.
\item[]\hspace{-0.6cm}(KB*7.1)  \emph{Strong relevance:} $KB*\alpha\vdash \alpha$ if $KB_{I}\nvdash\neg\alpha$
\item[]\hspace{-0.6cm}(KB*7.2)  \emph{Relevance:} If $\beta\in KB\backslash KB*\alpha$,
then there is a set $KB'$ such that\\ $KB*\alpha\subseteq
KB'\subseteq KB\cup\alpha$, $KB'$ is consistent $KB_{I}\cup KB_{IC}$
with $\alpha$, but $KB' \cup \{\beta\}$ is inconsistent $KB_{I}\cup
KB_{IC}$ with $\alpha$.
\item[]\hspace{-0.6cm}(KB*7.3)  \emph{Weak relevance:} If $\beta\in KB\backslash KB*\alpha$,
then there is a set $KB'$ such that $KB'\subseteq KB\cup\alpha$,
$KB'$ is consistent $KB_{I}\cup KB_{IC}$ with $\alpha$, but $KB'
\cup \{\beta\}$ is inconsistent $KB_{I}\cup KB_{IC}$ with $\alpha$.
\end{enumerate}
\end{definition}

To revise KB by $\alpha$, only those informations that are
relevant to $\alpha$ in some sense can be added (as example in the
introduction illustrates). $(KB*7.1)$  is very strong axiom allowing
only minimum changes, and certain rational revision can not be
carried out. So, relaxing this condition (example with more details
can be found in (Aravindan 1995, Hansson, SO 1997a, Ferme \& Hansson 2011 and Falappa, MA et al 2012), this can be weakened to
relevance. $(KB*7.2)$ is relevance policy that still can not permit
rational revisions, so we need to go next step. With $(KB*7.3)$ the
relevance axiom is further weakened and it is referred to as
"core-retainment".

\subsection{Revision function}

Suppose that we want to revise Horn knowledge base KB with respect
to a single clause without using negation. We may construct revision
using generalizing techniques from classical belief (base) revision
(Falappa et al. 2012). Partial meet revision operator is syntax
dependent and based on the foundational approach. In order to define
it, first we need to define $\alpha$-consistent-remainders.

\begin{definition}[Remainder Set] \label{D11} Let a Horn knowledge base KB be a set of Horn formulae, where
$\alpha$ is Horn clause. The $\alpha$-consistent-remainders of KB,
noted by $KB\downarrow_{\top}\alpha$, is the set of $KB'$ such that:

\begin{enumerate}
  \item $KB'\subseteq KB$, ensuring that $KB_{I}\subseteq KB'$ and $KB_{IC}\subseteq
  KB'$.
  \item $KB'\cup\alpha$ is consistent with $KB_{I}\cup KB_{IC}$.
  \item For any KB'' such that $KB'\subset KB''\subseteq KB$ then
  $KB''\cup\alpha$ is inconsistent with $KB_{I}\cup KB_{IC}$.

\end{enumerate}
\end{definition}

That is, $KB\downarrow_{\top}\alpha$ is the set of maximal
KB-subsets consistent with $\alpha$.

\begin{example} Suppose that KB=\{$KB_{I}: p\leftarrow a\wedge b,p\leftarrow
a,q\leftarrow a\wedge b;~\raggedright{KB_{U}: a\leftarrow,} \newline
b\leftarrow;~KB_{IC}: \emptyset$\} and $\alpha$= $\leftarrow p$.
Then we have that:

-~$KB\downarrow_{\top}\alpha$= \{\{$p\leftarrow a\wedge b$\},
\{$p\leftarrow a$\},\{$\leftarrow a$\}, \{$\leftarrow b$\}\}.
\end{example}

Revision by a Horn clause is based on the concept of a
$\alpha$-consistent-remainders. In order to complete the
construction, we must define a selection function that selects
consistent remainders.

\subsection{Principle of minimal change}

Let a Horn knowledge base KB be a set of Horn formulae and $\psi$ is
a Horn clause such that $KB=\{\phi~|~\psi\vdash\phi\}$ is derived by
$\phi$. Now we consider the revision of a Horn clause $\alpha$ wrt
KB, that is $KB*\alpha$.

The principle of minimal change (PMC) leads to the definition of
orders between interpretations. Let $\mathcal{I}$ be the set of all
the interpretations and $Mod(\psi)$ be the set of models of $\psi$.
A pre-order on $\mathcal{I}$, denoted $\leq_{\psi}$ is linked with
$\psi$. The relation $<_{\psi}$ is defined from $\leq_{\psi}$ as
usual:
$$I<_{\psi}I'~\text{iff}~I\leq_{\psi}I'~\text{and}~I'\nleq_{\psi}I.$$

The pre-order $\leq_{\psi}$ is \it faithful \rm to $\psi$ if it
verifies the following conditions:

\begin{enumerate}
\item[1)] If $I,I'\in Mod(\psi)$ then $I<_{\psi}I'$ does not hold;
\item[2)] If $I\in Mod(\psi)$ and $I'\notin Mod(\psi)$ then
$I<_{\psi}I'$ holds;
\item[3)] if $\psi\equiv \phi$ then $\leq_{\psi}=\leq_{\phi}$.
\end{enumerate}

A minimal interpretation may thus be defined by:

$\mathcal{M}\subseteq \mathcal{I}$, the set of minimal
interpretations in $\mathcal{M}$ according to $\leq_{\psi}$ is
denoted $Min(\mathcal{M},\leq_{\psi})$. And $I$ is minimal in
$\mathcal{M}$ according to $\leq_{\psi}$, if $I\in\mathcal{M}$ and
there is no $I'\in \mathcal{M}$ such that $I'<_{\psi}I$.

Revision operation * satisfies the postulates (KB*1) to (KB*6) and
(KB*7.3) if and only if there exists a total pre-order $\leq_{\psi}$
such that: \begin{equation*} Mod(\psi *
\phi)=Min(Mod(\phi),\leq_{\psi}).
\end{equation*}

\begin{definition} [Selection function] \label{D12} Let KB be a Horn knowledge
base. $\gamma$ is a selection function for KB iff for all Horn
clauses $\alpha$

\begin{enumerate}
  \item If $KB\downarrow_{\top} \alpha \neq {\emptyset}$ then ${\emptyset} \neq
\gamma(KB\downarrow_{\top} \alpha)\subseteq KB\downarrow_{\top}
\alpha$.

  \item If $KB\downarrow_{\top} \alpha = {\emptyset}$ then $\gamma(KB\downarrow_{\top} \alpha) = \{KB\}$
\end{enumerate}
\end{definition}

\begin{observation} \label{z2}
Let KB, KB' be an Horn knowledge base, KB' be consistent. Suppose that $\alpha \in KB$ and $\alpha \in KB'$. Then $\alpha \in X$ for all $X \in KB\downarrow_{\top} KB'$ and, therefore, $\alpha\in\bigcap(KB\downarrow_{\top} KB')$.\\
\end{observation}

From the above observation and definition \ref{D11} it follows that all the Horn knowledge base of $KB\cap\alpha$ are "protected", in the sense that they are included in the intersection of any collection of remainders. That is, a consolidated selection function selects a subset of the set $KB\downarrow_{\top} \alpha$ whose elements all contain the set $KB\cap\alpha$.

\begin{definition} [Partial meet revision] \label{D13} Let KB be a Horn knowledge base with an immutable part
$KB_{I}$ and $\gamma$ a selection function for KB. The partial meet
revision on KB that is generated by $\gamma$ is the operator
$*_\gamma$ such that, for all Horn clauses $\alpha$:

$KB*_\gamma \alpha =\left\{ \begin{array}{cc} \cap
\gamma(KB\downarrow_{\top}\alpha)\cup \alpha~~~~~&
\text{if}~\alpha ~\text{is consistent with}~KB_{I}\cup KB_{IC}\\
KB&\text{otherwise.}
\end{array}\right.$
\end{definition}

An operator * is a generalized revision (partial meet revision) on
KB if and only if there is a selection function $\gamma$ for KB such
that for all Horn clauses $\alpha$, KB*$\alpha$ = KB$*_\gamma
\alpha$.

\begin{example}

Given KB=\{$KB_{I}: p\leftarrow a\wedge b,p\leftarrow a,q\leftarrow
a\wedge b ;~\raggedright{KB_{U}: a\leftarrow, b\leftarrow;}
\newline ~KB_{IC}: \emptyset$\}, $\alpha$= $\leftarrow p$ and
$KB\downarrow_{\top}\alpha$= \{\{$p\leftarrow a\wedge b$\},
\{$p\leftarrow a$\}, \{$\leftarrow a$\}, \{$\leftarrow b$\}\}. We
have two possible results for the selection function and its
associated partial meet revision operator
\begin{eqnarray*}
\gamma_1(KB\downarrow_{\top}\alpha)&=&\{\leftarrow p\}~\text{and}~
KB*_{\gamma_1} \alpha=\{p\leftarrow a\wedge b ,\leftarrow a,
\leftarrow b\}\\
\gamma_2(KB\downarrow_{\top}\alpha)&=&\{\leftarrow p\}~\text{and}~
KB*_{\gamma_2} \alpha=\{p\leftarrow a, \leftarrow a\}
\end{eqnarray*}

The partial meet revision on KB that is generated by $\gamma_2$
gives minimal interpretation with respect to PMC.
\end{example}

\begin{theorem} \label{T6} For every Horn knowledge base $KB$, $*$ is a
generalized revision function iff it satisfies the postulates (KB*1)
to (KB*6) and (KB*7.3).
\end{theorem}

\begin{proof}
\hspace{0.5cm}

(\textbf{If part})~* satisfies (KB*1) to (KB*6) and (KB*7.3). We
must show that $*$ is a generalized revision. When
$KB_{I}\vdash\alpha$, (KB*1) to (KB*6) and (KB*7.3) imply that
$KB*\alpha=KB$ is coinciding with generalized revision.

When $KB_{I}\vdash\neg\alpha$, the required result follows from the
two observations:
\begin{itemize} 
\item[1.] $\exists KB'\in KB\downarrow_{\top}\alpha$ s.t.$KB*\alpha\subseteq
KB'$ (when $KB*\alpha = KB \cup \{ \alpha \}$)\\
Let $\gamma$ be an selection function for $KB$ and $*_\gamma$ be the
generalized revision on $KB$ that is generated by $\gamma$. Since
* satisfies closure (KB*1), $KB*_{\gamma}\alpha$ is KB contained in $\alpha$.
Also, satisfaction of weak success postulate (KB*2) ensures that
$\alpha\subseteq KB*_{\gamma}\alpha$. Every element of
$KB\downarrow_{\top}\alpha$ is a inclusion maximal subset that does
drive $\alpha$, and so any subset of KB that does derive $\alpha$
must be contained in a member of $KB\downarrow_{\top}\alpha$.

\item[2.] $\bigcap(KB\downarrow_{\top}\alpha)\subseteq KB*_{\gamma}\alpha$
(when $KB*\alpha = KB \cup \{ \alpha \}$)\\
Consider any $\beta \in \bigcap(KB\downarrow_{\top}\alpha)$. Assume
that $\beta\not\in KB*\alpha$. Since * satisfies weak relevance
postulate (KB*7.3), it follows that there exists a set KB' s.t.
$KB'\subseteq KB\cup \alpha$; $KB'$ is consistent with $\alpha$;
and $KB' \cup \{\beta\}$ is inconsistent with $\alpha$. But this
contradicts the that $\beta$ is present in every maximal subset of
KB that does derive $\alpha$. Hence $\beta$ must not be in
$KB*_{\gamma}\alpha$.
\end{itemize}

(\textbf{Only if part})~Let $KB*\alpha$ be a generalized revision of
$\alpha$ for KB. We have to show that $KB*\alpha$ satisfies the
postulate (KB*1) to (KB*6) and (KB*7.3).

Let $\gamma$ be an selection function for $KB$ and $*_\gamma$ be the
generalized revision on $KB$ that is generated by $\gamma$.
\begin{description}
\item[Closure] Since $KB*_{\gamma}\alpha$ is a Horn knowledge base, this
postulate is trivially shown.

\item[Weak Success] Suppose that $\alpha$ is consistent. Then it is trivial by definition that $\alpha \subseteq
KB*_{\gamma} \alpha$.

\item[Inclusion] Since every $X \in KB\downarrow_{\top}
\alpha$ is such that $X \subseteq KB$ then this postulate is
trivially shown.

\item[Immutable-inclusion] Since every $X \in KB\downarrow_{\top}
\alpha$ is such that $X \subseteq KB_I$ then this postulate is
trivially shown.

\item[Vacuity 1] Trivial by definition.

\item[Vacuity 2] If $KB\cup \{\alpha\}$ is consistent then
$KB\downarrow_{\top} \alpha$ = \{\{KB\}\}. Hence
$KB*_{\gamma}\alpha$ = $KB\cup \{\alpha\}$.

\item[Consistency] Suppose that $\alpha$ is consistent. Then $KB \downarrow_{\top} \alpha \neq= \emptyset$ and by
definition, every $X \in KB \downarrow_{\top} \alpha$ is consistent
with $\alpha$. Therefore, the intersection of any subset of $KB
\downarrow_{\top} \alpha$ is consistent with $\alpha$. Finally,
$KB*_{\gamma} \alpha$ is consistent.

\item[Uniformity] If $\alpha$ and $\beta$ are KB-equivalent, then $KB\downarrow_{\top}
\alpha$ = $KB\downarrow_{\top} \beta$

\item[Weak relevance] Suppose that $KB \downarrow_{\top} \alpha \neq \emptyset$. Let $\beta\in KB$.
Then there is some $X \in KB \downarrow_{\top} \alpha$ such that
$\beta\notin X$. Therefore, there is some $X$ such that $\beta\notin
X \subseteq KB$, $X \cup \alpha$ is consistent but $X \cup \alpha
\cup \{\beta\}$ is inconsistent.

Suppose that $KB \downarrow_{\top} \alpha = \{\emptyset\}$ in which
case $\alpha$ is inconsistent. By definition, $KB*_{\gamma} \alpha =
KB$ and weak relevance is vacuously satisfied.
\end{description}
\end{proof}

\subsection{Horn knowledge base revision with hitting set} \label{s11}

In this section, we show that Horn knowledge base revision has an
interesting connection with kernel change via hitting set.

\subsubsection{Kernel revision system}\hspace{0.5cm}

\hspace{0.1cm}

To revise a Horn formula $\alpha$ from a Horn knowledge base KB, the
idea of kernel revision is to \emph{keep at least} one element from
every inclusion-minimal subset of KB that derives $\alpha$. Because
of the immutable-inclusion postulate, no Horn formula from $KB_{I}$
can be deleted.

\begin{definition} [Kernel sets] \label{D14}
Let a Horn knowledge base KB be a set of Horn formulae, where
$\alpha$ is Horn clause. The $\alpha$-inconsistent kernel of KB,
noted by $KB\bot_{\bot}\alpha$, is the set of $KB'$ such that:

\begin{enumerate}
  \item $KB'\subseteq KB$ ensuring that $KB_{I}\subseteq KB'$ and $KB_{IC}\subseteq
  KB'$.
  \item $KB'\cup\alpha$ is inconsistent with $KB_{I}\cup KB_{IC}$ .
  \item For any KB'' such that $KB'' \subset KB'\subseteq KB$ then
  $KB''\cup\alpha$ is consistent with $KB_{I}\cup KB_{IC}$.
\end{enumerate}
\end{definition}

That is, given a consistent $\alpha$, $KB\bot_{\bot}\alpha$ is the
set of minimal KB-subsets inconsistent with $\alpha$.

\begin{example}

Suppose that KB=\{$KB_{I}: p\leftarrow a\wedge b,p\leftarrow
a,q\leftarrow a\wedge b ;~KB_{U}: a\leftarrow, b\leftarrow;
~KB_{IC}: {\emptyset}$\} and $\alpha$= $\leftarrow p$. Then we have
that:

$KB\bot_{\bot}\alpha$= \{$\{p\leftarrow a \wedge b\}, \{p\leftarrow
a \} $\}.
\end{example}

Revision by a Horn clause is based on the concept of a
$\alpha$-inconsistent-kernels. In order to complete the
construction, we must define a incision function that cuts in each
inconsistent-kernel.

\begin{definition}[Incision function] \label{D15} Let $KB$ be a set of Horn formulae.
$\sigma$ is a incision function for $KB$ if and only if, for all
consistent Horn clauses $\alpha$
\begin{enumerate}
  \item $\sigma(KB\bot_\bot \alpha) \subseteq \bigcup KB\bot_\bot
  \alpha$
  \item If $KB' \in KB\bot_\bot\alpha$ then
  $KB'\cap(\sigma(KB\bot_\bot\alpha))\neq 0$
\end{enumerate}
\end{definition}

\begin{definition} [Hitting set] \label{D16} A \emph{hitting set} H for $KB\bot_{\bot}\alpha$ is defined as
a set s.t. (i) $H\subseteq\bigcup(KB\bot_{\bot}\alpha)$, (ii) $H\cap
KB_{I}$ is empty and (iii) $\forall X \in KB\bot_{\bot}\alpha$,
$X\neq\emptyset$ and $X\cap KB_{U}$ is not empty, then $X\cap H
\neq\emptyset$.
\end{definition}

A hitting set is said to be \emph{maximal} when $H$ consists of all
updatable statements from $\bigcup(KB\bot_{\bot}\alpha)$ and
\emph{minimal} if no proper subset of $H$ is a hitting set for
$KB\bot_{\bot}\alpha$.

\begin{observation}
Let KB, KB' be an Horn knowledge base, KB' be consistent. Suppose that $\alpha \in KB$ and $\alpha \in KB'$. Then $\alpha\not\in \bigcup(KB\bot_{\bot} KB')$ and, therefore, $KB' \cap \bigcup(KB\bot_{\bot} KB')=0$
\end{observation}

From the above observation and definition \ref{D15} it follows that all the Horn knowledge base of $\alpha$ are "protected", in the sense that they can not be considered for removing
by the consolidated incision function. That is, a consolidated incision function
selects among the sentences of $KB\backslash\alpha$ that make $KB\cup\alpha$ inconsistent.

\begin{definition} [Generalized Kernel revision] \label{D17} An incision function for KB is a function s.t. for all
$\alpha$, $\sigma(KB\bot_{\bot}\alpha)$ is a hitting set for
$KB\bot_{\bot}\alpha$.  Generalized kernel revision on KB that is generated by $\sigma$ is the operator $*_{\sigma}$ such that, for all Horn clauses $\alpha$:
$$KB*_\sigma \alpha =\left\{ \begin{array}{cc} (KB\backslash
\sigma(KB\bot_{\bot}\alpha)\cup \alpha~~~~~&
\text{if}~\alpha ~\text{is consistent}~KB_{I}\cup KB_{IC}\\
KB&\text{otherwise.}
\end{array}\right.$$
\end{definition}

An operator $*$ is a generalized kernel revision for KB if and only if there is an incision function $\sigma$ for KB such that for all Horn clauses $\alpha$, $KB*\alpha$
= $KB*_{\sigma}\alpha$.

From the definition of hitting set, it is clear that when
$KB\vdash\neg\alpha$, $\alpha$ is the hitting set of
$KB\bot_\bot\alpha$. On the other hand, when $KB_{I}\vdash\alpha$,
the definition ensures that only updatable elements are inserted,
and $\alpha$ does follow from the revision. Thus, weak success
(KB*2), immutable-inclusion (KB*3.2) and vacuity (KB*4.1) are
satisfied by generalized kernel revision of  $\alpha$ from KB.

\begin{example}

Given KB=\{$KB_{I}: p\leftarrow a\wedge b,p\leftarrow a,q\leftarrow
a\wedge b ;~KB_{U}: a\leftarrow, b\leftarrow; ~KB_{IC}: {\emptyset}$
\}, $\alpha$= $\leftarrow p$ and $KB\bot_{\bot}\alpha =
\{\{p\leftarrow a\wedge b\}, \{ p\leftarrow a \}\}.$ We have two
possible results for the incision function  and its associated
kernel revision operator:
\begin{eqnarray*}
\sigma_1(KB\bot_{\bot}\alpha)&=&\{p\leftarrow a\wedge
b\}~\text{and}~
KB*_{\sigma_1} \alpha=\{\{\leftarrow a \}, \{ \leftarrow b \} \},\\
\sigma_2(KB\bot_{\bot}\alpha)&=&\{p\leftarrow a \}~\text{and}~
KB*_{\sigma_2} \alpha=\{\{\leftarrow a \} \}.
\end{eqnarray*}
Incision function $\sigma_2$ produces minimal hitting set for
$KB\bot_{\bot}\alpha$.
\end{example}

\begin{theorem} \label{T7} For every Horn knowledge base $KB$, $*_{\sigma}$ is a
generalized kernel revision function iff it satisfies the postulates
(KB*1) to (KB*6) and (KB*7.3).
\end{theorem}

\begin{proof}
\hspace{0.5cm}

(\textbf{If part})~* satisfies (KB*1) to (KB*6) and (KB*7.3). We
must show that $*$ is a generalized kernel revision. Let $\sigma$ be
a incision function and $\alpha$ Horn formula. When
$KB_{I}\vdash\alpha$, (KB*1) to (KB*6) and (KB*7.3) imply that
$KB*\alpha=KB$ coincides with generalized revision and follow PMC.

When $KB_{I}\vdash\neg\alpha$, the required result follows from the
two observations:
\begin{itemize}
\item[1.] $\exists KB'\in KB\bot_{\bot}\alpha$ s.t.$KB*\alpha\subseteq
KB'$ (when $KB_{I}\vdash \alpha$)\\
Let $\sigma$ be an incision function for $KB$ and $*_\sigma$ be the
generalized revision on $KB$ that is generated by $\sigma$. Since
* satisfies closure (KB*1), $KB*_{\sigma}\alpha$ is KB contained in $\alpha$.
Also, satisfaction of weak success postulate (KB*2) ensures that
$\alpha\subseteq KB*_{\sigma}\alpha$. Every element of
$KB\bot_{\bot}\alpha$ is a inclusion minimal subset that does derive
$\alpha$, and so any subset of KB that does derive $\alpha$ must be
contained in a member of $KB\bot_{\bot}\alpha$.

\item[2.] $\bigcap(KB\bot_{\bot}\alpha)\subseteq KB*_{\sigma}\alpha$
(when $KB_{I}\vdash \alpha$)\\
Consider any $\beta \in \bigcap(KB\bot_{\bot}\alpha)$. Assume that
$\beta\not\in KB*\alpha$. Since * satisfies weak relevance postulate
(KB*7.3), it follows that there exists a set KB' s.t. $KB'\subseteq
KB\cup \alpha$; $KB'$ is a consistent with $\alpha$; and $KB' \cup
\{\beta\}$ is inconsistent with $\alpha$. But this contradicts that
$\beta$ is present in every minimal subset of KB that does derive
$\alpha$. Hence $\beta$ must not be in $KB*_{\sigma}\alpha$.
\end{itemize}

(\textbf{Only if part})~Let $KB*\alpha$ be a generalized revision of
$\alpha$ for KB. We have to show that $KB*\alpha$ satisfies the
postulate (KB*1) to (KB*6) and (KB*7.3).

Let $\sigma$ be an incision function for $KB$ and $*_\sigma$ be the
generalized revision on $KB$ that is generated by $\sigma$.

\begin{description}
\item[Closure] Since $KB*_{\sigma}\alpha$ is a Horn knowledge base, this
postulate is trivially shown.

\item[Weak Success] Suppose that $\alpha$ is consistent. Then it is trivial by definition that $\alpha \subseteq
KB*_{\sigma} \alpha$.

\item[Inclusion] Trivial by definition.

\item[Immutable-inclusion] Since every $X \in KB\bot_{\bot}
\alpha$ is such that $X \subseteq KB_I$ then this postulate is
trivially shown.

\item[Vacuity 1] Trivial by definition.

\item[Vacuity 2] If $KB\cup \{\alpha\}$ is consistent then
$KB\bot_{\bot} \alpha$ = \{\{KB\}\}. Hence $KB*_{\sigma}\alpha$ =
$KB\cup \{\alpha\}$.

\item[Consistency] Suppose that $\alpha$ is consistent. Then $KB \bot_{\bot} \alpha \neq= \emptyset$ and by
definition, every $X \in KB \bot_{\bot} \alpha$ is consistent with
$\alpha$. Therefore, the intersection of any subset of $KB
\bot_{\bot} \alpha$ is consistent with $\alpha$. Finally,
$KB*_{\sigma} \alpha$ is consistent.

\item[Uniformity] If $\alpha$ and $\beta$ are KB-equivalent, then $KB\bot_{\bot}
\alpha$ = $KB\bot_{\bot} \beta$

\item[Weak relevance] Let $\beta\in KB$ and $\beta\notin KB*_{\sigma}\alpha$. Then
$KB*_{\sigma}\alpha \neq KB$ and, from the definition of
$*_{\sigma}$,it follows that:\\

\hspace{2.5cm}$KB*_{\sigma}\alpha$=$(KB\backslash\sigma(KB\bot_{\bot}\alpha))\cup\alpha$\\

Therefore, from $\beta\not\in
(KB\backslash\sigma(KB\bot_{\bot}\alpha))\cup\alpha$ and $\beta\in
KB$, we can conclude that $\beta\in\sigma(KB\bot_{\bot}\alpha)$. By
definition $\sigma(KB\bot_{\bot}\alpha) \subseteq \bigcup
KB\bot_{\bot}\alpha$, and it follows that there is some $X\in
KB\bot_{\bot}\alpha$ such that $\beta \in X$. X is a minimal
KB-subset inconsistent with $\alpha$. Let $Y=X\backslash \{\beta\}$.
Then Y is such that $Y\subset X\subseteq KB \subseteq KB\cup\alpha$.
Y is consistent with $\alpha$ but $Y\cup\{\beta\}$ is consistent
with $\alpha$.
\end{description}
\end{proof}

From the Theorem \ref{T6} and \ref{T7}, it immediately follows that
a revision operation on a Horn knowledge base is a generalized
kernel revision iff it is a generalized revision. The following
theorem formalizes this with additional insights into the
relationship between kernel and generalized revisions.

\begin{theorem} \label{T8}
\hspace{0.5cm}
\begin{enumerate} 
\item[1.] A revision operation over a Horn knowledge base KB is a
generalized kernel revision over KB iff it is a generalized revision
over KB.
\item[2.] When the incision function $\sigma$ is minimal, i.e. the
hitting set defined by $\sigma$ is inclusion-minimal, then the
generalized kernel revision defined by $\sigma$ is a partial meet
revision of $\alpha$ from KB.
\item[3.]When the incision function $\sigma$ is maximal, i.e.
$\sigma (KB\bot_\bot\alpha)$ consists of all updatable statements
from $\bigcup (KB\bot_\bot\alpha)$, then the kernel contraction
defined by $\sigma$ is the minimal generalized revision of $\alpha$
from KB.
\end{enumerate}
\end{theorem}

\begin{proof}
Follows from the Definition \ref{D16}, Theorem \ref{T6} and \ref{T7}
\end{proof}

\section{Knowledge base dynamics and abduction}
In this Section, we study the relationship between Horn knowledge
base dynamics, discussed in the previous Section, and abduction that
was introduced by the philosopher Peirce (see Aliseda 2006, Boutilier \& Beche 1995 and Pagnucco 1996). We show
how an abduction grammar could be used to realize revision with an
immutability condition. A special subset of literals (atoms) of
language $\mathcal{L_{H}}$, \emph{abducibles} Ab, are designated for
abductive reasoning. Our work is based on atoms (literals), so we
combine Christiansen and Dahl (Christiansen \& Dahl 2009) grammars
approach. Simply, we want to compute abducibles for Horn knowledge
base (Horn or Horn logic with stratified negation).

\begin{example} \label{E15}
Consider a Horn logic with stratified negation knowledge base KB
with immutable part $KB_I$, updatable part $KB_U$ and integrity
constraint $KB_{IC}$.

$$\begin{array}{cccccc} KB_I:&flies(x)\leftarrow bird(x), not~ab(x),&KB_U:&\hspace{-1.6cm}bird(tweety)\leftarrow
&\hspace{-1cm}KB_{IC}:\emptyset\\
&\hspace{-0.4cm}ab(x)\leftarrow broken\_wing(x)&&\hspace{-1.9cm}bird(opus)\leftarrow&&\\
&&&broken\_wing(tweety)\leftarrow&&\end{array}$$
\end{example}

If we observe that tweety flies, there is a good reason to assume
that the wound has already healed. Then, removing the fact
broken\_wing(tweety) from the $KB$ explains the observation
flies(tweety). On the other hand, suppose that we later notice that
opus does not fly anymore. Since flies(opus) is entailed by $KB_I$,
we now have to revise the Horn knowledge base to block the
derivation of flies(opus) by assuming, for instance,
broken\_wing(opus). In nonmonotonic theories, deletion of formulae
may introduce new formulae, thus positive ($\Delta^{+}$) and
negative ($\Delta^{-}$) explanations play a complementary role in
accounting for an observation in nonmonotonic theories. (more
explanation in Sakama \& Inoue 2003)

\begin{definition}[Abductive grammar] \label{D18} A abductive grammar $\Gamma$ is a 6-tuple $\langle
N,T,IC,\\KB,R,S\rangle$ where:

\begin{enumerate}
\item[-] N are nonterminal symbols in the immutable part ($KB_I$).
\item[-] T is a set of terminal symbols in the updatable part ($KB_U$).
\item[-] IC is the set of integrity constraints for the Horn knowledge base ($KB_{IC}$).
\item[-] KB is the Horn knowledge base which consists of $KB=KB_{I}\cup KB_{U}\cup KB_{IC}$.
\item[-] R is a set of rules, $R\subseteq KB$.
\item[-] S is the revision of literals (atoms), called the start symbol.
\end{enumerate}
\end{definition}

\begin{example} \label{E6}
Consider a Horn knowledge base KB (with immutable part $KB_I$, updatable part $KB_U$ and integrity constraint $KB_{IC}$) and a Horn clause $\alpha$ (p is $\alpha$) be revise.

$$\begin{array}{cccccc} KB_I:&p\leftarrow
q\wedge a&\hspace{0.5cm}KB_U:&a\leftarrow&\hspace{1.2cm}KB_{IC}:&\leftarrow b\\
&p\leftarrow r\wedge b&&r\leftarrow&&\\
&q\leftarrow c\wedge d&&&&\\
&r\leftarrow e\wedge f&&&&\\
&\hspace{-0.6cm}p\leftarrow b&&&&\end{array}$$
\end{example}

KB be a Horn knowledge base, represented by the grammar
($\Gamma=\langle N,T,KB,R,S\rangle$) as follows:

\begin{center}
\begin{itemize}
  \item []\hspace{1cm} N=\{p\}
  \item []\hspace{1cm} T=\{a,b,c,d,e,f,q,r\}
  \item []\hspace{1cm} IC=\{b\}
  \item []\hspace{1cm} KB=$KB_I\cup KB_U\cup KB_{IC}$
  \item []\hspace{1cm} R=\{$p\leftarrow q,a;~p\leftarrow r,b;~q\leftarrow c,d;~r\leftarrow
  e,f;~p\leftarrow b;~a;~r$\}
  \item []\hspace{1cm} S=\{p\}
\end{itemize}
\end{center}

\begin{definition}[Constraint system] \label{D19}
A constraint system for abduction is a pair $\langle
KB^{Ab},\\KB^{BG} \rangle$, where $KB^{Ab}(\Delta)$ is a set of
propositions (abducibles) and $KB^{BG}$ a background Horn knowledge
base.
\end{definition}

\begin{note} In the sequel, without any loss of generality, we assume that
$KB_I$ is a set of rules and $KB_U$ is a set of abducibles from Horn
knowledge base perspective. With respect to the considered grammars,
$KB^{BG}$ is a set all Horn formulae from R and $KB^{Ab}$ is set of
abducibles from T.
\end{note}

\begin{note} \label{N1}
Given a Horn knowledge base KB and a Horn clause $\alpha$, the
problem of abduction is to explain $\alpha$ in terms of an
abduction, i.e. to generate a set of abducibles $KB^{Ab}$, $\Delta$
s.t. $KB^{BG}\cup\Delta\vdash \alpha$.
\end{note}

\begin{definition}[Minimal abductive explanation] \label{D20} Let KB be a Horn knowledge base and $\alpha$ an
observation to be explained. Then, for a set of abducibles
$(KB^{Ab})$, $\Delta$ is said to be an abductive explanation with respect to
$KB^{BG}$ iff $KB^{BG}\cup \Delta\vdash \alpha$. $\Delta$ is said to
be \emph{minimal} with respect to $KB^{BG}$ iff no proper subset of $\Delta$ is
an abductive explanation for $\alpha$, i.e. $\nexists\Delta^{'}$
s.t. $KB^{BG}\cup\Delta^{'}\vdash\alpha$.
\end{definition}

Since an incision function is adding and removing only updatable
elements from each member of the kernel set, to compute a
generalized revision of $\alpha$ from KB, we need to compute only
the abduction in every $\alpha$-kernel of KB. So, it is now
necessary to characterize precisely the abducibles present in every
$\alpha$-kernel of KB. The notion of minimal abductive explanation
is not enough to capture this, and we introduce locally minimal and
KB-closed abductive explanations explanations.

\begin{definition}[Local minimal abductive explanations] \label{D21}
Let $KB^{BG'}$ be a subset of $KB^{BG}$, s.t $\Delta$ is a minimal
abductive explanation of $\alpha$ with respect to $KB^{BG'}$ (for some
$\Delta$). Then $\Delta$ is called local minimal for $\alpha$ with respect to
$KB^{BG}$.
\end{definition}

\begin{example} \label{E7}
From example \ref{E6}, suppose $\{ p \leftarrow q \wedge a, p
\leftarrow a\}$, where $a$ and $f$ are abducibles in the grammar
system R. Clearly, $\Delta_1=\{a\}$ is the only minimal abductive
explanation for $p$ with respect to R. $\Delta_2=\{a,q\}$ is an abductive
explanation for $p$ with respect to R, but not a minimal one. However,
$\Delta_2$ is a locally minimal abductive explanation for $p$ with respect to R,
since it is a minimal explanation for $p$ with respect to $\{ p \leftarrow q
\wedge a\}$ which is a subset of R.
\end{example}

The concept of locally minimal abductive explanation is
computationally attractive, since minimal abductive explanation is
more expensive to compute (Aravindan 1995). To find a minimal
admissible and denial literal (atom) from $KB^{Ab}$ that is positive
and negative literal (atom) from $KB^{Ab}$, we need to introduce a
constraint system (C) with integrity constraint (IC).

\begin{definition} [Constraint abduction system] \label{22} A constrained abductive grammar is a pair
$\langle \Gamma,C\rangle$, where $\Gamma$ is an abductive grammar
and $C$ a constraint system for abduction, $\Gamma$=$\langle N,T,R,S
\rangle$ and C=$\langle KB^{BG},KB^{Ab},IC\rangle$.
\end{definition}

Given a constrained abductive grammar $\langle \Gamma,C\rangle$ as
above, the constrained abductive recognition problem for $\tau \in
T^*$ is the problem of finding an admissible and denial knowledge
base (Horn knowledge base contained set of positive and negative
literal (atoms)) from $KB^{Ab}$ and such that $\tau \in
\mathcal{L}_{P}(\Gamma_{KB^{Ab}})$ where
$\mathcal{L}_{P}(\Gamma_{KB^{Ab}})$ is propositional language over
abducibles in $\Gamma$, where $\Gamma_{KB^{Ab}}$ = $\langle N,T,
KB^{BG}\cup KB^{Ab},R,S \rangle$. In this case, $KB^{Ab}$ is called
a \emph{constrained (abductive) system of} $\tau$. Such that
$KB^{Ab}$ is minimal whenever no proper subset of it is in $\tau$
given $\langle \Gamma,C\rangle$.

\begin{example} \label{E8}
We extend example \ref{E6}, in order to show that $C$ is constraint
system C, with C =$\langle KB^{BG},KB^{Ab},IC\rangle$
\begin{itemize}
  \item [] $KB^{BG}=\{p\leftarrow q,a;~p\leftarrow r,b;~q\leftarrow c,d;~r\leftarrow
  e,f;~p\leftarrow b;~a;~r$\}
  \item [] $KB^{Ab}=\{a,b,c,d,e,f,q,r\}$
  \item [] IC\hspace{0.6cm}=\{$\leftarrow b$\}
  \end{itemize}
\end{example}

\begin{note} \label{N2}
Let $KB^{Ab}\in(\{\Delta^{+},\Delta^{-}\})$. Here $\Delta^{+}$
refers to admission Horn knowledge base (positive atoms) and
$\Delta^{-}$ refers to denial Horn knowledge base (negative atoms)
with respect to given $\alpha$. The abduction problem is to explain $\Delta$
with abducibles $(KB^{Ab})$, s.t.
$KB^{BG}\cup\Delta^{+}\cup\Delta^{-}\models\alpha$ and
$KB^{BG}\cup\Delta^{+}\models\alpha\cup\Delta^{-}$ are both
consistent with IC.
\end{note}

An admission and denial Horn knowledge base, based on $\langle
KB^{BG},KB^{Ab} \rangle$ is a set $KB^{Ab}$ of atoms (literals)
whose propositions are in $KB^{Ab}$ such that $KB^{BG}\cup KB^{Ab}$
is consistent with IC.

\begin{example} \label{E9}
From Example \ref{E8} and Note \ref{N2}, the constraint system C,
with C =$\langle KB^{BG},KB^{Ab},IC\rangle$
\begin{itemize}
 \item [] $KB^{BG}=\{p\leftarrow q,a;~p\leftarrow r,b;~q\leftarrow c,d;~r\leftarrow
  e,f;~p\leftarrow b;~a;~r\}$
  \item [] $KB^{Ab}= \{\Delta^+ = \{a,c,d,e,f,q,r\}$ and $ \Delta^- = \{a,b,r\}\}$
  \item [] IC\hspace{0.6cm}=\{$\leftarrow b$\}
\end{itemize}
\end{example}


\begin{definition}[KB-closed abductive explanations] \label{D55}
For a set of abducibles $(KB^{Ab})$, $\Delta^{+}$ and $\Delta^{-}$
are said to be closed abductive explanations with respect to $KB^{BG}$ iff
$KB^{BG}\cup \Delta^{+} \cup \Delta^{-} \models \alpha$ and
$KB^{BG}\cup\Delta^{+}\models \alpha\cup\Delta^{-}$. $\Delta^{+}$
and $\Delta^{-}$ are said to be \emph{minimal} with respect to $KB^{BG}$ iff no
proper subset of $\Delta^{+}$ and $\Delta^{-}$ is an abductive
explanation for $\alpha$, i.e.
$\nexists\Delta^{{+}^{'}}\subsetneq\Delta^{+}$ and
$\nexists\Delta^{{-}^{'}}\subsetneq\Delta^{-}$ s.t.
$KB^{BG}\cup\Delta^{{+}^{'}}\cup\Delta^{{-}^{'}}\models \alpha$ and
$KB^{BG}\cup\Delta^{{+}^{'}}\models\alpha\cup\Delta^{{-}^{'}}$ both
consistent with IC.
\end{definition}

KB-closed abductive explanations are also known as KB-closed local
minimal explanations.

\begin{observation} \label{z1}
Let  $KB^{BG'}$ be a smallest subset of $KB^{BG}$ s.t, $\Delta^{+}$
and $\Delta^{-}$ minimal abductive explanations of $\alpha$ with respect to
$KB^{Ab'}$ and $KB^{BG'}$ (for some $\Delta^{+}$ and $\Delta^{-}$).
Then $\Delta^{+}$ and $\Delta^{-}$ are called \emph{locally minimal}
for $\alpha$ with respect to $KB^{Ab'}$ and $KB^{BG}$ and consistent with IC.
\end{observation}

\begin{example} \label{E10} $\Delta^{+} =\{a,c,d\}$ and $\Delta^{-} =\{c\}$ with respect
to IC are only locally minimal abductive explanations for $p$ with respect to
$KB^{BG'}$ (more explanations can be found in (Lu W 1999)).

From example \ref{E9} and definition \ref{D55},  we want to show
that the constraint system C, with C =$\langle
KB^{BG},KB^{Ab},IC\rangle$
\begin{itemize}
  \item [] $KB^{BG}=\{p\leftarrow q,a;~q\leftarrow c,d;~a;~c;~d\}$
  \item [] $KB^{Ab}= \{\Delta^+ = \{a,c,d\}$ and $ \Delta^- =
  \{b,r\}\}$ and IC\hspace{0.6cm}=\{$\leftarrow b$\}
\end{itemize}

\end{example}

Now, we need to connect the grammar system $\Gamma$ to the Horn
(stratified) knowledge base $KB$, such that $KB_I\cup KB_U\cup
KB_{IC}=KB^{BG}\cup KB^{Ab}\cup IC$ holds. The connection between
locally minimal abductive explanation for $\alpha$ with respect to $KB_I$ and
$\alpha$-kernel of KB, which is shown by the following lemma
immediately follows from their respective definitions.

\begin{observation} \label{l1}
\hspace{0.5cm}
\begin{enumerate} 
\item[1.] Let KB be a Horn (stratified) knowledge base and $\alpha$ a Horn clause s.t.
$\nvdash\neg \alpha$. Let $\Delta^{+}$ and $\Delta^{-}$ be a
KB-closed locally minimal abductive explanation for $\alpha$ with respect to
$KB_{I}$. Then, there exists an $\alpha$-kernel X of KB s.t. $X\cap
KB_{U}=\Delta^{+} \cup \Delta^{-}$.
\item[2.] Let KB be a Horn (Horn logic
with stratified negation) knowledge base and $\alpha$ a Horn clause
s.t. $\nvdash\neg \alpha$. Let X be a $\alpha$-kernel of KB and
$\Delta^{+} \cup \Delta^{-}=X\cap KB_{U}$. Then, $\Delta^{+}$ and
$\Delta^{-}$ are KB - closed locally minimal abductive explanations
for $\alpha$ with respect to $KB_{I}$.
\end{enumerate}
\end{observation}
\vspace{0.5cm}
\begin{proof}
\hspace{0.5cm}
\begin{enumerate}
\item[1.] \rm The fact that $\nvdash\neg\alpha$ and there exists a KB -
closed locally minimal abductive explanation for $\alpha$ with respect to
$KB_{I}$, it is clear that there exists at least one $\alpha$-
kernel of KB. Suppose $\Delta$ ($\Delta\in\Delta^{+} \cup
\Delta^{-}$) is empty (i.e. $KB_{I}\models\neg\alpha$), then the
required result follows immediately. If not, since $\Delta$ is a
locally minimal abductive explanation, there exists a minimal subset
$KB_{I}'\subseteq KB_{I}$, s.t. $\Delta$ is minimal abductive
explanation of $\alpha$ with respect to $KB_{I}'$. Since, $\Delta$ is KB-closed,
it is not difficult to see that $KB_{I}'\cup \Delta^{+} \cup
\Delta^{-}$ is a $\alpha$ - kernel of KB.
\item[2.] Since X is a $\alpha$ - kernel of KB and $\Delta$ is the
set of all abducibles in X, it follows that $\Delta^{+} \cup
\Delta^{-}$ is a minimal abductive explanation of $\Delta$ with respect to
$X\backslash \Delta^{-} \cup \Delta^{+}$. It is obvious that
$\Delta^{+} \cup \Delta^{-}$ is KB- closed, and so $\Delta$ is a
KB-closed locally minimal abductive explanation for $\alpha$ with respect to
$KB_{I}$.
\end{enumerate}
\end{proof}

\begin{theorem} \label{T9}
Consider a constrained abductive grammar $AG = \langle
\Gamma,C\rangle$ with $\Gamma = \langle N,T,KB,R, S\rangle$ and $C =
\langle KB^{BG},KB^{Ab},IC\rangle$. Construct a abductive grammar
$\Delta(AG) = \langle N, T,KB^{BG},R,S\rangle$ by having, for any
$(\Delta^{+})$ (or) $(\Delta^{-})$ from $KB^{Ab}$, the set of
acceptable results for accommodate $(\alpha,KB^{BG}\in\Delta^{+})$
being of the form $(KB^{Ab}\backslash \Delta^{+})$ where
($\Delta^{+}\in KB^{Ab'}$). $\Delta^{+}$ is a locally minimal set of
atoms (literals) $KB^{BG}\cup\Delta^{+}$ and $KB^{BG}\cup\Delta^{+}
\models \alpha$ is consistent with IC; if $(\Delta^{-})$ exists
procedure is similar, $($like denial $(\Delta^{-})$ being of the
form $(KB^{Ab}\backslash \Delta^{-})$. $\Delta^{-}$ is a locally
minimal set of atoms (literals) $KB^{BG}\cup\Delta^{-}$ and
$KB^{BG}\cup\Delta^{-} \models \alpha$ is consistent with IC$)$,
otherwise accommodate $(\alpha,KB^{BG}\in\Delta^{-})$ is not
possible.
\end{theorem}

\begin{proof} 
From Observation \ref{z1}, Let  $KB^{BG'}$ be a smallest subset of $KB^{BG}$ s.t, $\Delta^{+}$
and $\Delta^{-}$ minimal abductive explanations of $\alpha$ with respect to
$KB^{Ab'}$ and $KB^{BG'}$ (for some $\Delta^{+}$ and $\Delta^{-}$).
Then $\Delta^{+}$ and $\Delta^{-}$ are called \emph{locally minimal}
for $\alpha$ with respect to $KB^{Ab'}$ and $KB^{BG}$ and consistent with IC.

From Observation \ref{l1}, $(\Delta^{-})$ is follow to the kernel of KB and $\Delta$ is the
set of all abducibles in $(\alpha,KB^{BG}\in\Delta^{+})$, it follows that $\Delta^{+} \cup
\Delta^{-}$ is a minimal abductive explanation of $\Delta$ with respect to
$KB\backslash \Delta^{-} \cup \Delta^{+}$. It is obvious that
$\Delta^{+} \cup \Delta^{-}$ is KB- closed, and so $\Delta$ is a
KB-closed locally minimal abductive explanation for $\alpha$ with respect to
$KB_{I}$. 

$(\Delta^{-})$ is not follow from KB -
closed locally minimal abductive explanation for $\alpha$ with respect to
$KB_{I}$, it is clear that there exists at least one $\alpha$-
kernel of KB. 
\end{proof}

An immediate consequence of the above observation \ref{l1} is that it is
enough to compute all the KB-closed locally minimal abductive
explanations for $\alpha$ with respect to $KB_{I}$ in order to revise $\alpha$
from KB. Thus, a well-known abductive procedure to compute an
abductive explanation for $\alpha$ with respect to $KB_I$ could be used:

\begin{theorem} \label{T10}
 Let KB be a Horn (stratified) knowledge base and $\alpha$ a Horn clause.
\begin{enumerate}
\item[1.] If Algorithm 1 produces $KB'$ as a result of revision
$\alpha$ to KB, then $KB'$ is a generalized revision of $\alpha$
from KB.
\item[2.] If $KB'$ is a generalized revision of $\alpha$ from KB,
then there exists an incision function $\sigma$ s.t. $KB'$ is
produced by Algorithm 1 as a result of revision $\alpha$
from KB, using $\sigma$.
\end{enumerate}
\end{theorem}

\begin{proof}
Follows from Observation \ref{l1} and Theorem \ref{T8} \end{proof}


\subsection{Generalized revision algorithm}

The problem of Horn knowledge base revision is concerned with
determining how a request to change can be appropriately translated
into one or more atoms or literals.  In this section we develop a
new generalized revision algorithm. Note that it is enough to
compute all the KB-locally minimal abduction explanations for
$\alpha$ with respect to $KB_I \cup KB_U \cup KB_{IC}$. If $\alpha$ is
consistent with KB then a well-known abductive procedure for compute
an abductive explanation for $\alpha$ with respect to $KB_{I}$ could
be used to compute kernel revision.\\

$$\begin{array}{cc}\hline \text{\bf Algorithm 1} & \hspace{-4cm}
\text{\rm Generalized revision algorithm}\\\hline \text{\rm Input}:&
\hspace{-0.6cm}\text{\rm A Horn knowledge base}~ KB=KB_{I}\cup KB_{U}\cup KB_{IC}\\
&\text{\rm and a Horn clause}~
\alpha~ \text{\rm to be revised.}\\
\text{\rm Output:} & \text{\rm A new Horn knowledge base
}~KB'=KB_{I}\cup
KB_{U}^*\cup KB_{IC},\\
&\text{s.t.}~ KB'\text{\rm is
a generalized revision}~ \alpha~\text{\rm to KB.}\\
\text{\rm Procedure}~KB(KB,\alpha)&\\
\text{\rm begin}&\\
~~1.&\hspace{-0.5cm}\text{\rm Let V:=}~\{c\in KB_{IC}~|~ KB_I\cup
KB_{IC}~\text{\rm inconsistent
with}~\alpha~\text{\rm with respect to}~c\}\\
&P:=N:=\emptyset~\text{\rm and}~KB'=KB\\
~~2.&\text{\rm While}~(V\neq \emptyset)\\
&\text{\rm select a subset}~V'\subseteq V\\
&\text{\rm For each}~v\in~V',~\text{\rm select a literal to be}\\
&\hspace{-0.1cm}\text{\rm remove (add to N) or a literal to be added (add to P) with respect to KB}\\
&\text{\rm Let KB}~:=KR(KB,P,N)\\
&\hspace{-0.3cm}\text{\rm Let V:=}~\{c\in KB_{IC}~|~ KB_I~\text{\rm
inconsistent
with}~\alpha~\text{\rm with respect to}~c\}\\
&\hspace{-0.7cm}\text{\rm return}\\
~~3.&\text{\rm Produce a new Horn knowledge base}~KB'\\
\text{\rm end.}&\\ \hline
\end{array}$$

$$\begin{array}{cc}\hline
\text{\rm Procedure}~
KR(KB,\Delta^{+},\Delta^{-})&\\
\text{\rm begin}&\\
1.&\hspace{-1.6cm}\text{\rm Let}~ P :=\{ e \in \Delta^{+} |~
KB_I\not\models e\} ~\text{\rm and}~ N :=\{ e \in \Delta^{-}
 |~KB_I\models e\}\\
2.&\text{\rm While}~(P\neq 0)~\text{\rm or}~(N\neq 0)\\
&\text{\rm select a subset}~P'\subseteq P~ or ~N'\subseteq N \\
&\hspace{-1.7cm}\text{\rm Construct a set}~S_1=\{X~|~X~\text{\rm is
a KB-closed
locally}\\
&\text{\rm minimal abductive
wrt P explanation for}~\alpha~\text{\rm wrt}~KB_{I}\}.\\
&\hspace{-1.7cm}\text{\rm Construct a set}~S_2=\{X~|~X~\text{\rm is
a KB-closed
locally}\\
&\text{\rm  minimal abductive wrt N explanation for}~\alpha~\text{\rm wrt}~KB_{I}\}.\\
3.&\text{\rm Determine a hitting set}~\sigma (S_1) \text{\rm ~and}~\sigma (S_2)\\
&\hspace{-5.5cm}
\text{\rm If}~((N=0)~and~(P\neq0))\\
&\hspace{-1cm}\text{\rm Produce}~KB'=KB_{I}\cup \{(KB_{U} \cup \sigma (S_1)\}\\
&\hspace{-8.8cm}
\text{\rm else}\\
&\text{\rm Produce}~KB'=KB_{I}\cup \{(KB_{U}\backslash
\sigma(S_2) \cup \sigma (S_1)\}\\
&\hspace{-8.5cm}
\text{\rm end if}\\
&\hspace{-5.5cm}
\text{\rm If}~((N\neq0)~\text{\rm and}~(P=0))\\
&\hspace{-1.2cm}\text{\rm Produce}~KB'=KB_{I}\cup
\{(KB_{U}\backslash
\sigma(S_2)\}\\
&\hspace{-8.8cm}
\text{\rm else}\\
&\text{\rm Produce}~KB'=KB_{I}\cup \{(KB_{U}\backslash
\sigma(S_2) \cup \sigma (S_1)\}\\
&\hspace{-8.5cm}
\text{\rm end if}\\
4.& \hspace{-1.5cm}
\text{\rm return}~ KB'\\
\text{\rm end.}&\\ \hline
\end{array}$$

\subsubsection{Reasoning about Abduction}\hspace{0.5cm}

\begin{definition}[\text{(Teniente \& Olive 1995)}] \label{D23}
Let KB=($KB_I,KB_U,KB_{IC}$) be a Horn knowledge base, $T$ is
updatable part from KB. We define the abduction framework $\langle
KB^{BG},KB^{Ab},IC\rangle$. After Algorithm 1 is executed,
$u$ is derived part from $KB'$. The abduction explanation for $u$ in
$\langle KB_I\cup KB_U^*, KB_{IC}\rangle$ is any set $T_i$, where
$T_i \subseteq KB^{Ab}$ such that: $KB_I\cup KB_U^* \cup T \models
u$.

An explanation $T_i$ is minimal if no proper subset of $T_i$ is also
an explanation, i.e. if it does not exist any explanation $T_j$ for
$u$ such that $T_j \subset T_i$
\end{definition}

\subsubsection{Reasoning about Deduction}\hspace{0.5cm}

\begin{definition}[\text{(Teniente \& Olive 1995)}] \label{D24}
Let KB=($KB_I,KB_U,KB_{IC}$) be a Horn knowledge base, $T$ is
updatable part from KB. After Algorithm 1 is executed, $u$ is
derived part from $KB'$. The deduction consequence on $u$ due to the
application of $T$, $KB_I\cup KB_U^* \cup T\cup u$ is the answer to
any question.
\end{definition}

\begin{example}
Consider a Horn knowledge base KB with immutable part $KB_I$,
updatable part $KB_U$ and integrity constraint $KB_{IC}$, compute
closed local minimum with respect to to p.

$$\begin{array}{cccccc} KB_I:&p\leftarrow
q\wedge a&\hspace{0.5cm}KB_U:&a\leftarrow&\hspace{1.2cm}KB_{IC}:&\leftarrow b\\
&p\leftarrow r\wedge b&&r\leftarrow&&\\
&q\leftarrow c\wedge d&&&&\\
&r\leftarrow e\wedge f&&&&\\
&\hspace{-0.6cm}p\leftarrow b&&&&\end{array}$$

From algorithm 1, the above example execute following
steps:

\begin{enumerate}
\item[] \begin{center} \textbf{Step number with execution} \end{center}
\hspace{1cm}
\item[(Input)] $KB_I: p\leftarrow q\wedge a, p\leftarrow r\wedge b, q\leftarrow c\wedge d, r\leftarrow e\wedge f, p\leftarrow
b\\ KB_U: a\leftarrow, r\leftarrow\\ KB_{IC}:
\leftarrow b$
\item[(0)] $\{p\leftarrow q,a;~p\leftarrow r,b;~q\leftarrow c,d;~r\leftarrow
  e,f;~p\leftarrow b;~a;~r$\}
\item[(1)] $\{V = b\}$
\item[(2)] $\{P = \{a,c,d,e,f,q,r\}$ and $ N=
\{a, r\}\}$
\item[(2.1)] $\{\Delta^+ = \{a,c,d,e,f,q,r\}$ and $ \Delta^- =
\{a, r\}\}$
\item[(2.2)] $\{\Delta^+ = \{a,c,d\}$ and $ \Delta^- =
  \{ \}\}$
\item[(3)] $\{p\leftarrow q,a;~q\leftarrow c,d;~a;~c;~d;~r\}$
\item[(Output)] $KB_I: p\leftarrow q\wedge a, p\leftarrow r\wedge b, q\leftarrow c\wedge d, r\leftarrow e\wedge f, p\leftarrow a, p\leftarrow
b\\
KB_{U}^*: a\leftarrow, c\leftarrow, d\leftarrow, r\leftarrow\\
KB_{IC}: \leftarrow b$
\end{enumerate}
\end{example}

\begin{theorem} \label{T81} Let KB be a Horn knowledge base and $\alpha$
is (Horn or Horn logic with stratified negation) formula.
\begin{enumerate} 
  \item If Algorithm 1 produced KB' as a result of revising $\alpha$
  from KB, then KB' satisfies all the rationality postulates (KB*1) to
(KB*6) and (KB*7.3).
  \item Suppose $KB''$ satisfies all these rationality postulates
  for revising $\alpha$ from KB, then $KB''$ can be produced by Algorithm 1.
\end{enumerate}
\end{theorem}

\begin{proof}
Follows from Theorem \ref{T7} and Theorem \ref{T10}
\end{proof}

\section{Belief update VS Knowledge base update} In this section
we give overview of how belief update is related to knowledge base
update. This section is motivated by the works of Konieczny
2011 and Baral $\&$ Zhang 2005.

\subsection{Belief revision vs Belief update} Intuitively,
revision operators bring a minimal change to the base by selecting
the most plausible models among the models of the new information.
Whereas update operators Konieczny
2011 bring a minimal change to each
possible world (model) of the base in order to take into account the
change described by the new information, whatever the possible
world.

\begin{theorem}[\text{[29]}] \label{T11} If $\circ$ is a revision operator (i.e. it satisfies
(R1)-(R6)), then the update operator $\diamond$ defined by
$\psi\diamond\mu$ = $\bigvee_{w\models\psi} \psi_{\{w\}} \circ \mu$
is an update operator that satisfies (U1)-(U9). \end{theorem}

This theorem states that update can be viewed as a kind of pointwise
revision.

\subsection{Knowledge base revision vs Knowledge base update}
Generalized revision algorithm brings principle of minimal change,
according to new information; how a request to change Horn knowledge
base can be appropriately translated into one or more literals.
Whereas update operators (Baral $\&$ Zhang 2005). bring a minimal change to each
possible world (model) of the base in order to take into account the
change described by the new information.

\begin{theorem}[\text{[3]}] \label{T12} If $*_{\sigma}$  is a revision operator (i.e. it satisfies
(KB*1)-(KB*6) and (KB*7.3) and Theorem \ref{T10} and Lemma
\ref{l1}), then the update operator $\diamond$ defined by
$\psi\diamond\mu$ = $\bigvee_{w\models\psi} \psi_{\{w\}} *_{\sigma}
\mu$ is an update operator that satisfies (U1)-(U9).
\end{theorem}

\subsection{Belief update vs Knowledge base update} Formally
speaking, both updates  are aiming at maintaining the base of the
knowledge or belief up-to-date.

\begin{theorem} \label{T13} If $\circ$ are revision operators (i.e.
they satisfy (R1)-(R6)), then the update operator $\diamond$ defined
by $\psi\diamond\mu$ = $\bigvee_{w\models\psi} \psi_{\{w\}} \circ
\mu$ is an update operator that satisfies (U1)-(U9).
\end{theorem}

\begin{figure}
\begin{center}
   \includegraphics[height=3cm,angle=0]{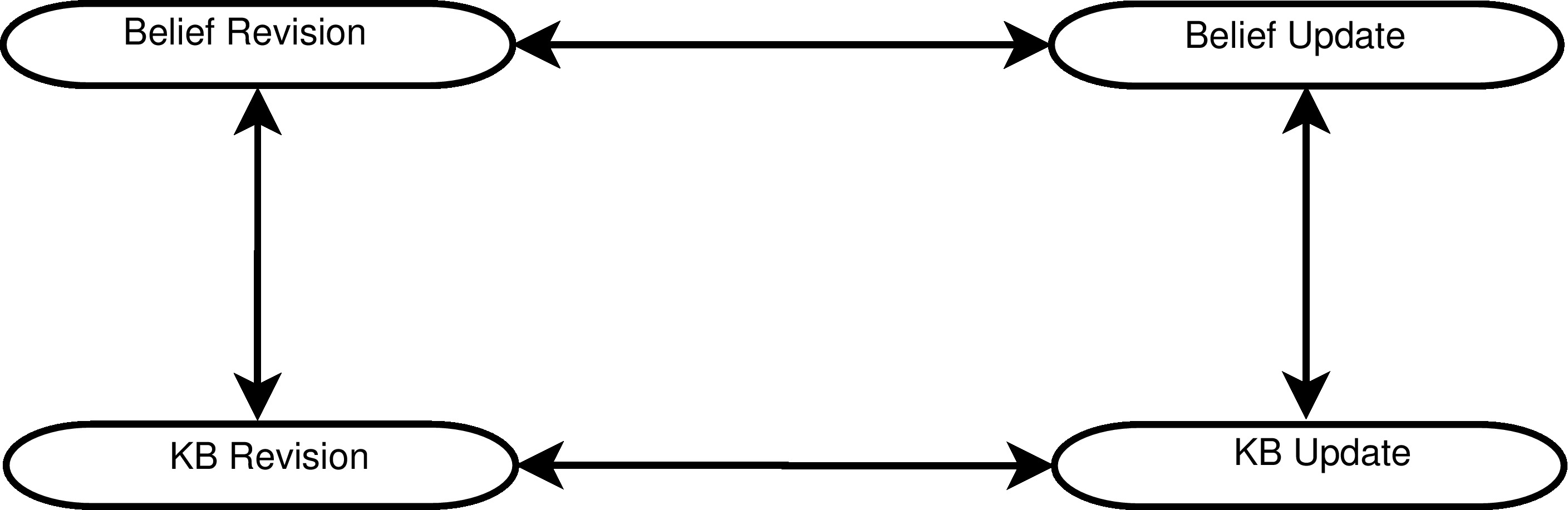}

   \caption{Belief Update Vs Knowledge base Update}
   \end{center}
\end{figure}


\section{Deductive database} A \it Deductive database \rm $DDB$ consists of three parts:
an \it intensional database \rm $IDB$ ($KB_I$), a set of definite
program clauses, \it extensional database \rm $EDB$ ($KB_U$), a set
of ground facts; and \it integrity constraints\rm ~$IC$. The
intuitive meaning of $DDB$ is provided by the \it Least Herbrand
model semantics \rm and all the inferences are carried out through
\rm SLD-derivation. All the predicates that are defined in $IDB$ are
referred to as \emph{view predicates}~and those defined in $EDB$ are
referred to as \emph{base predicates}. \rm Extending this notion, an
atom with a view predicate is said to be a \it view atom,\rm~and
similarly an atom with base predicate is a \it base atom. \rm
Further we assume that $IDB$ does not contain any unit clauses and
no predicate defined in a given $DDB$ is both view and base.

Two kinds of view updates can be carried out on a $DDB$: An atom,
that does not currently follow from $DDB$, can be \it inserted, \rm
or an atom, that currently follows from $DDB$ can be \it deleted.
 \rm When an atom $A$ is to be updated, the view update problem is to
insert or delete only some relevant $EDB$ facts, so that the
modified $EDB$ together with $IDB$ will satisfy the update of $A$ to
$DDB$.

Note that a $DDB$ can be considered as a Horn knowledge base to be
revised. The $IDB$ is the immutable part of the Horn knowledge base
dynamics, while the $EDB$ forms the updatable part. In general, it
is assumed that the language underlying a $DDB$ is fixed and the
semantics of $DDB$ is the least Herbrand model over this fixed
language. We assume that there are no function symbols implying that
the Herbrand Base is finite. Therefore, the $IDB$ is practically a
shorthand of its ground instantiation \footnotemark \footnotetext{a
ground instantiation of a definite program $P$ is the set of clauses
obtained by substituting terms in the Herbrand Universe for
variables in $P$ in all possible ways} written as $IDB_G$. In the
sequel, technically we mean $IDB_G$ when we refer simply to $IDB$.
Thus, a $DDB$ represents a Horn knowledge base dynamics  where the
immutable part is given by $IDB_G$ and updatable part is the $EDB$.
Hence, the rationality postulates (KB*1)-(KB*6) and (KB*7.3) provide
an axiomatic characterization for updating (insert and delete) a
view atom $A$ into a definite database $DDB$.


But before discussing the rationality postulates and algorithm, we
want to make it precise, how a relational database, along with
operations on relations, can be represented by definite deductive
database. We assume the reader is familiar with relational database
concepts. \it A relation scheme \rm $R$ can be thought of as a base
predicate whose arguments define the \it attributes $\mathbb A$ of
the scheme. Its \it relational extension \rm $r$, is a finite set of
base atoms $R(\mathbb A)$ containing the predicate $R$. A \it
database schema \rm consists of finite collection of relational
schemes $<R_1,\ldots,R_n>$, and a \it relational database \rm is a
specific extension of database schema, denoted as
$<r_1,\ldots,r_n>$. In our context, relational database can be
represented by $EDB=\bigcup_{i=1,\ldots,n}R_i(\mathbb{A}_i)$.

\it Join \rm is a binary operator for combining two relations. Let
$r$ and $s$ be two relational extensions of schema $R$ (with
attributes $\mathbb R$) and $S$ (with attributes $\mathbb S$),
respectively. Let $\mathbb T=\mathbb R\cup\mathbb S$. The join of
$r$ and $s$, written as $r\otimes s$, is the relational extension
$q(\mathbb T)$ of all tuples $t$ over $\mathbb T$ such that there
are $t_r\in r$ and $t_s\in s$, with $t_r=t(\mathbb R)$ and
$t_s=t(\mathbb S)$. Join can be captured by a constraint clause
$Q(\mathbb T)\leftarrow R(\mathbb R), S(\mathbb S)$.

Let us consider two relational schemes $R$ and $S$ from Example
\ref{e10}, with attributes $R$=$\{$Group,Chair$\}$ and
$S$=$\{$Staff,Group$\}$.Consider the following extensions $r$ and
$s$: (see definition and properties of similarity in works of
Christiansen (Christiansen \& Rekouts 2007) and Godfrey (Godfrey et
al. 1998)).

\begin{example} \label{E11}
\begin{table}[h]
\begin{center}
$\begin{array}{c|cc} \text{\rm s}&\text{Staff}&\text{Group}\\\hline
&\text{\rm delhibabu}&\text{\rm infor1}\\
&\text{\rm aravindan}&\text{\rm infor2}\\
\end{array}$
$~~~~~~~\begin{array}{c|cc} \text{\rm
r}&\text{Group}&\text{Chair}\\\hline
&\text{\rm infor1}&\text{\rm matthias}\\
&\text{\rm infor2}&\text{\rm gerhard}\\
\end{array}$
\end{center}

\begin{center} \caption{\rm Base table for $s$ and $r$} \end{center} \end{table}
\end{example}

\noindent The following rule, T(Staff,Group,Chair) $\leftarrow$
S(Staff,Group),R(Group,Chair) represents the join of $s$ and $r$,
which is given in Table 5.2:

\begin{table}[h]
$$\begin{array}{c|ccc}
s\otimes r&Staff&Group&Chair\\\hline
&\text{\rm delhibabu}&\text{\rm infor1}&\text{\rm matthias}\\
&\text{\rm aravindan}&\text{\rm infor2}&\text{\rm gerhard}\\
\end{array}$$

\centering \caption{\rm $s\otimes r$ }
\end{table}

Our first integrity constraint (IC) is that each research group has
only one chair ie. $\forall x,y,z$ (y=z) $\leftarrow$
group\_chair(x,y) $\wedge$ group\_chair(x,z). Second integrity
constraint is that a person can be a chair for only one research
group ie. $\forall x,y,z$ (y=z)$\leftarrow$ group\_chair(y,x)
$\wedge$ group\_chair(z,x).

An update request U = $A$, where $A$ is a set of base facts that are
not true in KB. Then, we need to find a transaction $T=T_{ins} \cup
T_{del}$, where $T_{ins} (\Delta_i)$ (resp. $T_{del}(\Delta_j)$) is
the set of facts, such that U is true in $DDB'=((EDB - T_{del} \cup
T_{ins}) \cup IDB \cup IC)$. Since we consider stratifiable
(definite) deductive databases, SLD-trees can be used to compute the
required abductive explanations. The idea is to get all EDB facts
used in a SLD-derivation of $A$ with respect to DDB, and construct that as an
abductive explanation for $A$ with respect to $IDB_G$.

Traditional methods translate a view update request into a
\textbf{transaction combining insertions and deletions of base
relations} for satisfying the request (Mota-Herranz et al. 2000).
Furthermore, a stratifiable (definite) deductive database can be
considered as a knowledge base, and thus the rationality postulates
and insertion algorithm from the previous section can be applied for
solving view update requests in deductive databases.

There are two ways to find minimal elements (insertion and deletion)
in the presence of integrity constraints. Algorithm 2
first checks consistency with integrity constraints and then reduces
steps with abductive explanation for $A$ . Algorithm 3 is
doing \emph{vice versa}, but both algorithm outputs are similar.

$$\begin{array}{cc}\hline
\text{\bf Algorithm 2} & \text{\rm Algorithm to compute all
DDB-closed locally minimal}\\  &\text{\rm abductive explanation of
an atom(literals)}\\\hline \text{\rm Input}:& \text{\rm A definite
deductive database}~DDB=IDB\cup EDB\cup IC~\text{\rm
an literals}\\
&\mathcal{A}\\
\text{\rm Output}:&\text{\rm Set of all DDB-closed locally minimal
abductive explanations}\\
&\text{\rm for}~\mathcal{A}~\text{\rm wrt}~IDB_{G}\\
\text{\rm begin}&\\
~~1.&\text{\rm Let}~ V :=\{ c\in IC~|~IDB\cup IC~\text{\rm
inconsistent
with}~\mathcal{A}~\text{\rm wrt}~c~\}\\
&\text{\rm While}~(V\neq 0)\\
&\hspace{-1.9cm}\text{\rm Construct a complete SLD-tree for} \leftarrow\mathcal{A}~\text{\rm wrt DDB.}\\
&\hspace{-0.8cm}\text{\rm For every successful branch $i$: construct}~\Delta_{i}=\{D~|~D \in EDB\\
&\text{\rm and D is used as an input clause in branch $i$}\}\\
&\hspace{-0.3cm}\text{\rm For every unsuccessful branch $j$: construct}~\Delta_{j}=\{D~|~D \in EDB\\
&\text{\rm and D is used as an input clause in branch $j$}\}\\
&\text{\rm Produce set of all}~\Delta_{i}~\text{\rm
and}~\Delta_{j}~\text{\rm computed in
 the previous step}\\
 &\text{\rm as the result.}\\
&\hspace{-0.7cm}\text{\rm return}\\
~~2.&\text{\rm Produce all DDB-closed locally minimal abductive}\\
&\text{\rm explanations in}~\Delta_{i}~\text{\rm and}~\Delta_{j}\\
 \text{\rm end.}\\\hline
\end{array}$$

Horn knowledge base revision algorithm 1, may be applied
to compute all DDB-closed locally minimal abductive explanation of
an atom (literals).  
Unfortunately, this algorithm does not work as intended for any
deductive database, and a counter example is produced below. Thus,
general algorithms 2 and 3 produced some unexpected sets in addition
to locally minimal abductive explanations

$$\begin{array}{cc}\hline
\text{\bf Algorithm 3} & \text{\rm Algorithm to compute all
DDB-closed locally minimal}\\  &\text{\rm abductive explanation of
an atom(literals)}\\\hline \text{\rm Input}:& \text{\rm A definite
deductive database}~DDB=IDB\cup EDB\cup IC~\text{\rm
an literals}\\
&\mathcal{A}\\
\text{\rm Output}:&\text{\rm Set of all DDB-closed locally minimal
abductive explanations}\\
&\text{\rm for}~\mathcal{A}~\text{\rm wrt}~IDB_{G}\\
\text{\rm begin}&\\
~~1.&\hspace{-1.9cm}\text{\rm Construct a complete SLD-tree for} \leftarrow\mathcal{A}~\text{\rm wrt DDB.}\\
&\hspace{-0.8cm}\text{\rm For every successful branch $i$: construct}~\Delta_{i}=\{D~|~D \in EDB\\
&\text{\rm and D is used as an input clause in branch $i$}\}\\
&\hspace{-0.3cm}\text{\rm For every unsuccessful branch $j$: construct}~\Delta_{j}=\{D~|~D \in EDB\\
&\text{\rm and D is used as an input clause in branch $j$}\}\\
~~2.&\text{\rm Let}~ V :=\{ c\in IC~|~IDB\cup IC~\text{\rm
inconsistent with}~\mathcal{A}~\text{\rm wrt}~c~\}\\
&\text{\rm While}~(V\neq 0)\\
&\text{\rm Produce set of all}~\Delta_{i}~\text{\rm
and}~\Delta_{j}~\text{\rm is consistent with IC}\\
 &\text{\rm as the result.}\\
&\hspace{-0.7cm}\text{\rm return}\\
&\text{\rm Produce all DDB-closed locally minimal abductive}\\
&\text{\rm explanations in}~\Delta_{i}~\text{\rm and}~\Delta_{j}\\
 \text{\rm end.}\\\hline
\end{array}$$

\begin{example} \label{E12} Consider a stratifiable (definite) deductive database DDB as
follows:

$$\begin{array}{cccccc} IDB:&p\leftarrow
a\wedge e&\hspace{0.5cm}EDB:&e\leftarrow&\hspace{1.2cm}IC:&\leftarrow b\\
&q\leftarrow a\wedge f&&f\leftarrow&&\\
&p\leftarrow b\wedge f&&&&\\
&q\leftarrow b\wedge e&&&&\\
&\hspace{-0.6cm}p\leftarrow q&&&&\\
&\hspace{-0.6cm}q\leftarrow a&&&&\end{array}$$
\end{example}

Suppose we want to insert $p$. First, we need to check consistency
with IC and afterwards, we have to find $\Delta_{i}$ and
$\Delta_{j}$ via tree deduction.


\begin{enumerate}
\item[(Input)] $IDB: p\leftarrow a\wedge e, q\leftarrow a\wedge f, p\leftarrow b\wedge f, q\leftarrow b\wedge e, p\leftarrow q, q\leftarrow a\\
EDB: e\leftarrow, f\leftarrow\\
IC: \leftarrow b$
\item[(0)] $\{p\leftarrow a,e;~q\leftarrow a,f;~p\leftarrow b,f;~q\leftarrow b,e;~p\leftarrow q;~q\leftarrow a;~e;~f$\}
\item[(1)] $\{V = b\}$
\item[(2)] \Tree[ {$\leftarrow a,e$\\$\blacksquare$} [.$\leftarrow q$
{$\leftarrow a,f$\\$\blacksquare$} {$\leftarrow a$\\$\blacksquare$}
{$\leftarrow b,e$\\$\Box$} ].$\leftarrow q$ {$\leftarrow
b,f$\\$\Box$} ].$\leftarrow p$
\item[(3-4)] $\Delta_{i} =\{a,e\}$ and $\Delta_{j} =\{ \}$
\item[(5)] $p\leftarrow a,e;~q\leftarrow a,f;~p\leftarrow q;~q\leftarrow a;~b;~e;~f$
\item[(Output)] $IDB: p\leftarrow a\wedge e, q\leftarrow a\wedge f, p\leftarrow q, q\leftarrow
a\\
EDB': a\leftarrow, e\leftarrow, f\leftarrow\\ IC: \leftarrow b$
\end{enumerate}

From the step, it is easy to conclude which branches are consistent
with respect to IC (indicated in the depicted tree by the symbol
$\blacksquare$). For the next step, we need to find minimal
accommodate (positive literal) and denial literal (negative literal)
with with respect to to $p$. The subgoals of the tree are $\leftarrow a,e$ and
$\leftarrow a,f$, which are minimal tree deductions of only facts.
Clearly, $\Delta_{i} =\{a,e\}$ and $\Delta_{j} =\{f\}$ with respect
to IC, are the only locally minimal abductive explanations for $p$
with respect to $IDB_G$, but these result are not closed-locally minimal
explanations.

%


For processing a given view update request, a set of all
explanations for that atom has to be generated through a complete
$SLD$-tree. The resulting hitting set of these explanations is then
a base update of the $EDB$ satsifying the view update request. We
present a different approach which is also rational. The generation
of a hitting set is carried out through a hyper tableaux calculus
(bottom-up) for implementing the deletion process as well as through
the magic sets approach (top-down) for performing insertions
focussed on the particular goal given.


\subsection{View update method}

View update (Behrend \& Manthey 2008) aims at determining one or
more base relation updates such that all given update requests with
respect to derived relations are satisfied after the base updates
have been successfully applied.

\begin{definition}[View update] Let $DDB = \langle IDB,EDB,IC\rangle$ be a
stratifiable (definite) deductive database $DDB(D)$. A VU request
$\nu_{D}$ is a pair $\langle \nu^+_{D},\nu^-_{D}\rangle$ where
$\nu^+_{D}$ and $\nu^-_{D}$ are sets of ground atoms representing
the facts to be inserted into $D$ or deleted from $D$, resp., such
that $pred(\nu^+_{D}\cup \nu^-_{D}) \subseteq pred(IDB)$,
$\nu^+_{D}\cap \nu^-_{D} = \emptyset$, $\nu^+_{D}\cap
PM_{D}=\emptyset$ and $\nu^-_{D}\subseteq PM_{D}$.\rm
\end{definition}

Note that we consider again true view updates only, i.e. ground
atoms which are presently not derivable for atoms to be inserted, or
are derivable for atoms to be deleted, respectively. A method for
view update determines sets of alternative updates satisfying a
given request. A set of updates leaving the given database
consistent after its execution is called \emph{VU realization}.

\begin{definition} [Induced update] Let $DDB = \langle IDB,EDB,
IC\rangle$ be a stratifiable (definite) deductive database and
$DDB=\nu_{D}$ a VU request. A VU realization is a base update
$u_{D}$ which leads to an induced update $u_{D\rightarrow D'}$ from
$D$ to $D'$ such that $\nu^+_{D}\subseteq PM_{D'}$ and
$\nu^-_{D}\cap PM_{D'}=\emptyset$.\rm
\end{definition}

There may be infinitely many realizations and even realizations of
infinite size which satisfy a given VU request. A breadth-first
search (BFS) is employed for determining a set of minimal
realizations $\tau_{D}= \{u^1_{D},\ldots, u^i_{D}\}$. Any $u^i_{D}$
is minimal in the sense that none of its updates can be removed
without losing the property of being a realization for $\nu_{D}$.

\subsubsection{Magic Set (Top-down computation):} \hspace{0.5cm}

Given a VU request $\nu_{DDB}$, view update methods usually determine further VU requests in order
to find relevant base updates. Similar to delta relations for UP we
will use the notion VU relation to access individual view updates
with respect to the relations of our system. For each relation $p\in
pred(IDB\cup EDB)$ we use the VU relation $\nabla^+_ p(\vec{x})$ for
tuples to be inserted into $DDB$ and $\nabla^-_ p(\vec{x})$ for
tuples to be deleted from $DDB$. The initial set of delta facts
resulting from a given VU request is again represented by so-called
\emph{VU seeds}.

\begin{definition}[View update seeds]  Let $DDB(D)$ be a stratifiable (definite) deductive database and
$\nu_{DDB}=\langle \nu^+_{D},\nu^-_{D}\rangle$ a VU request. The set
of VU seeds $vu\_seeds(\nu_D)$ with respect to $\nu_D$ is defined as
follows:
$$vu\_seeds(\nu_D) := \left\{\nabla^{\pi}_p (c_1,\ldots , c_n) | p(c_1,\ldots, c_n)\in \nu^{\pi}_D~and~\pi\in\{+,
-\}\right\} .$$ \rm
\end{definition}

\begin{definition}[View update rules]  Let $IDB$ be a normalized
stratifiable (definite) deductive rule set. The set of VU rules for
true view updates is denoted $IDB^{\nabla}$ and is defined as the
smallest set satisfying the following conditions:
\begin{enumerate}
\item[1.] For each rule of the
form $p(\vec{x})\leftarrow q(\vec{y})\land r(\vec{z})\in IDB$ with
$vars(p(\vec{x})) = (vars(q(\vec{y}))\cup vars(r(\vec{z})))$ the
following three VU rules are in $IDB^{\nabla}$:
$$\begin{array}{ccc} \nabla^+_p(\vec{x})\land \neg
q(\vec{y})\rightarrow \nabla^+_ q (\vec{y})&~~& \nabla^-_ p
(\vec{x})
\rightarrow \nabla^-_ q (\vec{y})\lor \nabla^-_ r (\vec{z})\\
\nabla^+_ p (\vec{x}) \land \neg r(\vec{z}) \rightarrow\nabla^+_ r
(\vec{z})&~~&\end{array}$$
\item[2.] For each rule of the form $p(\vec{x})\leftarrow
q(\vec{x})\land \neg r(\vec{x})\in IDB$ the following three VU rules
are in $IDB^{\nabla}$: $$\begin{array}{ccc} \nabla^+_ p
(\vec{x})\land\neg q(\vec{x}) \rightarrow\nabla^+_ q (\vec{x})&~~&
\nabla^- _p (\vec{x}) \rightarrow \nabla^-_ q (\vec{x}) \lor
\nabla^+_ r (\vec{x})\\ \nabla^+_ p (\vec{x})\land r(\vec{x})
\rightarrow \nabla^-_ r (\vec{x})&~~&\end{array}$$
\item[3.] For each two rules of the form $p(\vec{x})\leftarrow q(\vec{x})$ and
$p(\vec{x})\leftarrow r(\vec{x})$ the following three VU rules are
in $IDB^{\nabla}$: $$\begin{array}{ccc} \nabla^-_ p (\vec{x})\land
q(\vec{x}) \rightarrow\nabla^-_ q (\vec{x})&~~& \nabla^+_ p
(\vec{x}) \rightarrow \nabla^+_ q (\vec{x}) \lor \nabla^+_ r
(\vec{x})\\ \nabla^-_ p (\vec{x}) \land r(\vec{x})
\rightarrow\nabla^-_ r (\vec{x})&~~&\end{array}$$\\
\item[4.]\begin{enumerate} \item[a)] For each relation $p$ defined
by a single rule $p(\vec{x}) \leftarrow  q(\vec{y})\in IDB$ with
$vars(p(\vec{x})) = vars(q(\vec{y}))$ the following two VU rules are
in $IDB^{\nabla}$:
$$\begin{array}{ccc}
 \nabla^+_ p (\vec{x}) \rightarrow \nabla^+_ q (\vec{y})&~~& \nabla^-_ p (\vec{x}) \rightarrow \nabla^-_ q
(\vec{y})\end{array}$$
\item[b)] For each relation $p$ defined by a
single rule $p \leftarrow \neg q \in IDB$ the following two VU rules
are in $IDB^{\nabla}$: $$\begin{array}{ccc}\nabla^+_ p \rightarrow
\nabla^-_ q &~~& \nabla^-_ p \rightarrow \nabla^+_ q
\end{array}$$
\end{enumerate}
\item[5.] Assume without loss of
generality that each projection rule in $IDB$ is of the form
$p(\vec{x})\leftarrow q(\vec{x}, Y )\in IDB$ with $Y \notin
vars(p(\vec{x}))$. Then the following two VU rules
\begin{eqnarray*} &&\nabla^- p_(\vec{x})\land q(\vec{x}, Y )\rightarrow \nabla^-_ q (\vec{x}, Y
)\\
&&\nabla^+_ p (\vec{x})\rightarrow\nabla^+_ q (\vec{x},
c_1)\lor\ldots\lor \nabla^+_ q (\vec{x}, c_n) \lor \nabla^+_ q
(\vec{x}, c^{new})\end{eqnarray*} are in $IDB^{\nabla}$ where all
$c_i$ are constants from the Herbrand universe $\mathcal{U}_{DDB}$
of $DDB$ and $c^{new}$ is a new constant, i.e. $c^{new} \notin
\mathcal{U}_{DDB}$.
\end{enumerate}
\end{definition}

\begin{theorem} \label{T14} Let $DDB = \langle IDB,EDB,IC\rangle$ be a
stratifiable (definite)deductive database(D), $\nu_{D}$ a view
update request and $\tau_{D} = \{u^1_{D},\ldots, u^n_{D}\}$ the
corresponding set of minimal realizations. Let ${D}^{\nabla} =
\langle EDB \cup vu\_seeds(\nu_{D}), IDB \cup IDB^{\nabla}\rangle$
be the transformed deductive database of ${D}$. Then the VU
relations in ${PM_D^{\nabla}}$ with respect to base relations of
${D}$ correctly represent all direct consequences of $\nu_{D}$. That
is, for each realization $u^i_{D} = \langle u^{i^+}_{D} ,
u^{i^{-}}_{D}\rangle \in \tau_{D}$ the following condition holds:
\begin{equation*} \exists p(\vec{t})\in u^{i^+}_{D} : \nabla^+_ p
(\vec{t} )\in {MS_D^{\nabla}} \lor \exists p(\vec{t} )\in
u^{i^{-}}_{D} : \nabla^-_p (\vec{t} ) \in {MS_D^{\nabla}}.
\end{equation*}
\end{theorem}

\begin{proof}
Follows from the result of (Behrend \& Manthey 2008).
\end{proof}

\subsubsection{Hyper Tableau (Bottom-up computation):} \hspace{0.5cm}

In (Aravindan \& Baumgartner
1997) a variant of clausal normal form tableaux called "hyper
tableaux" is introduced for view deletion method. Since the hyper
tableaux calculus constitutes the basis for our view update
algorithm, \it Clauses, \rm i.e. multisets of literals, are usually
written as the disjunction $A_1\lor A_2\lor\cdots\lor
A_m\lor~\text{not}~B_1\lor
~\text{not}~B_2\\\cdots\lor~\text{not}~B_n$ ($M\geq 0,n\geq 0$). The
literals $A_1,A_2,\ldots A_m$ (resp. $B_1,B_2,\ldots, B_n$) are
called the \it head (\rm resp. \it body) \rm of a clause. With
$\overline{L}$ we denote the complement of a literal $L$. Two
literals $L$ and $K$ are complementary if $\overline{L}=K$.

From now on $D$ always denotes a finite ground clause set, also
called \it database, \rm and $\Sigma$ denotes its signature, i.e.
the set of all predicate symbols occurring in it. We consider finite
ordered trees $T$ where the nodes, except the root node, are labeled
with literals. In the following we will represent a branch $b$ in
$T$ by the sequence $b=L_1,L_2,\ldots, L_n$ ($n\geq 0$) of its
literal labels, where $L_1$ labels an immediate successor of the
root node, and $L_n$ labels the leaf of $b$. The branch $b$ is
called \it regular \rm iff $L_i\neq L_j$ for $1\leq i,j\leq n$ and
$i\neq j$, otherwise it is called \it irregular. \rm The tree $T$ is
\it regular\rm~iff every of its branches is regular, otherwise it is
\it irregular. \rm The set of \it branch literals \rm of $b$ is
$lit(b)= \{L_1,L_2,\ldots,L_n\}$. For brevity, we will write
expressions like $A \in b$ instead of $A\in lit(b)$.  In order to
memorize the fact that a branch contains a contradiction, we allow
to label a branch as either \it open \rm or \it closed. \rm A
tableau is closed if each of its branches is closed, otherwise it is
open.

\begin{definition} [Hyper Tableau]
A literal set is called \it inconsistent \rm iff it contains a pair
of complementary literals, otherwise it is called \it consistent.
Hyper tableaux \rm for $D$ are inductively defined as follows: \bf

Initialization step: \rm  The empty tree, consisting of the root
node only, is a hyper tableau for $D$. Its single branch is marked
as "open". \bf

Hyper extension step: \rm  If (1) $T$ is an open hyper tableau for
$D$ with open branch $b$,  and (2) $C=A_1\lor A_2\lor\cdots\lor
A_m\leftarrow B_1\land B_2\cdots\land B_n$ is  a clause from $D$
($n\geq 0,m\geq 0$), called \it extending clause \rm in this
context, and (3) $\{B_1,B_2,\ldots, B_n\}\subseteq b$ (equivalently,
we say that $C$ is \it applicable to $b$)\rm  then the tree $T$ is a
hyper tableau for $D$, where $T$ is obtained from $T$ by extension
of $b$ by $C$:
 replace $b$ in $T$ by the \it new branches \rm \begin{equation*}
 (b,A_1),(b,A_2),\ldots,(b,A_m),(b,\neg B_1),(b,\neg B_2),\ldots,
 (b,\neg B_n)
 \end{equation*}

and then mark every inconsistent new branch as "closed", and the
other new branches as "open".
\end{definition}

The applicability condition of an extension expresses that all body
literals have to be satisfied by the branch to be extended. From now
on, we consider only regular hyper tableaux. This restriction
guarantees that for finite clause sets no branch can be extended
infinitely often. Hence, in particular, no open finished branch can
be extended any further. This fact will be made use of below
occasionally. Notice as an immediate consequence of the above
definition that open branches never contain negative literals.

%
%
%
%

This paper work focused on stratified (definite) deductive database
without any auxiliary variable. In magic set rule play in minimal
case, our future goal is similar foundation using auxiliary variable
(Deductive Databases) side and more details found in
Behrend's (Behrend \& Manthey 2008) work.

\subsection{View update algorithm}

The key idea of the algorithm presented in this paper is to
transform the given database along with the view update request into
a logic program and apply known minimality techniques
to solve the original view update problem. The intuition behind the
transformation is to obtain a logic program in such a
way that each (minimal) model of this transformed program represent
a way to update the given view atom. We present two variants of our
algorithm. The one that is discussed in this section employs a
trivial transformation procedure but has to look for minimal models;
and another performs a costly transformation, but dispenses with the
requirement of computing the minimal models.

\subsubsection{Minimality test}\hspace{0.5cm}

We start presenting an algorithm for stratifiable (definite)
deductive databases by first defining precisely how the given
database is transformed into a logic program for the
view deletion process (Aravindan \& Baumgartner 1997)

\begin{definition} [$IDB$ Transformation] Given an $IDB$ and a set of ground atoms $S$, the
transformation of $IDB$ with respect to $S$ is obtained by translating each
clause $C\in IDB$ as follows: Every atom $A$ in the body (resp.
head) of $C$ that is also in $S$ is moved to the head (resp. body)
as $\neg A$.
\end{definition}

\begin{note} If $IDB$ is a stratifiable deductive database then the
transformation introduced above is not necessary.
\end{note}

\begin{definition} [$IDB^*$ Transformation] Let $IDB\cup EDB$ be a given database.
Let $S_0=EDB\cup \{A~|~A~\text{ is a ground \it IDB \rm atom}\}$.
Then, $IDB^*$ is defined as the transformation of $IDB$ with respect to $S_0$.
\end{definition}

\begin{note} $IDB^*$ is in general a logic
program. The negative literals $(\neg A)$ appearing in the clauses
are intuitively interpreted as deletion of the corresponding atom
($A$) from the database. Technically, a literal $\neg A$ is to be
read as a \it positive \rm atom, by taking the $\neg$-sign as part
of the predicate symbol. To be more precise, we treat $\neg A$ as an
atom with respect to $IDB^*$, but as a negative literal with respect to $IDB$. \end{note}

Note that there are no facts in $IDB^*$. So when we add a delete
request such as $\neg A$ to this, the added request is the only fact
and any bottom-up reasoning strategy is fully focused on the goal
(here the delete request)

\begin{definition}[Update Tableaux Hitting Set] An update tableau for a database
$IDB\cup EDB$ and delete request $\neg A$ is a hyper tableau $T$ for
$IDB^*\cup\{\neg A\leftarrow\}$ such that every open branch is
finished. For every open finished branch $b$ in $T$ we define the
\it hitting set (of b in $T$) \rm as $HS(b)=\{A \in EDB | \neg A \in
b\}$.
\end{definition}

The next step is to consider the view insertion process (Behrend \&
Manthey 2008):

\begin{definition} [$IDB^{\bullet}$ Transformation] Let $IDB\cup EDB$ be a given database.
Let $S_1=EDB\cup \{A~|~A~\text{ is a ground \it IDB \rm atom}\}$ (that is either body or head  empty).
Then, $IDB^{\bullet}$ is defined as the transformation of $IDB$ with respect to
$S_1$.
\end{definition}

\begin{note} $IDB$ is in general a (stratifiable)
logic program. The positive literals $(A)$ appearing in
the clauses are intuitively interpreted as an insertion of the
corresponding atom ($A$) from the database.
\end{note}

\begin{definition}[Update magic Hitting Set] An update magic set rule for a database
$IDB\cup EDB$ and insertion request $A$ is a magic set rule $M$ for
$IDB^{\bullet}\cup\{A\leftarrow\}$ such that every close branch is
finished. For every close finished branch $b$ in $M$ we define the
\it magic set rule (of b in $M$) \rm as $HS(b)=\{A \in EDB | A \in
b\}$.
\end{definition}

\begin{example} \label{E20}

Given stratifiable (definite) deductive database $DDB=IDB \cup EDB\cup IC$ and insert $p$.

$$\begin{array}{cccccc} IDB:&p \leftarrow
a\wedge e&\hspace{0.5cm}EDB:&a\leftarrow&\hspace{1.2cm}IC:&\emptyset\\
&q \leftarrow a \wedge e&&c\leftarrow&&\\
&p \leftarrow a \wedge f&&&&\\
&q\leftarrow c&&&&\\
\end{array}$$

$IDB^*$ Transformation:

$$\begin{array}{cccccc} IDB^*:&\neg a \lor \neg e \leftarrow
\neg p&\hspace{0.5cm}EDB:&a\leftarrow&\hspace{1.2cm}IC:&\emptyset\\
&\neg a \lor \neg e \leftarrow \neg q&& c\leftarrow&&\\
&\neg a \lor \neg f \leftarrow \neg p&&&&\\
&\neg c \leftarrow \neg q&&&&\\
\end{array}$$

$IDB^{\bullet}$ Transformation: (Body empty)

$$\begin{array}{cccccc} IDB^{\bullet}:&p \lor \neg a \lor \neg e \leftarrow
&\hspace{0.5cm}EDB:&a\leftarrow&\hspace{1.2cm}IC:&\emptyset\\
&q \lor \neg a \lor \neg e \leftarrow&& c\leftarrow&&\\
&p \lor \neg a \lor \neg f \leftarrow&&&&\\
&q \lor \neg c \leftarrow &&&&\\
\end{array}$$

$IDB^{\bullet}$ Transformation: (Head empty)

$$\begin{array}{cccccc} IDB^{\bullet}:&\leftarrow \neg p \wedge a \wedge e
&\hspace{0.5cm}EDB:&a\leftarrow&\hspace{1.2cm}IC:&\emptyset\\
&\leftarrow \neg q \wedge a \wedge e && c\leftarrow&&\\
&\leftarrow \neg p \wedge a \wedge f &&&&\\
&\leftarrow \neg q \wedge c &&&&\\
\end{array}$$

The set $S_{0}$ is determined by all the $IDB$ atoms and the current
$EDB$ atoms and in our case it is $\{ p,q,a,c,e,f \}$. $IDB^*$ and
$IDB^{\bullet}$ is the transformation of $IDB$ with respect to $S_{0}$ which is
given above.
\end{example}

Suppose a ground view atom $\mathcal{A}$ is to be insert. Then, an
update tableau for $IDB^{\bullet}$ with insert request $\mathcal{A}$
($IDB^*$ with delete request $\neg \mathcal{A}$) is built. The
intuition is that the set of $EDB$ atoms appearing in a model
(open/close branch) constitute a hitting set, and removing/adding
this set from EDB should achieve the required view insertion.
Unfortunately, this does not result in a rational insertion, as
relevance policy may be violated.

\begin{example}
Let us continue with example \ref{E20} Suppose the view atom $p$ is
to be insert. Then according to the above proposal, an update
tableau for $IDB^{\bullet}$ ( $IDB^*$ ) and $p$ $(\neg p$) is to be
built. This is illustrated in the accompanying figure below. As
shown, open/close branches constitute two hitting sets $\{a\}$ and
$\{f,a\}$ ($\{ \neg a\}$ and $\{ \neg f, \neg a\}$). It is not
difficult to see that $\{f,a\}$($\{ \neg f, \neg a\}$) does not
satisfy any of the relevance policies (KB*7.1) or (KB*7.2) or
(KB*7.3). Hence simple model computation using hyper tableau
calculus does not result in rational hitting sets. The branch is
closed if the corresponding hitting set does not satisfy this strong
relevance postulate.
\end{example}

\begin{figure}[h]
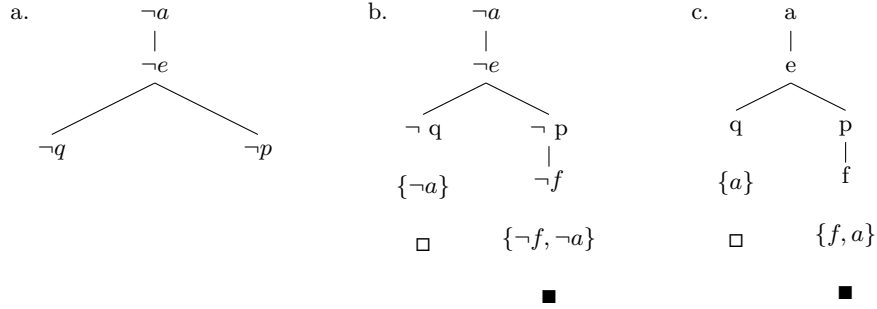

\begin{center}
\qtreecenterfalse a. \Tree [ [ ${\neg q}$ !\qsetw{5cm} ${\neg p}$ ]
.${\neg e}$ ].${\neg a}$ \hskip 0.5in b. \Tree [ [ {$\neg$ q \\\\
$\{ \neg a \}$ \\\\ $\Box$} !\qsetw{2cm} [ {$\neg f$
\\\\ $\{
\neg f, \neg a \}$ \\\\ $\blacksquare$} ] .{$\neg$ p} ] .{$\neg e$}
] .${\neg a}$ \hskip 0.5in c. \Tree [ [ {q \\\\
$\{ a \}$ \\\\ $\Box$} !\qsetw{2cm} [ {f \\\\ $\{ f, a \}$
\\\\ $\blacksquare$} ] .p ] .e ] .a
\caption{$IDB^*$ and $IDB^{\bullet}$ transformation with hitting
set}
   \label{Figure 5.2}
\end{center}
\end{figure}

\begin{definition}[Minimality test] Let $T$ be an update tableau for $IDB^*\cup
EDB$ and delete request $\neg A$. We say that open finished branch
$b$ in $T$ satisfies \it the strong minimality test \rm iff $\forall
s\in HS(b):IDB^*\cup EDB\backslash HS(b)\cup\{s\}\vdash \neg A$.
\end{definition}

\begin{definition}[Update Tableau satisfying strong minimality] An update tableau
for given $IDB\cup  EDB$ and delete request $\neg A$ is transformed
into an update tableau satisfying strong minimality by marking every
open finished branch as closed which does not satisfy strong
minimality.
\end{definition}

The next step is to consider the view insertion process (Behrend \&
Manthey 2008):

\begin{definition}[Minimality test] Let $M$ be an update magic set rule for $IDB\cup
EDB$ and insert request $A$. We say that close finished branch $b$
in $M$ satisfies \it the strong minimality test \rm iff $\forall
s\in HS(b):IDB^{\bullet}\cup EDB\backslash HS(b)\cup\{s\}\vdash A$.
\end{definition}

\begin{definition}[Update magic set rule satisfying strong minimality] An update
magic set rule for given $IDB\cup  EDB$ and insert request $A$ is
transformed into an update magic set rule satisfying strong
minimality by marking every close finished branch as open which does
not satisfy strong minimality.
\end{definition}

\begin{example}
Continuing with the above example, after constructing the branch
corresponding to the hitting set $\{f,a\}$($\{ \neg f, \neg a\}$),
the strong minimality test is carried out as follows: It is checked
if the resulting database with each member of hitting set implies
the insert atom $p$. For example, $IDB\cup EDB\backslash \{ f,a\}
\cup \{a\}\nvdash p$, and hence this branch fails the strong
minimality test.
\end{example}

Interestingly, this minimality test is equivalent to the
groundedness test used for generating minimal models of
logic programs. The key idea of the groundedness test is to check if
the members in the model are implied by the program together with
the negation/positive of the atoms not present in the model. The
groundedness test for generating minimal models can be stated as
follows: Let T be an update tableau for $IDB\cup EDB$ and insert
request $\mathcal{A}$.We say that open finished branch b in T
satisfies the groundedness test iff $\forall s\in
HS(b):IDB^{\bullet}\cup EDB\backslash HS(b)\cup\{\mathcal{A}\}\vdash
s$, similar for $IDB^*$ ($\forall s\in HS(b):IDB^*\cup EDB\backslash
HS(b)\cup\{\mathcal{\neg A}\}\vdash \neg s$). It is not difficult to
see that this is equivalent to the minimality test. This means that
every minimal model (minimal with respect to the base atoms) of $IDB^* \cup
\{\mathcal{A}\}$ provides a minimal hitting set for insertion the
ground view atom $\mathcal{A}$.

$$\begin{array}{cc}\hline
\text{\bf Algorithm 4} &  \text{\rm View
update Algorithm based on minimality test}\\\hline \text{\rm
Input}:&
\text{\rm A stratifiable (definite) deductive database}~DDB=IDB\cup EDB\cup IC\\
&\text{\rm an literals}~\mathcal{A}\\
\text{\rm Output:}&\text{\rm A new stratifiable (definite) database}~IDB\cup EDB'\cup IC\\
\text{\rm begin}&\\
~~1.&\text{\rm Let}~ V :=\{ c\in IC~|~IDB\cup IC~\text{\rm
inconsistent
with}~\mathcal{A}~\text{\rm with respect to}~c~\}\\
&\text{\rm While}~(V\neq \emptyset)\\
~~2.&\hspace{-1.9cm}\text{\rm Construct a complete SLD-tree for} \leftarrow\mathcal{A}~\text{\rm with respect to DDB.}\\
~~3.&\hspace{-0.7cm}\text{\rm For every successful branch $i$:construct}~\Delta_{i}=\{D~|~D\in EDB \}\\
&\text{\rm and D is used as an input clause in branch $i$}.\\
&\text{\rm Construct a branch i of an update tableau satisfying
minimality}\\
&\text{\rm for}~ IDB\cup EDB~\text{\rm and delete request}~\neg A.\\
&\text{\rm  Produce}~IDB\cup EDB \backslash HS(i)~\text{\rm as a result}\\
~~4.&\text{\rm For every unsuccessful branch $j$:construct}~\Delta_{j}=\{D~|~D\in EDB \}\\
&\text{\rm and D is used as an input clause in branch $j$}.\\
&\text{\rm Construct a branch j of an update magic set rule
satisfying minimality}\\
&\text{\rm for}~ IDB\cup EDB~\text{\rm and insert request}~A.\\
&\text{\rm  Produce}~IDB\cup EDB \backslash HS(j)~\text{\rm as a
result}\\
&\text{\rm Let}~ V :=\{ c\in IC~|~IDB\cup IC~\text{\rm inconsistent
with}~\mathcal{A}~\text{\rm with respect to}~c~\}\\
&\hspace{-0.7cm}\text{\rm return}\\
~~5.&\text{\rm Produce}~DDB~\text{\rm as the result.}\\
\text{\rm end.}&\\\hline
\end{array}$$

This means that every minimal model (minimal with respect to the base atoms) of
$IDB^* \cup \{\neg A\}$ ($IDB^{\bullet} \cup \{A\}$ )provides a
minimal hitting set for deleting the ground view atom $A$.
Similarly, $IDB^* \cup \{A\}$ provides a minimal hitting set for
inserting the ground view atom $A$. We are formally present our
algorithms.

Given a database and a view atom to be updated, we first transform
the database into a logic program and use hyper
tableaux calculus to generate models of this transformed program for
deletion of an atom.  Second, magic sets transformed rules are used
is used to generate models of this transformed program for
determining an induced insertion of an atom. Models that do not
represent rational update are filtered out using the strong
minimality test. The procedure for stratifiable (definite) deductive
databases is presented in Algorithms in 4 and 5.

$$\begin{array}{cc}\hline
\text{\bf Algorithm 5} & \text{\rm View
update Algorithm based on minimality test}\\\hline \text{\rm
Input}:&
\text{\rm A stratifiable (definite) deductive database}~DDB=IDB\cup EDB\cup IC\\
&\text{\rm an literals}~\mathcal{A}\\
\text{\rm Output:}&\text{\rm A new stratifiable (definite) database}~IDB\cup EDB'\cup IC\\
\text{\rm begin}&\\
~~1.&\hspace{-1.9cm}\text{\rm Construct a complete SLD-tree for} \leftarrow\mathcal{A}~\text{\rm with respect to DDB.}\\
~~2.&\hspace{-0.7cm}\text{\rm For every successful branch $i$:construct}~\Delta_{i}=\{D~|~D\in EDB \}\\
&\text{\rm and D is used as an input clause in branch $i$}.\\
~~3.&\text{\rm For every unsuccessful branch $j$:construct}~\Delta_{j}=\{D~|~D\in EDB \}\\
&\text{\rm and D is used as an input clause in branch $j$}.\\
~~4.&\text{\rm Let}~ V :=\{ c\in IC~|~IDB\cup IC~\text{\rm
inconsistent
with}~\mathcal{A}~\text{\rm with respect to}~c~\}\\
&\text{\rm While}~(V\neq \emptyset)\\
&\text{\rm Construct a branch i of an update tableau satisfying
minimality}\\
&\text{\rm for}~ IDB\cup EDB~\text{\rm and delete request}~\neg A.\\
&\text{\rm  Produce}~IDB\cup EDB \backslash HS(i)~\text{\rm as a result}\\
&\text{\rm Construct a branch j of an update magic set rule
satisfying minimality}\\
&\text{\rm for}~ IDB\cup EDB~\text{\rm and insert request}~A.\\
&\text{\rm  Produce}~IDB\cup EDB \backslash HS(j)~\text{\rm as a
result}\\
&\text{\rm Let}~ V :=\{ c\in IC~|~IDB\cup IC~\text{\rm inconsistent
with}~\mathcal{A}~\text{\rm with respect to}~c~\}\\
&\hspace{-0.7cm}\text{\rm return}\\
~~5.&\text{\rm Produce}~DDB~\text{\rm as the result.}\\
\text{\rm end.}&\\\hline
\end{array}$$

\begin{lemma} \label{l2}
The strong minimality test and the groundedness test are equivalent.
\end{lemma}

\begin{proof}
Follows from the result of (Aravindan \& Baumgartner 1997).
\end{proof}

\begin{example}
$$\begin{array}{cccccc} IDB:&p \leftarrow
a\wedge e&\hspace{0.5cm}EDB:&a\leftarrow&\hspace{1.2cm}IC:&\leftarrow b\\
&q \leftarrow a \wedge e&&f\leftarrow&&\\
&p\leftarrow b \wedge f&&&&\\
&q\leftarrow b \wedge f&&&&\\
&p\leftarrow g \wedge a&&&&\\
&q\leftarrow p&&&&\\
\end{array}$$
\end{example}

Suppose we want to insert $p$. First, we need to check consistency
with IC and afterwards, we have to find $\Delta_{i}$ and
$\Delta_{j}$ via tree deduction.

From algorithm 4 or 5 (only different is
checking IC condition), the above example execute following steps:

\begin{enumerate}
\item[] \begin{center} \textbf{Step number with execution} \end{center}
\hspace{1cm}
\item[(Input)] $IDB: p \leftarrow
a \wedge e;~ q \leftarrow
a \wedge e;~p \leftarrow
b \wedge f;~q \leftarrow
b \wedge f;~p \leftarrow
g \wedge a;~ q\leftarrow p\\
EDB: a\leftarrow, f\leftarrow\\
IC: \leftarrow b$
\item[(0)] $\{p \leftarrow
a,e;~q \leftarrow
a,e;~p \leftarrow
b,f;~q \leftarrow
b,f;~p \leftarrow
g,a;~q\leftarrow p;~a;~f$\}
\item[(1)] $\{V = b\}$
\item[(2.1)] \qtreecenterfalse a. \Tree [ [ q !\qsetw{5cm} p ] .e ].a \hskip
0.5in b. \Tree [ [ [ [ q !\qsetw{5cm} p ].e ].a ] .b ] .f \hskip
0.5in c. \Tree [ [ [ [ q !\qsetw{7cm} [ q ] .p ] .e ].a ] .b ] .f \\
d. \Tree [ [ [ [ [ q !\qsetw{7cm} [ q ] .p ] .e ].a ] .b ] .f ] .g
\hskip 1.5in e. \Tree [ [ [ [ [ {q\\\\$\Box$} !\qsetw{7cm} [
q\\\\$\blacksquare$ ] .p ] .e ].a ] .b ] .f ] .g
\item[(2.2)] \qtreecenterfalse a. \Tree [ [ q !\qsetw{5cm} p ] .e ].a \hskip
0.5in b. \Tree [ [ q !\qsetw{5cm} p ] .e ].a \hskip
0.5in c. \Tree [ [ q !\qsetw{7cm} [ q ] .p ] .e ].a \\
d. \Tree [ [ [ q !\qsetw{7cm} [ q ] .p ] .e ].a ] .g \hskip 1.5in e.
\Tree [ [ [ {q\\\\$\Box$} !\qsetw{7cm} [ q\\\\$\blacksquare$ ] .p ]
.e ].a ] .g
\item[(3-4)] $\Delta_{i} =\{a,e,g\}$ and $\Delta_{j} =\{ \}$
\item[(5)] $p\leftarrow a,e;~q\leftarrow a,e;~p\leftarrow g,a;~q\leftarrow p;~a;~e;~f;~g$
\item[(Output)] $IDB:p \leftarrow
a \wedge e;~ q \leftarrow
a \wedge e;~p \leftarrow
g \wedge a;~ q\leftarrow p\\
EDB': a\leftarrow, e\leftarrow, f\leftarrow, g\leftarrow\\ IC:
\leftarrow b$
\end{enumerate}

To show the rationality of this approach, we study how this is
related to the previous approach presented in the last section, i.e.
generating explanations and computing hitting sets of these
explanations. To better understand the relationship it is imperative
to study where the explanations are in the hyper tableau approach
and magic set rules. We first define the notion of an $EDB$-cut and
then view update seeds.

\begin{definition}[$EDB$-Cut] Let $T$ be update tableau with open branches
$b_1,b_2,\ldots,b_n$. A set $S=\{A_1,A_2,\ldots,A_n\}\subseteq EDB$
is said to be $EDB$-cut of $T$ iff  $\neg A_i\in b_i$ ($A_i\in
b_i$), for $1\leq i\leq n$.
\end{definition}

\begin{definition}[$EDB$ seeds]
Let $M$ be an update seeds with close branches
$b_1,b_2,\ldots,$$b_n$. A set $S=\{A_1,A_2,\ldots,A_n\}\subseteq
EDB$ is said to be a $EDB$-seeds of $M$ iff EDB seeds
$vu\_seeds(\nu_D)$ with respect to $\nu_D$ is defined as follows:
$$vu\_seeds(\nu_D) := \left\{\nabla^{\pi}_p (c_1,\ldots , c_n) | p(c_1,\ldots, c_n)\in \nu^{\pi}_D~and~\pi\in\{+,
-\}\right\} .$$ \rm
\end{definition}

\begin{lemma} \label{l3} Let $T$ be an update tableau for $IDB\cup EDB$ and
update request $A$. Similarly, for $M$ be an update magic set rule.
Let $S$ be the set of all $EDB$-closed minimal abductive
explanations for $A$ with respect to. $IDB$.  Let $S'$ be the set of all
$EDB$-cuts of $T$ and $EDB$-seeds of $M$ . Then the following hold
\begin{enumerate}
\item[$\bullet$] $S\subseteq S'$.\\
\item[$\bullet$] $\forall \Delta'\in S':\exists \Delta\in S s.t. \Delta\subseteq \Delta'$.
\end{enumerate}
\end{lemma}

\begin{proof}
\hspace{0.5cm}
\begin{enumerate} 
\item[1.] Consider a $\Delta (\Delta\in\Delta_i \cup \Delta_j)\in S$.
We need to show that $\Delta$ is generated by algorithm 4
at step 2. From observation \ref{l1}, it is clear that there exists a
$A$-kernel $X$ of $DDB_G$ s.t. $X \cap EDB = \Delta_j$ and $X \cup
EDB = \Delta_i$. Since $X \vdash A$, there must exist a successful
derivation for $A$ using only the elements of $X$ as input clauses
and similarly $X \nvdash A$. Consequently $\Delta$ must have been
constructed at step 2.
\item[2.] Consider a $\Delta'((\Delta'\in\Delta_i \cup \Delta_j)\in S'$. Let $\Delta'$ be
constructed from a successful (unsuccessful) branch $i$ via
$\Delta_i$($\Delta_j$). Let $X$ be the set of all input clauses used
in the refutation $i$. Clearly $X\vdash A$($X\nvdash A$). Further,
there exists a minimal (with respect to set-inclusion) subset $Y$ of $X$ that
derives $A$ (i.e. no proper subset of $Y$ derives $A$). Let $\Delta
= Y \cap EDB$ ($Y \cup EDB$). Since IDB does not (does) have any
unit clauses, $Y$ must contain some EDB facts, and so $\Delta$ is
not empty (empty) and obviously $\Delta\subseteq \Delta'$. But, $Y$
need not (need) be a $A$-kernel for $IDB_G$ since $Y$ is not ground
in general. But it stands for several $A$-kernels with the same
(different) EDB facts $\Delta$ in them. Thus, from observation \ref{l1},
$\Delta$ is a DDB-closed locally minimal abductive explanation for
$A$ with respect to $IDB_G$ and is contained in $\Delta'$.
minimal.
\end{enumerate}
\end{proof}

The above lemma precisely characterizes what explanations are
generated by an update tableau. It is obvious then that a branch
cuts through all the explanations and constitutes a hitting set for
all the generated explanations. This is formalized below.

\begin{lemma} \label{l4} Let $S$ and $S'$ be sets of sets s.t. $S\subseteq S'$ and
every member  of $S'\backslash S$ contains an element of S. Then, a
set $H$ is a minimal hitting set for $S$ iff it is a minimal hitting
set for $S'$.
\end{lemma}
\vspace{0.5cm}

\begin{proof}
\hspace{0.5cm}
\begin{enumerate}
\item[1.] (\textbf{Only if part})~Suppose $H$ is a minimal hitting set for $S$. Since $S
\subseteq S'$, it follows that $H \subseteq \bigcup S'$. Further,
$H$ hits every element of $S'$, which is evident from the fact that
every element of $S'$ contains an element of $S$. Hence $H$ is a
hitting set for $S'$. By the same arguments, it is not difficult to
see that $H$ is minimal for $S'$ too.\\

(\textbf{If part})~Given that $H$ is a minimal hitting set for $S'$,
we have to show that it is a minimal hitting set for $S$ too. Assume
that there is an element $E \in H$ that is not in $\bigcup S$. This
means that $E$ is selected from some $Y \in S'\backslash S$. But $Y$
contains an element of $S$, say $X$. Since $X$ is also a member of
$S'$, one member of $X$ must appear in $H$. This implies that two
elements have been selected from $Y$ and hence $H$ is not minimal.
This is a contradiction and hence $H \subseteq \bigcup S$. Since $S
\subseteq S'$, it is clear that $H$ hits every element in $S$, and
so $H$ is a hitting set for $S$. It remains to be shown that $H$ is
minimal. Assume the contrary, that a proper subset $H'$ of $H$ is a
hitting set for $S$. Then from the proof of the only if part, it
follows that $H'$ is a hitting set for $S'$ too, and contradicts the
fact that $H$ is a minimal hitting set for $S'$ .
Hence, $H$ must be a minimal hitting set for $S$.
\end{enumerate}
\end{proof}

\begin{lemma} \label{l5} Let $T$ be an update tableau for $IDB\cup EDB$ and
update request $A$ that satisfies the strong minimality test.
Similarly, for $M$ be an updating magic set rule. Then, for every
open (close) finished branch $b$ in $T$, $HS(b)$ ($M$, $HS(b)$) is a
minimal hitting set of all the abductive explanations of $A$.
\end{lemma}

\begin{proof}
Follows from the Observation \ref{l1} (minimal test) in and (Behrend \&
Manthey 2008).
\end{proof}

So, Algorithms 4 and 5 generate a minimal
hitting set (in polynomial space) of all $EDB$-closed locally
minimal abductive explanations of the view atom to be deleted. From
the belief dynamics results recalled in section \ref{s11}, it
immediately follows that Algorithms 4 and 5 are
rational, and satisfy the strong relevance postulate (KB*7.1).

\begin{theorem} \label{T15} Algorithms 4 and 5 are rational, in the sense that
they satisfy all the rationality postulates (KB*1)-(KB*6) and the
strong relevance postulate (KB*7.1). Further, any update that
satisfies these postulates can be computed by these algorithms.
\end{theorem}

\begin{proof}
Follows from Observation \ref{l1},\ref{l5} and Theorem \ref{T8}.
\end{proof}

\subsection{Materialized view}
In many cases, the view update to be materialized, i.e. the least
Herbrand Model is computed and kept, for efficient query answering.
In such a situation, rational hitting sets can be computed without
performing any minimality test. The idea is to transform the given
$IDB$ with respect to the materialized view.

\begin{definition} [$IDB^+$ Transformation] Let $IDB\cup  EDB$ be a given
database. Let $S$ be the Least Herbrand Model of this database.
Then, $IDB^+$ is defined as the transformation of $IDB$ with respect to $S$.
\end{definition}

\begin{note} If $IDB$ is a stratifiable deductive database then the
transformation introduced above is not necessary.
\end{note}

\begin{definition}[Update Tableau based on Materialized view] An update tableau
based on materialized view for a database $IDB\cup EDB$ and delete
request $\neg A$ is a hyper tableau $T$ for $IDB^+\cup\{\neg
A\leftarrow \}$ such that  every open branch is finished.
\end{definition}

\begin{definition} [$IDB^-$ Transformation] Let $IDB\cup  EDB$ be a given
database. Let $S_1$ be the Least Herbrand Model of this database  (that is either body or head  empty).
Then, $IDB^-$ is defined as the transformation of $IDB$ with respect to $S_1$.
\end{definition}

\begin{definition}[Update magic set rule based on Materialized view] An update
magic set rule based on materialized view for a database $IDB\cup
EDB$ and insert request $A$ is a magic set $M$ for
$IDB^+\cup\{A\leftarrow \}$ such that  every close branch is
finished.
\end{definition}

Now the claim is that every model of $IDB^+\cup\{\neg A\leftarrow
\}$ ($A\leftarrow$) constitutes a rational hitting set for the
deletion and insertion of the ground view atom $A$. So, the
algorithm works as follows: Given a database and a view update
request, we first transform the database with respect to its Least Herbrand
Model (computation of the Least Herbrand Model can be done as a
offline preprocessing step. Note that it serves as materialized view
for efficient query answering). Then the hyper tableaux calculus
(magic set rule) is used to compute models of this transformed
program. Each model represents a rational way of accomplishing the
given view update request. This is formalized in Algorithms
6 and 7.

Like the approach with minimality test, this algorithm runs on not
polynomial space. This approach require minimality test, but our
focus on integrity constrain open/close branch. Again, this requires
some offline pre-processing of computing the Least Herbrand Model.
Note that, our future direction to construct minimality test, this
method may generate a non-minimal (but rational) hitting set.

\begin{example}

Given stratifiable (definite) deductive database $DDB=IDB \cup EDB\cup IC$ and insert $p$.

$$\begin{array}{cccccc} IDB:&p \leftarrow
a&\hspace{0.5cm}EDB:&c\leftarrow&\hspace{1.2cm}IC:&\emptyset\\
&q \leftarrow a&&d\leftarrow&&\\
&q \leftarrow c \wedge b&&&&\\
&q\leftarrow p&&&&\\
\end{array}$$
\end{example}

$IDB^+$ Transformation:

$$\begin{array}{cccccc} IDB^+:&\neg a \leftarrow
\neg p&\hspace{0.5cm}EDB:&c\leftarrow&\hspace{1.2cm}IC:&\emptyset\\
&\neg a \leftarrow \neg q&& d\leftarrow&&\\
&\neg c \lor \neg b \leftarrow \neg p&&&&\\
&\neg p \leftarrow \neg q&&&&\\
\end{array}$$

$IDB^-$ Transformation:  (Body empty)

$$\begin{array}{cccccc} IDB^-:&p \lor \neg a \leftarrow
&\hspace{0.5cm}EDB:&c\leftarrow&\hspace{1.2cm}IC:&\emptyset\\
&q \lor \neg a \leftarrow &&d\leftarrow&&\\
&q \lor \neg c \lor \neg b \leftarrow &&&&\\
&q\lor \neg p \leftarrow &&&&\\
\end{array}$$

$IDB^-$ Transformation:  (Head empty)

$$\begin{array}{cccccc} IDB^-:&\leftarrow
\neg p \wedge a&\hspace{0.5cm}EDB:&c\leftarrow&\hspace{1.2cm}IC:&\emptyset\\
&\leftarrow \neg q \wedge a&&d\leftarrow&&\\
&\leftarrow \neg g \wedge c \wedge b&&&&\\
&\leftarrow \neg q \wedge p&&&&\\
\end{array}$$

\begin{figure}[h]
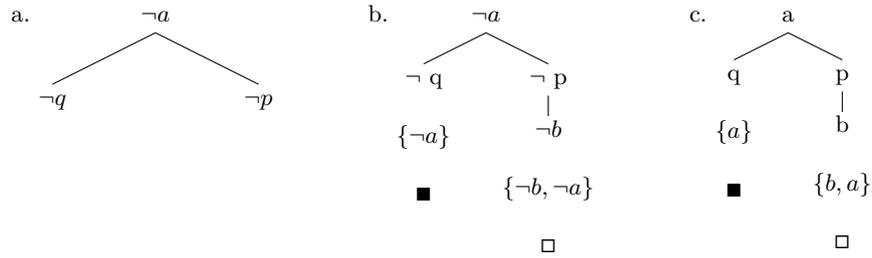

\begin{center}
\qtreecenterfalse a. \Tree [ ${\neg q}$ !\qsetw{5cm} ${\neg p}$ ]
.${\neg a}$ \hskip 0.5in b. \Tree [ {$\neg$ q \\\\
$\{ \neg a \}$ \\\\ $\blacksquare$} !\qsetw{2cm} [ {$\neg b$
\\\\ $\{ \neg b, \neg a \}$ \\\\ $\Box$ } ] .{$\neg$ p} ] .${\neg a}$
\hskip 0.5in c. \Tree [ {q \\\\ $\{ a \}$ \\\\ $\blacksquare$}
!\qsetw{2cm} [ {b \\\\ $\{ b, a \}$ \\\\ $\Box$ } ] .p ] .a
\caption{$IDB^+$ and $IDB^-$ transformation with hitting set}
\label{Figure 5.3}
\end{center}
\end{figure}

The Least Herbrand Model of this database is $\{p,q,a,b\}$. The
transformed database $IDB^+$ and $IDB^-$  based on this model,
together with an update tableaux for insertion request $p$ based on
materialised view is as above figure:

Observe that the last two clauses are never used and the necessarily
failing attempt of deleting t to delete p is never made, thus
greatly reducing the search space. Also note that the two cuts with
only EDB atoms $\{a,b\}$ and $\{a\}$ are exactly the two locally
minimal explanations for $p$. The two open branches provide the two
models of $IDB^+ \cup \{\neg p\}$ ($IDB^- \cup \{p\}$ which stand
for the hitting sets $\{a,b\}$ and $\{a\}$. Clearly, $\{a,b\}$ not
minimal.

$$\begin{array}{cc}\hline
\text{\bf Algorithm 6} &  \text{\rm View
update algorithm based on Materialized view}\\\hline \text{\rm
Input}:&
\text{\rm A stratifiable (definite) deductive database}~DDB=IDB\cup EDB\cup IC\\
&\text{\rm an literals}~\mathcal{A}\\
\text{\rm Output:}&\text{\rm A new stratifiable (definite) database}~IDB\cup EDB'\cup IC\\
\text{\rm begin}&\\
~~1.&\text{\rm Let}~ V :=\{ c\in IC~|~IDB\cup IC~\text{\rm
inconsistent
with}~\mathcal{A}~\text{\rm with respect to}~c~\}\\
&\text{\rm While}~(V\neq \emptyset)\\
~~2.&\hspace{-1.9cm}\text{\rm Construct a complete SLD-tree for} \leftarrow\mathcal{A}~\text{\rm with respect to DDB.}\\
~~3.&\hspace{-0.7cm}\text{\rm For every successful branch $i$:construct}~\Delta_{i}=\{D~|~D\in EDB \}\\
&\text{\rm and D is used as an input clause in branch $i$}.\\
&\text{\rm Construct a branch i of an update tableau based on view}\\
&\text{\rm for}~ IDB\cup EDB~\text{\rm and delete request}~\neg A.\\
&\text{\rm  Produce}~IDB\cup EDB \backslash HS(i)~\text{\rm as a result}\\
~~4.&\text{\rm For every unsuccessful branch $j$:construct}~\Delta_{j}=\{D~|~D\in EDB \}\\
&\text{\rm and D is used as an input clause in branch $j$}.\\
&\text{\rm Construct a branch j of an update magic set rule
based on view}\\
&\text{\rm for}~ IDB\cup EDB~\text{\rm and insert request}~A.\\
&\text{\rm  Produce}~IDB\cup EDB \backslash HS(j)~\text{\rm as a
result}\\
&\text{\rm Let}~ V :=\{ c\in IC~|~IDB\cup IC~\text{\rm inconsistent
with}~\mathcal{A}~\text{\rm with respect to}~c~\}\\
&\hspace{-0.7cm}\text{\rm return}\\
~~5.&\text{\rm Produce}~DDB~\text{\rm as the result.}\\
\text{\rm end.}&\\\hline
\end{array}$$

\begin{example}
$$\begin{array}{cccccc} IDB:&p \leftarrow
a&\hspace{0.5cm}EDB:&f\leftarrow&\hspace{1.2cm}IC:&\leftarrow b\\
&q \leftarrow a&&g\leftarrow&&\\
&p \leftarrow b \wedge f&&&&\\
&q \leftarrow b \wedge f&&&&\\
&p \leftarrow g \wedge a&&&&\\
\end{array}$$
\end{example}

Suppose we want to insert $p$. First, we need to check consistency
with IC and afterwards, we have to find $\Delta_{i}$ and
$\Delta_{j}$ via tree deduction.

From algorithm 6 or 7~(only different is
checking IC condition), the above example execute following steps:

$$\begin{array}{cc}\hline
\text{\bf Algorithm 7} &  \text{\rm View
update algorithm based on Materialized view}\\\hline \text{\rm
Input}:&
\text{\rm A stratifiable (definite) deductive database}~DDB=IDB\cup EDB\cup IC\\
&\text{\rm an literals}~\mathcal{A}\\
\text{\rm Output:}&\text{\rm A new stratifiable (definite) database}~IDB\cup EDB'\cup IC\\
\text{\rm begin}&\\
~~1.&\hspace{-1.9cm}\text{\rm Construct a complete SLD-tree for} \leftarrow\mathcal{A}~\text{\rm with respect to DDB.}\\
~~2.&\hspace{-0.7cm}\text{\rm For every successful branch $i$:construct}~\Delta_{i}=\{D~|~D\in EDB \}\\
&\text{\rm and D is used as an input clause in branch $i$}.\\
~~3.&\text{\rm For every unsuccessful branch $j$:construct}~\Delta_{j}=\{D~|~D\in EDB \}\\
&\text{\rm and D is used as an input clause in branch $j$}.\\
~~4.&\text{\rm Let}~ V :=\{ c\in IC~|~IDB\cup IC~\text{\rm
inconsistent
with}~\mathcal{A}~\text{\rm with respect to}~c~\}\\
&\text{\rm While}~(V\neq \emptyset)\\
&\text{\rm Construct a branch i of an update tableau satisfying
based on view}\\
&\text{\rm for}~ IDB\cup EDB~\text{\rm and delete request}~ A.\\
&\text{\rm  Produce}~IDB\cup EDB \backslash HS(i)~\text{\rm as a result}\\
&\text{\rm Construct a branch j of an update magic set rule
based on view}\\
&\text{\rm for}~ IDB\cup EDB~\text{\rm and insert request}~A.\\
&\text{\rm  Produce}~IDB\cup EDB \backslash HS(j)~\text{\rm as a
result}\\
&\text{\rm Let}~ V :=\{ c\in IC~|~IDB\cup IC~\text{\rm inconsistent
with}~\mathcal{A}~\text{\rm with respect to}~c~\}\\
&\hspace{-0.7cm}\text{\rm return}\\
~~5.&\text{\rm Produce}~DDB~\text{\rm as the result.}\\
\text{\rm end.}&\\\hline
\end{array}$$

\begin{enumerate}
\item[] \begin{center} \textbf{Step number with execution} \end{center}
\hspace{1cm}
\item[(Input)] $IDB: p \leftarrow
a, q\leftarrow a, p\leftarrow b \wedge f,
q\leftarrow b \wedge f, p\leftarrow g \wedge a\\
EDB: f\leftarrow, g\leftarrow\\
IC: \leftarrow b$
\item[(0)] $\{p \leftarrow a;~q \leftarrow a;~p\leftarrow b,f;~q\leftarrow b,f;~p\leftarrow g,a;~b;~g$\}
\item[(1)] $\{V = b\}$
\item[(2.1)] \qtreecenterfalse a. \Tree [ q !\qsetw{5cm} p ] .a \hskip
0.5in b. \Tree [ [ [ q !\qsetw{5cm} p ] .a ] .b ] .f \hskip
0.5in c. \Tree [ [ [ q !\qsetw{7cm} [ g ] .p ] .a ] .b ] .f \\
\begin{center} d. \Tree [ [ [ [ {q\\\\$\blacksquare$}  !\qsetw{7cm} [ {g\\\\$\Box$} ] .p ] .a ] .b ] .f ]
.g \end{center}
\item[(2.2)] \qtreecenterfalse a. \Tree [ q !\qsetw{5cm} p ] .a \hskip
0.5in b. \Tree [ q !\qsetw{5cm} p ] .a \hskip
0.5in c. \Tree [ q !\qsetw{7cm} [ g ] .p ] .a \\
\begin{center} d. \Tree [ [ {q\\\\$\blacksquare$} !\qsetw{7cm} [ {g\\\\$\Box$} ] .p ] .a ] .g
\end{center}
\item[(3-4)] $\Delta_{i} =\{a,g\}$ and $\Delta_{j} =\{ \}$
\item[(5)] $p\leftarrow a;~q \leftarrow a;~p\leftarrow g\wedge a;~a,~f,~g$
\item[(Output)] $IDB: p\leftarrow
a, q\leftarrow a, p\leftarrow g \wedge a\\
EDB': a\leftarrow;~f\leftarrow;~g\leftarrow\\ IC: \leftarrow b$
\end{enumerate}

This approach for view update may not satisfy (KB*7.1) in general.
But, as shown in the sequel, conformation to (KB*7.3) is guaranteed
and thus this approach results in rational update.

\begin{lemma} \label{l6} Let $T$ be an update tableau based on a materialized view for
$IDB\cup  EDB$ and delete request $\neg A$ ($A$), Similarly, let $M$
be an update magic set rule. Let $S$ be the set of all $EDB$-closed
locally minimal abductive explanations for $A$ with respect to $IDB$. Let $S'$
be the set of all $EDB$-cuts of $T$ and  $EDB$-seeds of $M$. Then,
the following hold:
\begin{enumerate}
\item[$\bullet$] $S\subseteq S'$.
\item[$\bullet$] $\forall \Delta'\in S':\exists \Delta\in S ~ s.t.~ \Delta\subseteq \Delta'$.
\item[$\bullet$] $\forall \Delta'\in S':\Delta'\subseteq \bigcup S$.
\end{enumerate}
\end{lemma}

\vspace{0.5cm}
\begin{proof}
\hspace{0.5cm}
\begin{enumerate}
\item[1.] Consider a $\Delta (\Delta\in\Delta_i \cup \Delta_j)\in S$.
 We need to show that $\Delta$ is generated by algorithm
6 at step 2. From observation \ref{l1}, it is clear that there
exists a $A$-kernel $X$ of $DDB_G$ s.t. $X \cap EDB = \Delta_j$ and
$X \cup EDB = \Delta_i$. Since $X \vdash A$, there must exist a
successful derivation for $A$ using only the elements of $X$ as
input clauses and similarly $X \nvdash A$. Consequently $\Delta$
must have been constructed at step 2.
\item[2.] Consider a $\Delta'((\Delta'\in\Delta_i \cup \Delta_j)\in S'$. Let $\Delta'$ be
constructed from a successful(unsuccessful) branch $i$ via
$\Delta_i$($\Delta_j$). Let $X$ be the set of all input clauses used
in the refutation $i$. Clearly $X\vdash A$($X\nvdash A$). Further,
there exists a minimal (with respect to set-inclusion) subset $Y$ of $X$ that
derives $A$ (i.e. no proper subset of $Y$ derives $A$. Let $\Delta =
Y \cap EDB$ ($Y \cup EDB$). Since IDB does not (does) have any unit
clauses, $Y$ must contain some EDB facts, and so $\Delta$ is not
empty (empty) and obviously $\Delta\subseteq \Delta'$. But, $Y$ need
not (need) be a $A$-kernel for $IDB_G$ since $Y$ is not ground in
general. But it stands for several $A$-kernels with the same
(different) EDB facts $\Delta$ in them. Thus, from observation \ref{l1},
$\Delta$ is a DDB-closed locally minimal abductive explanation for
$A$ with respect to $IDB_G$ and is contained in $\Delta'$.
minimal.
\end{enumerate}
\end{proof}

\begin{lemma} \label{l7}Let S and S' be sets of sets s.t.
$S\subseteq S'$ and for every member X of $S'\backslash S$: X is a superset of some member of S and X is a subset of $\bigcup S$. Then, a set H is a hitting set for S iff it
is a hitting set for S'
\end{lemma}
\vspace{0.5cm}

\begin{proof}
\hspace{0.5cm}
\begin{enumerate} 

\item[1.] (\textbf{If part})~Given that $H$ is a hitting set for $S'$, we have to
show that it is a hitting set for $S$ too. First of all, observe
that $\bigcup S = \bigcup S'$, and so $H \subseteq \bigcup S$.
Moreover, by definition, for every non-empty member $X$ of $S'$, $H
\cap X$ is not empty. Since $S \subseteq S'$, it follows that $H$
is a hitting set for $S$ too.\\

(\textbf{Only if part})~Suppose $H$ is a hitting set for $S$. As
observed above, $H \subseteq \bigcup S'$ . By definition, for every
non-empty member $X\in S$, $X \cap H$ is not empty. Since every
member of $S'$ is a superset of some member of $S$, it is clear that $H$ hits
every member of $S'$, and hence a hitting set for $S'$ .
\end{enumerate}
\end{proof}

\begin{lemma} \label{l8} Let $T$ and $M$ as in Lemma \ref{l6}. Then $HS(b)$ is a rational hitting set
for $A$, for every open finished branch $b$ in $T$ (close finished
branch $b$ in $M$).
\end{lemma}

\begin{proof}
Follows from the observation \ref{l1}~(materialized view) in and (Behrend
\& Manthey 2008)
\end{proof}

\begin{theorem} \label{T16} Algorithms 6 and 7 are rational, in the sense that they satisfy all
the rationality postulates (KB*1) to (KB*6) and (KB*7.3).
\end{theorem}

\begin{proof}
Follows from Observation \ref{l1},\ref{l8} and Theorem \ref{T8}.
\end{proof}

\subsection{Incomplete to Complete Information}
Many of the proposals in the literature on incomplete databases have
focussed on the extension of the relational model by the
introduction of null values. In this section, we show how view
update provides completion  of incomplete information. More detailed
surveys of this area can be found in (Meyden 1998).

The earliest extension of the relational model to incomplete
information was that of Codd (Codd 1979) who suggested that missing
values should be represented in tables by placing a special
\emph{null value} symbol $'*'$  at any table location for which the
value is unknown. Table 5.3, shows an example of a database using
this convention. Codd proposed an extension to the relational
algebra for tables containing such nulls, based on three valued
logic and a null substitution principle.

In terms of our general semantic scheme, the intended semantics of a
database $D$ consisting of Codd tables can be described by defining
$Mod(D)$ to be the set of structures $M_{D'}$, where $D'$ ranges
over the relational databases obtained by replacing each occurrence
of $'*'$ in the database $D$ by some domain value. Different values
may be substituted for different occurrences.

A plausible integrity constraint on the meaning of a relational
operator on tables in $\mathcal{T}$ is that the result should be a
table that represents the set of relations obtained by pointwise
application of the operator on the models of these tables.  For
example, if $R$ and $S$ are tables in $\mathcal{T}$ then the result
of the join $R \Join S$ should be equal to a table T in
$\mathcal{T}$ such that

$$Mod(T)=\{r\Join t~|~r\in Mod(R),~ s\in Mod(S)\}$$

In case the definitions of the operators satisfy this integrity
constraint (with respect to the definition of the semantics \it Mod
\rm on $\mathcal{T}$).

Let us consider what above equation requires if we take $R$ and $S$
to be the Codd Tables 5.3.  First of all, note that in each model,
if we take the value of the null in the tuple (delhibabu,*) to be
$v$, then the join will contain one tuples (delhibabu, $v$), which
include the value $v$. If $T$ is to be a Codd table, it will need to
contain tuples (delhibabu,$X$) to generate each of these tuples,
where $X$ are either constants or '*'. We now face a problem. First,
$X$ cannot be a constant $c$, for whatever the choice of $c$ we can
find an instance $r\in Mod(R)$ and $s\in Mod(S)$ for which the tuple
(delhibabu, $c$) does not occur in $r\Join s$. If they were, $X$
would have their values in models of $T$ assigned independently.

Here the repetition of $*$ indicates that the \it same \rm value is
to be occurrence of the null in constructing a model of the table.
Unfortunately, this extension does not suffice to satisfy the
integrity constraint ($\forall x,y,z$ (y=x) $\leftarrow$
group\_chair(x,y) $\wedge$ group\_chair(x,z)).
\begin{table}[h]
\begin{center}
$\begin{array}{|c|c|}\hline
  \text{\rm Staff}&\text{\rm Group}\\\hline
  \text{\rm delhibabu}&\text{\rm infor1}\\
  \text{\rm delhibabu}&\text{\rm *}\\\hline
 \end{array}$
  $~~~~~~~~~~\begin{array}{|c|c|} \hline
  \text{\rm Group}&\text{\rm Chair}\\\hline
  \text{\rm infor1}&\text{\rm mattias}\\
  \text{\rm *}&\text{\rm aravindan}\\  \hline
 \end{array}$
\end{center}

\centering \caption{\rm Base Table after Transaction}\end{table}

In the model of these tables in which $*=infor1$,  the join contains
the tuple (delhibabu, infor1) and (infor1, aravindan).

$$\text{\rm If}~*_1=\text{\rm infor1 then (delhibabu, infor1)}\in R\Join S$$
$$\text{\rm If}~*_2=\text{\rm infor1 then (infor1, aravindan)}\in R\Join S$$

The following table shows when transaction is made to base table:
\begin{table}[h]
$$\begin{array}{|c|c|c|}\hline
\text{\rm Staff}&\text{\rm Group}&\text{\rm Chair}\\\hline
  \text{\rm delhibabu}& \text{\rm infor1} &\text{\rm mattias}\\
\text{\rm delhibabu}& \text{\rm * } &\text{\rm aravindan}\\\hline
\end{array}$$

\centering \caption{\rm $s\otimes r$  after Transaction }\end{table}

The following table shows completion of incomplete information with
application of integrity constraint and redundancy:
\begin{table}[h]
$$\begin{array}{|c|c|c|}\hline
\text{\rm Staff}&\text{\rm Group}&\text{\rm Chair}\\\hline
  \text{\rm delhibabu}& \text{\rm infor1} &\text{\rm aravindan}\\\hline
\end{array}$$

\centering \caption{\rm~Redundant Table}\end{table}

\subsection{A Comparative Study of view update algorithm and integrity constraint with our
axiomatic method}

During the process of updating database, two interrelated problems
could arise. On one hand, when an update is applied to the database,
integrity constraints could become inconsistent with request, then
stop the process. On the other hand, when an update request consist
on updating some derived predicate, a view update mechanism must be
applied to translate the update request into correct updates on the
underlying base facts. Our work focus on the integrity
constraint maintenance approach. In this section, we extend Mayol
and Teniente's (Mayol \& Teniente 1999) survey for view update and
integrity constraint.

The main aspects that must be taken into account during the process
of view update and integrity constraint enforcement are the
following: the problem addressed, the considered database schema,
the allowed update requests, the used technique, update change and
the obtained solutions. These six aspects provide the basic
dimensions to be taken into account. We explain each dimension in
this section and results are presented in Appendix.

\textbf{Problem Addressed}
\begin{enumerate}
\item[]\hspace{-0.6cm}(\emph{Type})~-~What kind of program to be used (stratified (S), Horn clause (H), Disjunctive database (D), Normal Logic program (N) and Other (O)).
\item[]\hspace{-0.6cm}(\emph{View Update})~-~Whether they are able to deal with view update or not (indicated by Yes or
No in the second column in the appendix section).
\item[]\hspace{-0.6cm}(\emph{integrity-constraint Enforcement})~-~Whether they incorporate an integrity constraint checking (C)or an integrity
constraint maintenance (M) or both apply (C-M) approach (indicated
by check or maintain in the third column).
\item[]\hspace{-0.6cm}(\emph{Run/Comp})~-~Whether the method follows a run-time (transaction) or a compile-time approach (indicated
by Run or Compile in the fourth column).
\end{enumerate}

\textbf{Database Schema Considered}
\begin{enumerate}
\item[]\hspace{-0.6cm}(\emph{Definition Language})~-~The language mostly used is logic (L), although some
methods use a relational language (R) and also uses an
object-oriented (O-O).
\item[]\hspace{-0.6cm}(\emph{The DB Schema Contains Views})~-~ All methods
that deal with view update need views to be defined in the database
schema. Some of other method allow to define views.
\item[]\hspace{-0.6cm}(\emph{Restrictions Imposed on the Integrity Constraints})~-~ Some proposals impose
certain restrictions on the kind of integrity constraints that can
be defined and, thus, handled by their methods.
\item[]\hspace{-0.6cm}(\emph{Static vs Dynamic Integrity Constraints})~-~ Integrity
constraints may be either static (S), and impose restrictions
involving only a certain state of the database, or dynamic (D).
\end{enumerate}

\textbf{Update Request Allowed}
\begin{enumerate}
\item[]\hspace{-0.6cm}(\emph{Multiple Update Request})~-~An update request is multiple if it contains several
updates to be applied together to the database.
\item[]\hspace{-0.6cm}(\emph{Update Operators})~-~Traditionally, three different basic update operators are
distinguished: insertion ($\iota$), deletion ($\delta$) and
modification ($\chi$). Modification can always be simulated by a
deletion followed by an insertion.
\end{enumerate}

\textbf{Update Processing Mechanism}
\begin{enumerate} 
\item[]\hspace{-0.6cm}(\emph{Applied Technique})~-~ The techniques applied by these methods can be classified
according to four different kinds of procedures, unfolding, SLD,
active and predefined programs, respectively.
\item[]\hspace{-0.6cm}(\emph{Taking Base Facts into Account}~-~ Base facts can either be taken into
account or not during update processing.
\item[]\hspace{-0.6cm}(\emph{User Participation})~-~ User participation during update
processing or not.
\end{enumerate}

\textbf{Update Changing Mechanism}
\begin{enumerate}

\item[]\hspace{-0.6cm}(\emph{Type of modification})~-~
Changing table by singleton like atom (S), sets of each types of
modification(SS) and group of changes (G).
\item[]\hspace{-0.6cm}(\emph{Changing Base Fact})~-~
Base fact can be changed either using principle of minimal change or
complete change (maximal change).
\item[]\hspace{-0.6cm}(\emph{Changing View Definition})~-~ Whether
update process view definition is changed or not.
\end{enumerate}

\textbf{Obtained Solution}
\begin{enumerate}
\item[]\hspace{-0.6cm}(\emph{Our Axiom follow})~-~ When update
process done, we are comparing our axiomatized method and which
relevance policy holds ((KB*1) to (KB*6),(KB*7.1),(KB*7.2) and
(KB*7.3) is enumerated 1 to 9)
\item[]\hspace{-0.6cm}(\emph{Soundness})~-~ A
method is correct if it only obtains solutions that satisfy the
requested update, note NP mean Not Proved.
\item[]\hspace{-0.6cm}(\emph{Completeness)}~-~ A method is complete if it is able to obtain all
solutions that satisfy a given update request.
\end{enumerate}

Results of each method according to these features are summarized in Appendix.

\section{Belief Update Vs Database Update} In this section we give
overview of how belief update is related to database update. This
section is motivated by works of Hansson's (Hansson 1991) and Keller's (Keller 1985)

\subsection{View update vs Database update}

\begin{figure}
\begin{center}
   \includegraphics[height=4cm,angle=0]{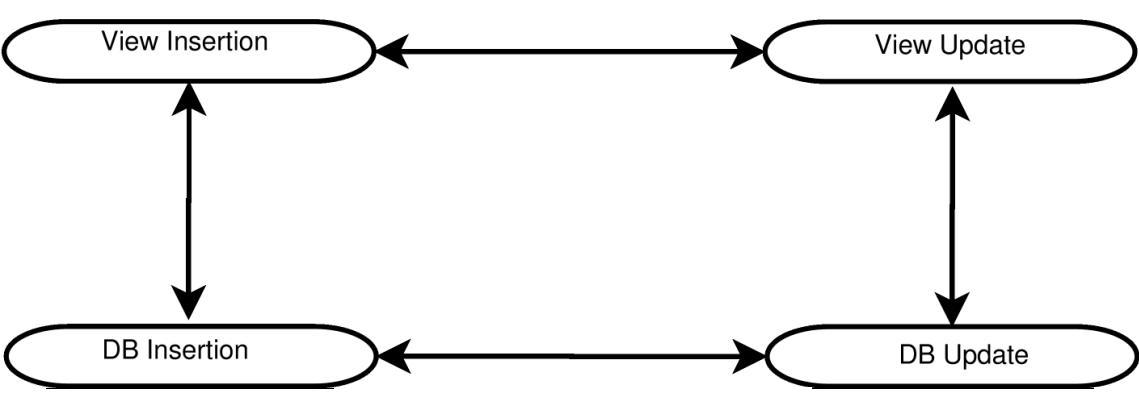}
   \caption{View Update Vs Database Update}
   \end{center}
\end{figure}

The view update problem exists already three decades Chen $\&$ Liao 2010 and Minker 1996.
We are taking proof from Keller 1985, given a view definition of the
question of choosing a view update translator arises.

This requires understanding the ways in which individual view update
requests may be satisfied by database updates. Any particular view
update request may result in a view state that does not correspond
to any database state. Such a view update request may not be
translated without relaxing the constraint which precludes view side
effects. Otherwise, the update request is rejected by the view
update translator. If we are lucky, there will be precisely one way
to perform the database update that results in the desired view
update. Since the view is many-to-one, the new view state may
correspond to many database states. Of these database states, we
would like to choose one that is "as close as possible", under some
measure, to the original database state. That is, we would like to
minimize the effect of the view update on the database.

\subsection{Belief update vs Database update}

If we look closely to the section (6.3 and 8.1), we easily find the
following results. With evidence of Hansson's (Hansson 1991) and
Liberatore (Liberatore $\&$ Schaerf 2004). Here BR and BU mean Belief Revision and
Belief Update, respectively.

\begin{figure}
\begin{center}
   \includegraphics[height=5cm,angle=0]{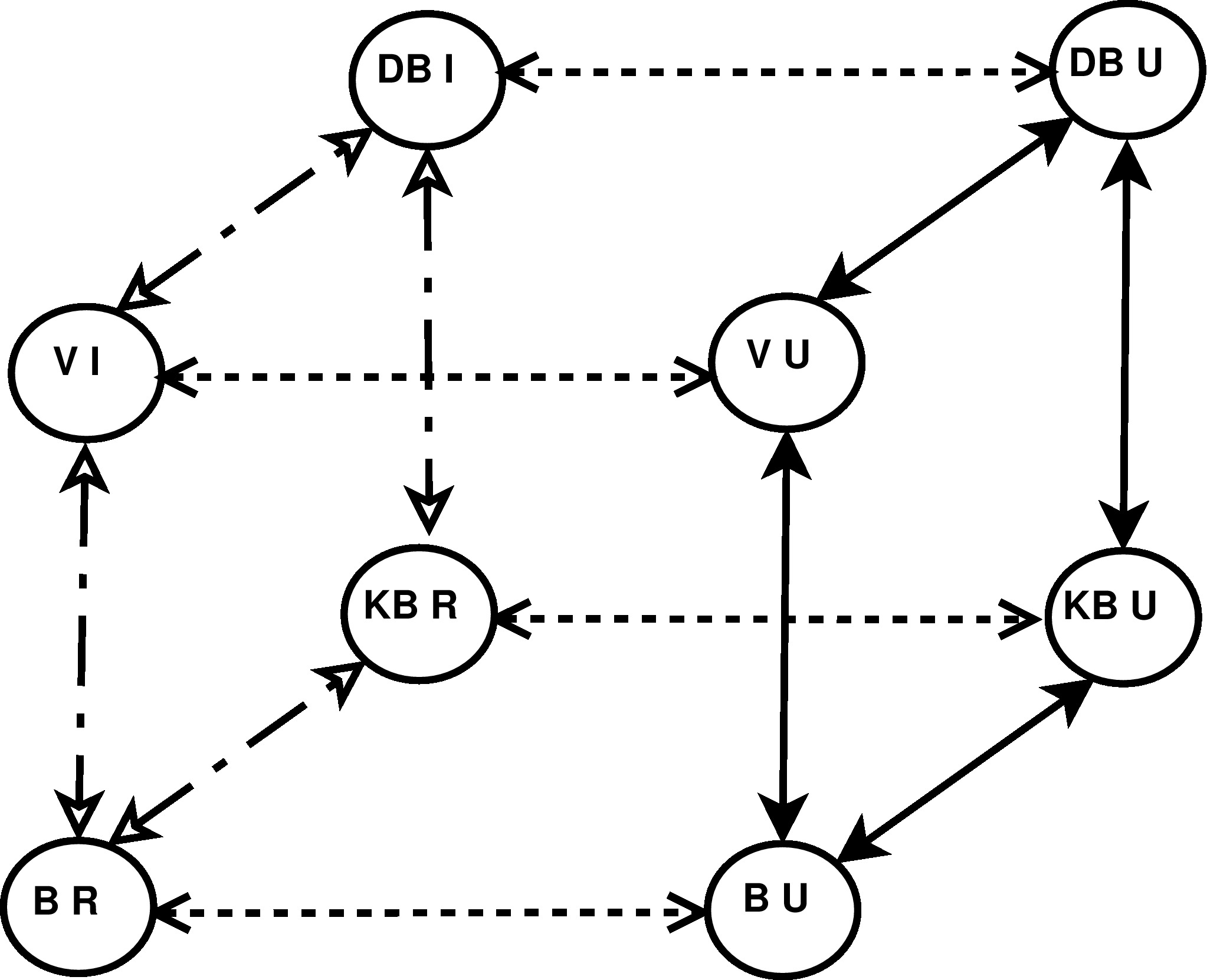}
   \caption{Belief Update Vs Database Update}
   \end{center}
\end{figure}

\section{Abductive framework for Horn knowledge base dynamics}

As discussed in Section 5, we introduced Horn knowledge
base dynamics to deal with two important points: first, to handle
belief states that need not be deductively closed; and the second
point is the ability to declare certain parts of the belief as
immutable. There is yet another, radically new approach to handle
this problem, and this Section addresses this. In fact, this
approach is very close to the Hansson's (Hansson 1992) dyadic
representation of belief. In the similar abduction model by Boutilier \& Beche 1995 and Pagnucco 1996 Here, we consider the immutable part as
defining a new logical system. By a logical system, we mean that it
defines its own consequence relation and closure operator. Based on
this, we provide an abductive framework for Horn knowledge base
dynamics.

A first order language consists of an alphabet $\mathcal{A}$ of a
language $\mathcal{L}$. We assume a countable universe of variables
\textmd{Var}, ranged over x,y,z, and a countable universe of
relation (i.e. predicate) symbols, ranged over by $\mathcal{A}$ are finite. The
following defines \textmd{FOL}, the language of first order
logic with equality and binary relations:

\begin{center}
$\varphi ::=$ $x=x$ $|$ $a(x,x)$ $|$ $\neg\varphi$ $|$ $\bigvee
\phi$ $|$ $\bigwedge\phi$ $|$ $\exists X:\phi$.
\end{center}

Here $\phi\subseteq FOL$ and $X\subseteq Var$ are finite sets of
formulae and variables, respectively.

\begin{definition}[Normal Logic Program (NLP) \text{[22]}]
By an alphabet $\mathcal{A}$ of a language $\mathcal{L}$ we mean
disjoint sets of constants, predicate symbols, and function symbols,
with at least one constant. In addition, any alphabet is assumed to
contain a countably infinite set of distinguished variable symbols.
A term over $\mathcal{A}$ is defined recursively as either a
variable, a constant or an expression of the form $f(t_1,...,t_n)$
where f is a function symbol of $\mathcal{A}$, n its arity, and the
$t_i$ are terms. An atom over  $\mathcal{A}$ is an expression of the
form $P(t_1,...,t_n)$ where P is a predicate symbol of $\mathcal{A}$
and the $t_i$ are terms. A literal is either an atom A or its
default negation not A. We dub default literals those of the form
not A. A term (atom or literal) is said ground if it does not
contain variables. The set of all ground terms (atoms) of
$\mathcal{A}$ is called the Herbrand universe (base) of
A. A Normal Logic Program is a possibly infinite set of
rules (with no infinite descending chains of syntactical dependency)
of the form:
\begin{center}
$H\leftarrow B_1,...,B_n,not~C_1,...,not~C_m,~(with~m,n \geq 0)$
\end{center}

Where H,$B_i$ and $C_j$ are atoms, and each rule stands for
all its ground instances. In conformity with the standard
convention, we write rules of the form $H\leftarrow$ also simply as
H (known as fact). An NLP P is called definite if none of its rules
contain default literals. H is the head of the rule r, denoted by
head(r), and body(r) denotes the set $\{B_1,...,B_n,not~C_1,...,not~
C_m\}$ of all the literals in the body of r.
\end{definition}

When doing problem  modeling with logic programs, rules of the form

\begin{center}
$\bot\leftarrow B_1,...,B_n,not~C_1,...,not~C_m,~(with~m,n \geq 0)$
\end{center}

with a non-empty body are known as a type of integrity constraints
(ICs), specifically denials, and they are normally used to prune out
unwanted candidate solutions. We abuse the $\emph{not}$ default
negation notation applying it to non-empty sets of literals too: we
write not S to denote $\{not ~ s: s\in S\}$, and duality of $not~ not~
a\equiv a$. When S is an arbitrary, non-empty set of literals
$S=\{B_1,...,B_n,not~C_1,...,not~ C_m\}$ we use:

\begin{enumerate}
\item[-] $S^+$ denotes the set $\{B_1,\ldots,B_n\}$ of positive literals in $S$ .
\item[-] $S^-$ denotes the set $\{not~ C_1,\ldots, not~ C_m\}$ of negative literals
in $S$ .
\item[-] $|S| = S^+ \cup (not~ S^-)$ denotes the set
$\{B_1,\ldots,B_n,C_1,\ldots ,C_m\}$ of atoms of $S$.
\end{enumerate}

As expected, we say a set of literals $S$ is consistent iff $S^+
\cap |S^-| = \emptyset$. We also write $heads(P)$ to denote the set
of heads of non-IC rules of a (possibly constrained) program $P$,
i.e. $heads(P) = \{head(r) : r \in P\} \backslash \{\bot\}$, and
$facts(P)$ to denote the set of facts of $P$  - $facts(P) =
\{head(r) : r\in P \land body(r) =\emptyset\}$.

\begin{definition}[Level mapping\text{[4]}] Let P be a normal logic program and $B_{P}$ its
Herbrand base. A \emph{level} mapping for P is a function
$\parallel: B_{P} \rightarrow \mathbb{N}$ of ground atoms to natural
numbers. The mapping $\parallel$ is also extended to ground literals
by assigning $\mid\neg A \mid$ = $\mid A\mid$ for all ground atoms
$A\in B_{P}$. For every ground literal L, $\mid L \mid$ is called as
the \emph{level} of L in P.
\end{definition}

\begin{definition}[Acyclic program \text{[4]}] Let P be a normal logic program and $\parallel$ a level mapping
for P.  P is called as acyclic with respect to $\parallel$ if for
every ground clause $H\leftarrow L_{1},...,L_{n}~(with~n \geq 0
~and~ finit)$ in P the level of A is higher then the level of every
$L_{i}$ (1 $\leq $i$ \leq$ n). Moreover P is called acyclic if P is
acyclic with respect to some level mapping for P.
\end{definition}

Unlike Horn knowledge base dynamics, where knowledge is defined as a
set of sentences, here we wish to define a Horn knowledge base KB
with respect to a language $\mathcal{L}$, as an abductive framework
$<P,Ab,IC,K>$, where,

\begin{enumerate} 
\item[*] $P$ is an acyclic normal logic program with all abducibles in P at level 0
and no non-abducible at level 0. $P$ is referred to as a
\emph{logical system}. This in conjunction with the integrity
constraints corresponds to immutable part of the Horn knowledge
base, here $P$ is defined by immutable part. This is discussed
further in the next subsection;

\item[*] $Ab$ is a set of atoms from $\mathcal{L}$, called the \emph{abducibles}. This notion is required
in an abductive framework, and this corresponds to the atoms that
may appear in the updatable part of the knowledge;

\item[*] $IC$ is the set of \it integrity constraints, \rm a set of sentences from language $\mathcal{L}$.
This specifies the integrity of a Horn knowledge base and forms a
part of the knowledge that can not be modified over time;

\item[*] $K$ is a set of sentences from $\mathcal{L}$. It is the \emph{current knowledge}, and the only part of
$KB$ that changes over time. This corresponds to the updatable part
of the Horn knowledge base. The main requirement here is that no
sentence in $K$ can have an atom that does not appear in $Ab$.
\end{enumerate}

\subsection{Logical system}

The main idea of our approach is to consider the immutable part of
the knowledge to define a new logical system. By a logical system,
we mean that $P$ defines its own consequence relation $\models_{P}$
and its closure $Cn_{p}$. Given $P$, we have the Herbrand Base
$HB_{P}$ and $G_{P}$, the ground instantiation of $P$.

An \it abductive interpretation \rm $I$ is a set of abducibles, i.e.
$I\subseteq Ab$. How $I$ interprets all the ground atoms of $L$
\footnotemark \footnotetext{the set of all the ground atoms of $L$,
in fact depends of $L$, and is given as $HB_{P}$, the Herbrand Base
of P}  is defined, inductively on the level of atoms with respect to $P$, as
follows:

\begin{enumerate} 
\item[*] An atom $A$ at level 0 (note that only abducibles are at level 0) is interpreted as: $A$ is \it true
\rm in I iff $A\in I$, else it is \it false \rm in $I$.
\item[*] An atom (literal) $A$ at level $n$ is interpreted as: $A$ is true in $I$ iff $\exists$ clause $A\leftarrow L_{1},\ldots,L_{n}$
in $G_{P}$ s.t. $\forall L_{j}\;(1\leq j\leq n)$ if $L_{j}$
 is an atom then $L_{j}$ is true in $I$, else if $L_{j}$ is a negative literal
$\neg B_{j}$, then $B_{j}$ is false in I.
\end{enumerate}

This interpretation of ground atoms can be extended, in the usual
way, to interpret  sentences in $L$, as follows (where $\alpha$ and
$\beta$ are sentences):
\begin{enumerate}
\item[*] $\neg \alpha$ is true in $I$ iff $\alpha$ is false in $I$.
\item[*] $\alpha \land  \beta$ is true in $I$ iff both $\alpha$ and $\beta$ are true in $I$.
\item[*] $\alpha \lor  \beta$ is true in $I$ iff either $\alpha$ is true in $I$ or $\beta$ is true in $I$.
\item[*] $\forall \alpha$ is true in $I$ iff all ground instantiations of $\alpha$ are true in $I$.
\item[*] $\exists \alpha$ is true in $I$ iff some ground instantiation of $\alpha$ is true in $I$.
\end{enumerate}

Given a sentence $\alpha$ in $L$, an abductive interpretation $I$ is
said to be an \it abductive model \rm of $\alpha$ iff $\alpha$ is
true in $I$. Extending this to a set of sentences $K$, $I$ is a
abductive model of $K$ iff $I$ is an abductive model of every
sentence $\alpha$ in $K$.

Given a set of sentences $K$ and a sentence $\alpha$, $\alpha$ is
said to be a $P$-\it consequence \rm of $K$, written as $K
\models_{P} \alpha$, iff every abductive model of $K$ is an
abductive model of $\alpha$ also. Putting it in other words, let
$Mod(K)$ be the set of all abductive models of $K$. Then $\alpha$ is
a $P$-consequence of $K$ iff $\alpha$ is true in all abductive
interpretations in $Mod(K)$. The \it consequence operator \rm
$Cn_{P}$ is then defined as $Cn_{P}(K)=\{\alpha~|~K\models _{P}
\alpha\}=\{\alpha~|~\alpha ~ \text{\rm is true in all abductive
interpretations in} ~ Mod(K)\}$. K is said to be P-\it consistent
\rm iff there is no expression $\alpha$ s.t. $\alpha \in Cn_{P}(K)$
and $\neg \alpha\in Cn_{P}(K)$. Two sentences $\alpha$ and $\beta$
are said to be $P$-\it equivalent \rm to each other, written as
$\alpha \equiv \beta$, iff they have the same set of abductive
models , i.e. $Mod(\alpha)=Mod(\beta)$.

\subsubsection{Properties of consequences operator}\hspace{0.5cm}

Since a new consequence operator is defined, it is reasonable, to
ask whether it satisfies certain properties that are required in the
Horn knowledge base dynamics context. Here, we observe that all the
required properties, listed by various researchers in Horn knowledge
base dynamics, are satisfied by the defined consequence operator.
The following propositions follow from the above definitions, and
can be verified easily.

\begin{center} $Cn_{P}$ satisfies \it inclusion, i.e. $K\subseteq Cn_{P}(K)$. \end{center}
  \begin{center} $Cn_{P}$ satisfies \it iteration, i.e. $Cn_{P}(K)=Cn_{P}(Cn_{P}(K))$. \end{center}

Anther interesting property is \it monotony, \rm i.e. if $K\subseteq
K'$, then $Cn_{P}(K)\subseteq Cn_{P}(K')$. $Cn_{P}$ satisfies
monotony. To see this, first observe that $Mod(K')\subseteq Mod(K)$.

$Cn_{P}$ satisfies \it superclassicality \rm, i.e. if $\alpha$ can
be derived from K by first order classical logic, then $\alpha \in
Cn_{P}(K)$.

$Cn_{P}$ satisfies \it deduction \rm, i.e. if $\beta \in
Cn_{P}(K\cup \{\alpha\})$, then $(\beta\leftarrow \alpha)\in Cn(K)$.

$Cn_{P}$ satisfies \it compactness \rm, i.e. if $\alpha \in
Cn_{P}(K)$, then $\alpha\in Cn_{P}(K')$ for some finite subset $K'$
of $K$.

\subsubsection{Statics of a Horn knowledge base}\hspace{0.5cm}

The statics of a Horn knowledge base $KB$, is given by the current
knowledge K and the integrity constraints $IC$. An abductive
interpretation $M$ is an abductive model of $KB$ iff it is an
abductive model of $K\cup IC$. Let $Mod(KB)$ be the set of all
abductive models of $KB$. The \it belief set \rm represented by
$KB$, written as $KB^{\bullet}$ is given as,
$$KB^{\bullet}=Cn_{P}(K\cup IC)=\{\alpha| \alpha \; \text{\rm is true in
every abductive model of}\; KB \}.$$ A belief (represented by a
sentence in $\mathcal{L}$) $\alpha$ is \emph{accepted} in $KB$ iff
$\alpha\in KB^{\bullet}$ (i.e. $\alpha$ is true in every model of
$KB$). $\alpha$ is \it rejected \rm in $KB$ iff $\neg \alpha \in
KB^{\bullet}$ (i.e. $\alpha$ is false in every model of $KB$). Note
that there may exist a sentence $\alpha$ s.t. $\alpha$ is neither
accepted nor rejected in $KB$ (i.e. $\alpha$ is true in some but not
all models of $KB$), and so $KB$ represents a partial description of
the world.

Two Horn knowledge bases $KB_{1}$ and $KB_{2}$ are said to be \it
equivalent \rm to each other, written as $KB_{1}\equiv KB_{2}$, iff
they are based on the same logical system and their current
knowledge are $P$-equivalent, i.e. $P_{1}=P_{2},\;Ab_{1}=Ab_{2},\;
IC_{1}=IC_{2}$ and $K_{1}\equiv K_{2}$. Obviously, two equivalent
Horn knowledge bases $KB_{1}$ and $KB_{2}$ represent the same belief
set, i.e. $KB^{\bullet}_{1}=KB^{\bullet}_{2}$.


\subsection{Horn knowledge base dynamics}
In AGM (Alchourron et al. 1985b) three kinds of belief dynamics are
defined: expansion, contraction and revision. We consider all of
them, one by one, in the sequel.

\subsubsection{Expansion}\hspace{0.5cm}

Let $\alpha$ be new information that has to be added to a knowledge
base $KB$. Suppose $\neg \alpha$ is not accepted in $KB$. Then,
obviously $\alpha$ is $P$ - consistent with $IC$, and $KB$ can be
\it expanded \rm by $\alpha$, by modifying $K$ as follows:
$$KB+\alpha\equiv <P, Ab, IC, K\cup \{\alpha\}>$$
Note that we do not force the presence of $\alpha$ in the new $K$,
but only say that $\alpha$ must be in the belief set represented by
the expanded Horn knowledge base. If in case $\neg \alpha$ is
accepted in $KB$ (in other words, $\alpha$ is inconsistent with IC),
then expansion of $KB$ by $\alpha$ results in a inconsistent Horn
knowledge base with no abductive models, i.e.
$(KB+\alpha)^{\bullet}$ is the set of all sentences in
$\mathcal{L}$.

Putting it in model-theoretic terms, $KB$ can be expanded by a
sentence $\alpha$, when $\alpha$ is not false in all models of $KB$.
The expansion is defined as:
$$Mod(KB+\alpha)=Mod(KB)\cap Mod(\alpha).$$

If $\alpha$ is false in all models of $KB$, then clearly
$Mod(KB+\alpha)$ is empty, implying that expanded Horn knowledge
base is inconsistent.

\subsubsection{Revision} \label{l89}\hspace{0.5cm}

As usual, for revising and contracting a Horn knowledge base, the
rationality of the change is discussed first. Later a construction
is provided that complies with the proposed rationality postulates.

\subsubsection{Rationality postulates}\hspace{0.5cm}

Let $KB=<P,Ab,IC,K>$ be revised by a sentence $\alpha$ to result in
a new Horn knowledge base $KB\dotplus\alpha=<P',Ab',IC',K'>$.

When a Horn knowledge base is revised, we do not (generally) wish to
modify the underlying logical system P or the set of abducibles
$Ab$. This is refereed to as \it inferential constancy \rm by
Hansson (Hansson 1991 \& 1992).

\begin{enumerate} 
\item[$(\dotplus1)$] (\it Inferential constancy) $P'=P$ and $Ab'=Ab$,$IC'=IC$.
\item[$(\dotplus 2)$] (\it Success)\rm $\alpha$ is accepted in $KB\dotplus\alpha$ , i.e. $\alpha$ is true in all models of $KB\dotplus\alpha$.
\item[$(\dotplus 3)$] (\it Consistency) \rm $\alpha$ is satisfiable and $P$-consistent with IC iff $KB\dotplus \alpha$ is P-consistent,
i.e. $Mod(\{\alpha\}\cup IC)$ is not empty iff $Mod(KB\dotplus
\alpha)$ is not empty.
\item[$(\dotplus 4)$] (\it Vacuity) \rm If $\neg \alpha$ is not accepted in KB, then $KB\dotplus \alpha\equiv KB+\alpha$, i.e. if $\alpha$
is not false in all models of KB, then $Mod(KB\dotplus
\alpha)=Mod(KB) \cap Mod(\alpha)$.
\item[$(\dotplus 5)$] (\it Preservation)\rm If $KB \equiv KB'$ and $\alpha\equiv \beta$, then $KB\dotplus\alpha\equiv KB'\dotplus\beta$, i.e.
if $Mod(KB)=Mod(KB')$ and $Mod(\alpha)=Mod(\beta)$, then
$Mod(KB\dotplus\alpha)=Mod(KB\dotplus\beta)$.
\item[$(\dotplus 6)$] (\it Extended Vacuity 1)\rm $(KB\dotplus\alpha)+\beta$ implies $KB\dotplus(\alpha \land \beta)$, i.e.
$(Mod(KB\dotplus \alpha)\cap Mod(\beta))\subseteq
Mod(KB\dotplus(\alpha \land \beta))$.
\item[$(\dotplus 7)$] (\it Extended Vacuity 2)\rm If $\neg \beta$ is not accepted in $(KB\dotplus \alpha)$, then
$KB\dotplus(\alpha\land \beta)$ implies $(KB\dotplus \alpha)+\beta$,
i.e. if $\beta$ is not false in all models of $KB\dotplus\alpha$,
then $Mod(KB\dotplus(\alpha\land\beta))\subseteq
(Mod(KB\dotplus\alpha)\cap Mod(\beta))$.
\end{enumerate}

\subsubsection{Construction}\hspace{0.5cm}

Let $\mathcal{S}$ stand for the set of all abductive interpretations
that are consistent with $IC$, i.e. $\mathcal{S}=Mod(IC)$. We do not
consider abductive interpretations that are not models of $IC$,
simply because $IC$ does not change during revision. Observe that
when $IC$ is empty, $\mathcal{S}$ is the set of all abductive
interpretations. Given a Horn knowledge base $KB$, and two abductive
interpretations $I_{1}$ and $I_{2}$ from $\mathcal{S}$, we can
compare how close these interpretations are to $KB$ by using an
order $\leq_{KB}$ among abductive interpretations in $\mathcal{S}$.
$I_{1}<_{KB}I_{2}$ iff $I_{1}\leq_{KB}I_{2}$ and $I_{2}\nleq_{KB}
I_{1}$.

Let $\mathcal{F}\subseteq \mathcal{S}$. An abductive interpretation
$I\in \mathcal{F}$ is minimal in $\mathcal{F}$ with respect to $\leq_{KB}$ if
there is no $I'\in \mathcal{F}$ s.t. $I'<_{KB}I$. Let,
$Min(\mathcal{F},\leq_{KB})=\{I~|~I~\text{\rm is minimal in}~\\
\mathcal{F} $ with respect to$~\leq_{KB}\}$.

For any Horn knowledge base KB, the following are desired properties
of $\leq_{KB}$:
\begin{enumerate}
\item[($\leq 1$)] (\it Pre-order)\rm $\leq_{KB}$ is a \it pre-order \rm, i.e. it is transitive and reflexive.
\item[($\leq 2$)] (\it Connectivity)\rm $\leq_{KB}$ is \it total \rm in $\mathcal{S}$, i.e. $\forall I_{1},I_{2}\in \mathcal{S}$:
either $I_{1}\leq_{KB} I_{2}$ or $I_{2}\leq_{KB} I_{1}$.
\item[($\leq 3$)] (\it Faithfulness)\rm $\leq_{KB}$ is \it faithful \rm to KB, i.e. $I \in Min(\mathcal{S},\leq_{KB})$ iff
$I \in Mod(KB)$.
\item[($\leq 4$)] (\it Minimality)\rm For any non-empty subset $\mathcal{F}$ of $\mathcal{S}$, $Min(\mathcal{F},\leq_{KB})$
is not empty.
\item[($\leq 5$)] (\it Preservance)] \rm For any Horn knowledge base KB', if $KB\equiv KB'$ then $\leq_{KB}=\leq_{KB'}$.
\end{enumerate}

Let $KB$ (and consequently $K$) be revised by a sentence $\alpha$,
and $\leq_{KB}$ be a rational order that satisfies $(\leq 1)$ to
$(\leq 5)$. Then the abductive models of the revised Horn knowledge
base are given precisely by: $Min(Mod(\{\alpha\}\cup
IC),\leq_{KB})$. Note that, this construction does not say what the
resulting K is, but merely says what should be the abductive models
of the new Horn knowledge base.


\subsubsection{Representation theorem}\hspace{0.5cm}

Now, we proceed to show that revision of $KB$ by $\alpha$, as
constructed above, satisfies all the rationality postulates
stipulated in the beginning of this section. This is formalized by
the following lemma.

\begin{lemma} \label{l9} Let $KB$ be a Horn knowledge base, $\leq_{KB}$ an order
among $\mathcal{S}$  that satisfies $(\leq 1)$ to $(\leq 5)$. Let a
revision operator $\dotplus$ be defined as: for any sentence
$\alpha$, $Mod(KB\dotplus \alpha)=Min(Mod(\{\alpha\}\cup
IC),\leq_{KB})$. Then $\dotplus$ satisfies all the rationality
postulates for revision $(\dotplus 1)$ to $(\dotplus 7)$.
\end{lemma}

\begin{proof}
\begin{enumerate} 
\item[$(\dotplus 1)$] $P'=P$ and $Ab'=Ab$ and $IC'=IC$\\ This is
satisfied obviously, since our construction does not touch $P$ and
$Ab$, and $IC$ follows from every abductive interpretation in
$Mod(KB\dotplus \alpha)$.
\item[$(\dotplus 2)$] $\alpha$ is accepted in $KB\dotplus \alpha$\\
Note that every abductive interpretation $M\in Mod(KB+\alpha)$ is a
model of $\alpha$. Hence $\alpha$ is accepted in $KB\dotplus
\alpha$.
\item[$(\dotplus 3)$] $\alpha$ is satisfiable and $P$-consistent with IC iff $KB\dotplus \alpha$ is
$P$-consistent.\\
If part: If $KB\dotplus \alpha$ is $P$-consistent , then
$Mod(KB\dotplus \alpha)$ is not empty. This implies that
$Mod(\{\alpha\}\cup IC)$ is not empty, and hence $\alpha$ is
satisfiable and $P$-consistent with $IC$. \\ Only if part: If
$\alpha$ is satisfiable and P-consistent with $IC$, then
$Mod(\{\alpha\}\cup IC)$ is not empty, and $(\leq 4)$ ensures that
$Mod(KB\dotplus \alpha)$ is not empty. Thus, $KB\dotplus \alpha$ is
$P$-consistent.
\item[$(\dotplus 4)$] If $\neg \alpha$ is not accepted in $KB$, then $KB\dotplus \alpha\equiv
KB+\alpha$.\\
We have to establish that  $Min(Mod(\{\alpha\}\cup
IC),\leq_{KB})=Mod(KB)\cap Mod(\alpha)$. Since $\neg \alpha$ is not
accepted in KB, $Mod(KB)\cap Mod(\alpha)$ is not empty. The required
result follows immediately from the fact that $\leq_{KB}$ is
faithful to KB (i.e. satisfies $\leq 3$), which selects only and all
those models of $\alpha$ which are also models of KB.
\item[$(\dotplus 5)$] If $KB \equiv KB'$ and $\alpha\equiv \beta$ then $KB\dotplus \alpha=KB'\dotplus
\beta$\\
$(\leq 5)$ ensures that $\leq_{KB}=\leq_{KB'}$. The required result
follows immediately from this and the fact that
$Mod(\alpha)=Mod(\beta)$.
\item[$(\dotplus 6)$] $(KB\dotplus \alpha)+\beta$ implies $KB\dotplus (\alpha\land
\beta)$.\\
We consider this in two cases. When $\neg \beta$ is accepted in
$KB\dotplus \alpha$, $(KB\dotplus \alpha)+\beta$ is the set of all
sentences from $\mathcal{L}$, and the postulate follows immediately.
Instead when $\neg \beta$ is not accepted in $KB\dotplus \alpha$,
this postulates coincides with the next one.

\item[$(\dotplus 7)$] If $\neg \beta$ is not accepted in $KB\dotplus \alpha$, then $KB\dotplus (\alpha\land \beta)$ implies  $(KB\dotplus
\alpha)+\beta$.\\
Together with the second case of previous postulate, we need to show
that $KB\dotplus (\alpha\land \beta)=(KB\dotplus \alpha)+\beta$. In
other words, we have to establish that  $Min(Mod(\{\alpha\land
\beta\}\cup IC),\leq_{KB})=Mod(KB\dotplus \alpha)\cap Mod(\beta)$.
For the sake of simplicity, let us represent $Min(Mod(\{\alpha\land
\beta\}\cup IC),\leq_{KB})$ by P, and $Mod(KB\dotplus \alpha)\cap
Mod(\beta)$, which is the same as $Min(Mod(\{\alpha\}\cup
IC),\leq_{KB})\cap Mod(\beta)$, by Q. The required result is
obtained in two parts:
\end{enumerate}
\begin{enumerate} 
\item[1)] $\forall$ (abductive interpretation)M: if $M\in P$, then $M\in
Q$\\
Obviously $M\in Mod(\beta)$. Assume that $M\notin
Min(Mod(\{\alpha\}\cup IC),\leq_{KB})$. This can happen in two
cases, and we show that both the cases lead to contradiction. \\
Case A: No model of $\beta$ is selected by $\leq _{KB}$ from
$Mod(\{\alpha\}\cup IC)$. But this contradicts our initial condition
that $\neg \beta$ is not accepted in $KB\dotplus \alpha$.\\
Case B: Some model, say $M'$, of $\beta$ is selected by $\leq_{KB}$
from $Mod(\{\alpha\}\cup IC)$. Since M is not selected, it follows
that $M'<_{KB}M$. But then this contradicts our initial assumption
that $M\in P$. So, $P\subseteq Q$.
\item[2)] $\forall$ (abductive interpretation)M: if $M\in Q$, then $M\in
P$\\
$M \in Q$ implies that $M$ is a model of both $\alpha$ and $\beta$,
and $M$ is selected by $\leq_{KB}$ from $Mod(\{\alpha\}\cup IC)$.
Note that $Mod(\{\alpha\land \beta\}\cup IC)\subseteq
Mod(\{\alpha\}\cup IC)$. Since $M$ is selected by $\leq_{KB}$ in a
bigger set (i.e. $Mod(\{\alpha\}\cup IC)$), $\leq_{KB}$ must select
$M$ from its subset $Mod(\{\alpha\land \beta\}\cup IC)$ also. Hence
$Q\subseteq P$.
\end{enumerate}
\end{proof}

But, that is not all. Any rational revision of $KB$ by $\alpha$,
that satisfies all the rationality postulates, can be constructed by
our construction method, and this is formalized below.

\begin{lemma} \label{l10} Let $KB$ be a Horn knowledge base and $\dotplus$ a revision
operator that satisfies all the rationality postulates for revision
$(\dotplus 1)$ to $(\dotplus 7)$. Then, there exists an order
$\leq_{KB}$ among $\mathcal{S}$, that satisfies $(\leq 1)$ to $(\leq
5)$, and for any sentence $\alpha$, $Mod(KB\dotplus \alpha)$ is
given in $Min(Mod(\{\alpha\}\cup IC),\leq_{KB})$.\it

\begin{proof} Let us construct an order $\leq_{KB}$ among
interpretations in $\mathcal{S}$ as follows: For any two abductive
interpretations $I$ and $I'$ in $\mathcal{S}$, define $I\leq_{KB}I'$
iff either $I\in Mod(KB)$ or $I \in Mod(KB\dotplus form(I,I'))$,
where $form(I,I')$ stands for sentence whose only models are $I$ and
$I'$. We will show that $\leq_{KB}$ thus constructed satisfies
$(\leq 1)$ to $(\leq 5)$ and $Min(Mod(\{\alpha\}\cup
IC),\leq_{KB})=Mod(KB\dotplus \alpha)$.

First, we show that $Min(Mod(\{\alpha\}\cup
IC),\leq_{KB})=Mod(KB\dotplus \alpha)$.Suppose $\alpha$ is not
satisfiable, i.e. $Mod(\alpha)$ is empty, or $\alpha$ does not
satisfy $IC$, then there are no abductive models of $\{\alpha\}\cup
IC$, and hence $Min(Mod(\{\alpha\}\cup IC),\leq_{KB})$ is empty.
From $(\dotplus 3)$, we infer that $Mod(KB\dotplus \alpha)$ is also
empty. When $\alpha$ is satisfiable and $\alpha$ satisfies $IC$, the
required result is obtained in two parts:

\begin{enumerate}
\item[1)] If $I\in Min(Mod(\{\alpha\}\cup IC),\leq_{KB})$, then $I\in Mod(KB\dotplus
\alpha)$\\
Since $\alpha$ is satisfiable and consistent with $IC$, $(\dotplus
3)$ implies that there exists at least one model, say $I'$, for
$KB\dotplus \alpha$. From $(\dotplus 1)$, it is clear that $I'$ is a
model of $IC$, from $(\dotplus 2)$ we also get that $I'$ is a model
of $\alpha$, and consequently $I\leq_{KB}I'$ (because $I\in
Min(Mod(\{\alpha\}\cup IC),\leq_{KB})$). Suppose $I\in Mod(KB)$,
then $(\dotplus 4)$ immediately gives $I \in Mod(KB\dotplus
\alpha)$. If not, from our definition of $\leq_{KB}$, it is clear
that $I \in Mod(KB\dotplus form(I,I'))$. Note that $\alpha \land
form(I,I')\equiv form(I,I')$, since both $I$ and $I'$ are models of
$\alpha$. From $(\dotplus 6)$ and $(\dotplus 7)$, we get
$Mod(KB\dotplus \alpha)\cap\{I,I'\}=Mod(KB\dotplus form(I,I'))$.
Since $I \in Mod(KB\dotplus form(I,I'))$, it immediately follows
that $I\in Mod(KB\dotplus \alpha)$.
\item[2)] If $I\in Mod(KB\dotplus \alpha)$, then $I \in Min(Mod(\{\alpha\}\cup
IC),\leq_{KB})$.\\
From $(\dotplus 1)$ we get $I$ is a model of $IC$, and from
$(\dotplus 2)$, we obtain $I\in Mod(\alpha)$. Suppose $I\in
Mod(KB)$, then from our definition of $\leq_{KB}$, we get
$I\leq_{KB}I'$, for any other model $I'$ of $\alpha$ and $IC$, and
hence $I \in Min(Mod(\{\alpha\}\cup IC),\leq_{KB})$. Instead, if $I$
is not a model of $KB$, then, to get the required result, we should
show that $I \in Mod(KB\dotplus form(I,I'))$, for every model $I'$
of $\alpha$  and $IC$. As we have observed previously, from
$(\dotplus 6)$ and $(\dotplus 7)$, we get $Mod(KB\dotplus
\alpha)\cap \{I,I'\}=Mod(KB\dotplus form(I,I'))$. Since $I\in
Mod(KB\dotplus \alpha)$, it immediately follows that $I \in
Mod(KB\dotplus form(I,I'))$. Hence $I\leq_{KB} I'$ for any model
$I'$ of $\alpha$ and $IC$, and consequently, $I\in
Min(Mod(\{\alpha\}\cup IC),\leq_{KB})$.
\end{enumerate}

Now we proceed to show that the order $\leq_{KB}$ among
$\mathcal{S}$, constructed as per our definition, satisfies all the
rationality axioms $(\leq 1)$ to $(\leq 5)$.

\begin{enumerate} 
\item[$(\leq 1)$] $\leq_{KB}$ is a pre-order.\\
Note that we need to consider only abductive interpretations from
$\mathcal{S}$. From $(\dotplus 2)$ and $(\dotplus 3)$, we have
$Mod(KB\dotplus form(I,I'))=\{I\}$, and so $I \leq_{KB}I$. Thus
$\leq_{KB}$ satisfies reflexivity. let $I_{1}\in Mod(IC)$ and
$I_{2}\notin Mod(IC)$. Clearly, it is possible that two
interpretations $I_{1}$ and $I_{2}$ are not models of $KB$, and
$Mod(KB\dotplus form(I_{1},I_{2})) =\{I_{1}\}$. So, $I_{1}\leq_{KB}
I_{2}$ does not necessarily imply $I_{2}\leq_{KB}I_{1}$, and thus
$\leq_{KB}$ satisfies anti-symmetry.

To show the transitivity, we have to prove that $I_{1}\leq_{KB}
I_{3}$, when $I_{1}\leq_{KB} I_{2}$ and $I_{2}\leq_{KB} I_{3}$ hold.
Suppose $I_{1}\in Mod(KB)$, then  $I_{1}\leq_{KB} I_{3}$ follows
immediately from our definition of $\leq_{KB}$. On the other case,
when $I_{1}\notin Mod(KB)$, we first observe that $I_{1}\in
Mod(KB\dotplus form(I_{1},I_{2}))$, which follows from definition of
$\leq_{KB}$ and $I_{1}\leq_{KB} I_{2}$. Also observe that
$I_{2}\notin Mod(KB)$. If $I_2$ were a model of $KB$, then it
follows from $(\dotplus 4)$ that $Mod(KB\dotplus
form(I_1,I_2))=Mod(KB)\cap \{I_1,I_2\}=\{I_2\}$, which is a
contradiction, and so $I_2\notin Mod(KB)$. This, together with
$I_{2}\leq_{KB}I_{3}$, implies that $I_{2}\in Mod(KB\dotplus
form(I_{2},I_{3}))$. Now consider $Mod(KB+form(I_{1},I_{2},I_{3}))$.
Since $\dotplus $ satisfies $(\dotplus 2)$ and $(\dotplus 3)$, it
follows that this is a non-empty subset of $\{I_{1},I_{2},I_{3}\}$.
We claim that $Mod(KB\dotplus form(I_{1},I_{2},I_{3}))\cap
\{I_{1},I_{2}\}$ can not be empty. If it is empty, then it means
that $Mod(KB\dotplus form(I_{1},I_{2},I_{3}))=\{I_{3}\}$. Since
$\dotplus $ satisfies $(\dotplus 6)$ and $(\dotplus 7)$, this
further implies that $Mod(KB\dotplus
form(I_{2},I_{3}))=Mod(KB\dotplus
form(I_{1},I_{2},I_{3}))\cap\{I_{2},I_{3}\}=\{I_{3}\}$. This
contradicts our observation that $I_{2}\in Mod(KB\dotplus
form(I_{2},\-I_{3}))$, and so $Mod(KB\dotplus
form(I_{1},I_{2},I_{3}))\cap \{I_{1},I_{2}\}$ can not be empty.
Using $(\dotplus 6)$ and $(\dotplus 7)$ again, we get
$Mod(KB\dotplus form(I_{1},I_{2}))=Mod(KB\dotplus
form(I_{1},I_{2},I_{3}))\cap\{I_{1},I_{2}\}$. Since we know that
$I_{1} \in Mod(KB\dotplus form(I_{1},I_{2}))$, it follows that
$I_{1}\in Mod(KB\dotplus form(I_{1},I_{2},I_{3}))$. From $(\dotplus
6)$ and $(\dotplus 7)$ we also get $Mod(KB\dotplus
form(I_{1},I_{3}))=Mod(KB+form(I_{1},I_{2},I_{3}))\cap
\{I_{1},I_{3}\}$, which clearly implies that $I_{1} \in
Mod(KB\dotplus form(I_{1},I_{3}))$. From our definition of
$\leq_{KB}$, we now obtain $I_{1}\leq_{KB}I_{3}$. Thus,  $\leq_{KB}$
is a pre-order.
\item[$(\leq 2)$] $\leq_{KB}$ is total.\\
Since $\dotplus $ satisfies $(\dotplus 2)$ and $(\dotplus 3)$, for
any two abductive interpretations $I$ and $I'$ in $\mathcal{S}$, it
follows that $Mod(KB\dotplus form(I,I'))$ is a non-empty subset of
$\{I,I'\}$. Hence, $\leq_{KB}$ is total.
\item[$(\leq 3)$] $\leq_{KB}$ is faithful to $KB$.\\
From our definition of $\leq_{KB}$, it follows that $\forall
I_{1},I_{2} \in Mod(KB):I_{1}<_{KB}I_{2}$ does not hold. Suppose
$I_{1}\in Mod(KB)$ and $I_{2}\notin Mod(KB)$. Then, we have
$I_{1}\leq_{KB}I_{2}$. Since $\dotplus $ satisfies $(\dotplus 4)$,
we also have $Mod(KB\dotplus form(I_{1},I_{2}))=\{I_{1}\}$. Thus,
from our definition of $\leq_{KB}$, we can not have
$I_{2}\leq_{KB}I_{1}$. So, if $I_{1}\in Mod(KB)$ and $I_{2}\notin
Mod(KB)$, then $I_{1}<_{KB}I_{2}$ holds. Thus, $\leq_{KB}$ is
faithful to $KB$.
\item[$(\leq 4)$] For any non-empty subset $\mathcal{F}$ of $\mathcal{S}$, $Min(\mathcal{F},\leq_{KB})$
is not empty.\\ Let $\alpha$ be a sentence such that
$Mod(\{\alpha\}\cup IC)=\mathcal{F}$. We have already shown that
$Mod(KB\dotplus \alpha)=Min(\mathcal{F},\leq_{KB})$. Since,
$\dotplus $ satisfies $(\dotplus 3)$, it follows that
$Mod(KB\dotplus \alpha)$ is not empty, and thus
$Min(\mathcal{F},\leq_{KB})$ is not empty.
\item[$(\leq 5)$] If $KB\equiv KB'$, then $\leq_{KB}=\leq_{KB'}$.\\
This follows immediately from the fact that $\dotplus $ satisfies
$(\dotplus 5)$.
\end{enumerate}
\end{proof}
\end{lemma}

Thus, the order among interpretations $\leq_{KB}$, constructed as
per our definition, satisfies $(\leq 1)$ to $(\leq 5)$, and
$Mod(KB\dotplus \alpha)=Min(Mod(\{\alpha\}\cup IC),\leq_{KB}).$

So, we have a one to one correspondence between the axiomatization
and the construction, which is highly desirable, and this is
summarized by the following \it representation theorem. \rm

\begin{theorem} \label{T17} Let $KB$ be revised by $\alpha$, and $KB\dotplus
\alpha$ be obtained by the construction discussed above. Then,
$\dotplus$ is a revision operator iff it satisfies all the
rationality postulates $(\dotplus 1)$ to $(\dotplus 7)$.
\end{theorem}

\begin{proof} Follows from Lemma \ref{l9} and Lemma \ref{l10}
\end{proof}

\subsubsection{Contraction}\hspace{0.5cm}

Contraction of a sentence from a Horn knowledge base $KB$ is studied
in the same way as that of revision. We first discuss the
rationality of change during contraction and proceed to provide a
construction for contraction using duality between revision and
contraction.


\subsubsection{Rationality Postulates}\hspace{0.5cm}

Let $KB=<P,Ab,IC,K>$ be contracted by a sentence $\alpha$ to result
in a new Horn knowledge base $KB\dot{-}\alpha=<P',Ab',IC',K'>$.

\begin{enumerate} 
\item[$(\dot{-}1)$] (\it Inferential Constancy)\rm $P'=P$ and $Ab'=Ab$ and $IC'=IC$.
\item[$(\dot{-}2)$] (\it Success)\rm If $\alpha \notin
Cn_{P}(KB)$, then $\alpha$ is not accepted in $KB\dot{-}\alpha$,
i.e. if $\alpha$ is not true in all the abductive interpretations,
then $\alpha$ is not true in all abductive interpretations in
$Mod(KB\dot{-}\alpha)$.
\item[$(\dot{-}3)$](\it Inclusion)\rm $\forall$ (belief)
$\beta$:if $\beta$ is accepted in $KB\dot{-}\alpha$, then $\beta$ is
accepted in $KB$, i.e. $Mod(KB)\subseteq Mod(KB\dot{-}\alpha)$.
\item[$(\dot{-}4)$](\it Vacuity)\rm If $\alpha$ is not
accepted in $KB$, then $KB\dot{-}\alpha=KB$, i.e. if $\alpha$ is not
true in all the abductive models of $KB$, then
$Mod(KB\dot{-}\alpha)=Mod(KB)$.
\item[$(\dot{-}5)$](\it Recovery)\rm $(KB\dot{-}\alpha)+\alpha$
implies $KB$, i.e. $Mod(KB\dot{-}\alpha)\cap Mod(\alpha)\subseteq
Mod(KB)$.
\item[$(\dot{-}6)$](\it Preservation)\rm If $KB\equiv KB'$ and
$\alpha\equiv \beta$, then $KB\dot{-}\alpha=KB'\dot{-}\beta$, i.e.
if $Mod(KB)=Mod(KB')$ and $Mod(\alpha)=Mod(\beta)$, then
$Mod(KB\dot{-}\alpha)=Mod(KB'\dot{-}\beta)$.
\item[$(\dot{-}7)$] (\it Conjunction 1) \rm$KB\dot{-}(\alpha\land \beta)$
implies $KB\dot{-}\alpha\cap KB\dot{-}\beta$, i.e.
$Mod(KB\dot{-}(\alpha\land \beta))\subseteq Mod(KB\dot{-}\alpha)\cup
Mod(KB\dot{-}\beta)$.
\item[$(\dot{-}8)$] (\it Conjunction 2)\rm  If $\alpha$ is not accepted in
$KB\dot{-}(\alpha\land \beta)$, then $KB\dot{-}\alpha$ implies
$KB\dot{-}(\alpha\land \beta)$, i.e. if $\alpha$ is not true in all
the models of $KB\dot{-}(\alpha\land \beta)$, then
$Mod(KB\dot{-}\alpha)\subseteq Mod(KB\dot{-}(\alpha\land \beta))$.
\end{enumerate}

Before providing a construction for contraction, we wish to study
the duality between revision and contraction. The Levi and Harper
identities still holds in our case, and is discussed in the sequel.


\subsubsection{Relationship between contraction and
revision}
\hspace{0.5cm}

Contraction and revision are related to each other. Given a
contraction function $\dot{-}$, a revision function $\dotplus$ can
be obtained as follows:
$$\text{(\it Levi
Identity)}~~~~~Mod(KB\dotplus \alpha)=Mod(KB\dot{-}\neg\alpha)\cap
Mod(\alpha)$$ The following theorem formally states that Levi
identity holds in our approach.

\begin{theorem}  \label{T18} Let $\dot{-}$ be a contraction operator that
satisfies all the rationality postulates $(\dot{-}1)$ to
$(\dot{-}8)$. Then, the revision function $\dotplus$, obtained from
$\dot{-}$ using the Levi Identity, satisfies all the rationality
postulates $(\dotplus 1)$ to $(\dotplus 7)$.
\end{theorem}

\begin{proof} 
Let $KB=<P,Ab,IC,K>$ be contracted by a sentence $\alpha$ to result
in a new Horn knowledge base $KB\dot{-}\alpha=<P',Ab',IC',K'>$.
\begin{enumerate} 
\item[$(\dotplus1)$] (\it Inferential constancy) $P'=P$ and $Ab'=Ab$,$IC'=IC$.
\item[$(\dotplus 2)$] (\it Success)\rm $\alpha$ is accepted in $KB\dotplus\alpha$ , i.e. $\alpha$ is true in all models of $(Mod(KB\dot{-}\neg\alpha)\cap
Mod(\alpha))$.
\item[$(\dotplus 3)$] (\it Consistency) \rm $\alpha$ is satisfiable and $P$-consistent with IC iff $KB\dotplus \alpha$ is P-consistent,
i.e. $(Mod(\{\alpha\}\cup IC))$ is not empty iff $Mod(KB\dot{-}\neg\alpha)\cap
Mod(\alpha)$ is not empty.
\item[$(\dotplus 4)$] (\it Vacuity) \rm If $\neg \alpha$ is not accepted in KB, then $KB\dotplus \alpha\equiv KB+\alpha$, i.e. if $\alpha$
is not false in all models of KB, then $(Mod(KB\dotplus
\alpha))=(Mod(KB\dot{-}\neg\alpha)\cap
Mod(\alpha)) \cap Mod(\alpha)$
\item[$(\dotplus 5)$] (\it Preservation)\rm If $KB \equiv KB'$ and $\alpha\equiv \beta$, then $KB\dotplus\alpha\equiv KB'\dotplus\beta$, i.e.
if $Mod(KB)=Mod(KB')$ and $Mod(\alpha)=Mod(\beta)$, then
$(Mod(KB\dot{-}\neg\alpha)\cap
Mod(\alpha))=(Mod(KB\dot{-}\neg\beta)\cap
Mod(\beta))$.
\item[$(\dotplus 6)$] (\it Extended Vacuity 1)\rm $(KB\dotplus\alpha)+\beta$ implies $((KB\dot{-}\neg\alpha)\cap (\alpha)) \land \beta)$, i.e.
$(Mod(KB\dotplus \alpha)\cap Mod(\beta))\subseteq
Mod(KB\dotplus(\alpha \land \beta))$.
\item[$(\dotplus 7)$] (\it Extended Vacuity 2)\rm If $\neg \beta$ is not accepted in $(KB\dotplus \alpha)$, then
$((KB\dot{-}\neg\alpha)\cap (\alpha)) \land \beta)$ implies $(KB\dotplus \alpha)+\beta$,
i.e. if $\beta$ is not false in all models of $KB\dotplus\alpha$,
then $Mod(KB\dotplus(\alpha\land\beta))\subseteq
(Mod(KB\dotplus\alpha)\cap Mod(\beta))$.
\end{enumerate}
\end{proof}

Similarly, a contraction function $\dot{-}$ can be constructed using
the given revision function $\dotplus $ as follows:
$$\text{(\it Harper
Identity)}~~~~~Mod(KB\dot{-}\alpha)=Mod(KB)\cup Mod(KB\dotplus \neg
\alpha)$$

\begin{theorem}  \label{T19} Let $\dotplus$ be a revision operator that
satisfies all the rationality postulates $(\dotplus 1)$ to
$(\dotplus 7)$. Then, the contraction function $\dot{-}$, obtained
from $\dotplus $ using the Harper Identity, satisfies all the
rationality postulates $(\dot{-}1)$ to $(\dot{-}8)$.
\end{theorem}

\begin{proof} Let $KB=<P,Ab,IC,K>$ be contracted by a sentence $\alpha$ to result
in a new Horn knowledge base $KB\dot{-}\alpha=<P',Ab',IC',K'>$.\begin{enumerate}
\item[$(\dot{-}1)$] (\it Inferential Constancy)\rm $P'=P$ and $Ab'=Ab$ and $IC'=IC$.
\item[$(\dot{-}2)$] (\it Success)\rm If $\alpha \notin
Cn_{P}(KB)$, then $\alpha$ is not accepted in $KB\dot{-}\alpha$,
i.e. if $\alpha$ is not true in all the abductive interpretations,
then $\alpha$ is not true in all abductive interpretations in
$Mod(KB)\cup Mod(KB\dotplus \neg
\alpha)$.
\item[$(\dot{-}3)$](\it Inclusion)\rm $\forall$ (belief)
$\beta$:if $\beta$ is accepted in $KB\dot{-}\alpha$, then $\beta$ is
accepted in $KB$, i.e. $Mod(KB)\subseteq (Mod(KB)\cup Mod(KB\dotplus \neg
\alpha))$.
\item[$(\dot{-}4)$](\it Vacuity)\rm If $\alpha$ is not
accepted in $KB$, then $KB\dot{-}\alpha=KB$, i.e. if $\alpha$ is not
true in all the abductive models of $KB$, then
$Mod(KB\dot{-}\alpha)=Mod(KB)$.
\item[$(\dot{-}5)$](\it Recovery)\rm $(KB\dot{-}\alpha)+\alpha$
implies $KB$, i.e. $(Mod(KB\dot{-}\alpha)\cap Mod(\alpha))\subseteq
Mod(KB)$.
\item[$(\dot{-}6)$](\it Preservation)\rm If $KB\equiv KB'$ and
$\alpha\equiv \beta$, then $KB\dot{-}\alpha=KB'\dot{-}\beta$, i.e.
if $Mod(KB)=Mod(KB')$ and $Mod(\alpha)=Mod(\beta)$, then
$(Mod(KB)\cup Mod(KB\dotplus \neg
\alpha))=(Mod(KB)\cup Mod(KB\dotplus \neg
\beta))$.
\item[$(\dot{-}7)$] (\it Conjunction 1) \rm$KB\dot{-}(\alpha\land \beta)$
implies $KB\dot{-}\alpha\cap KB\dot{-}\beta$, i.e.
$Mod(KB\dot{-}(\alpha\land \beta))\subseteq (Mod(KB)\cup Mod(KB\dotplus \neg
\alpha))\cup
(Mod(KB)\cup Mod(KB\dotplus \neg
\beta))$.
\item[$(\dot{-}8)$] (\it Conjunction 2)\rm  If $\alpha$ is not accepted in
$KB\dot{-}(\alpha\land \beta)$, then $KB\dot{-}\alpha$ implies
$KB\dot{-}(\alpha\land \beta)$, i.e. if $\alpha$ is not true in all
the models of $KB\dot{-}(\alpha\land \beta)$, then
$(Mod(KB)\cup Mod(KB\dotplus \neg
\alpha))\subseteq Mod(KB\dot{-}(\alpha\land \beta))$.
\end{enumerate}

\end{proof}


\subsubsection{Construction}\hspace{0.5cm}

Given the construction for revision, based on order among
interpretation in $\mathcal{S}$, a construction for contraction can
be provided as:

$$Mod(KB\dot{-}\alpha)=Mod(KB)\cup Min(Mod(\{\neg\alpha\}\cup
IC),\leq_{KB}),$$ where $\leq_{KB}$ is the relation among
interpretations in $\mathcal{S}$ that satisfies the rationality
axioms $(\leq 1)$ to $(\leq 5)$. As in the case of revision, this
construction says what should be the models of the resulting Horn
knowledge base, and does not explicitly say what the resulting Horn
knowledge base is.


\subsubsection{Representation theorem}\hspace{0.5cm}

Since the construction for contraction is based on a rational
contraction for revision, the following lemma and theorem follow
obviously.

\begin{lemma} \label{l11} Let $KB$ be a Horn knowledge base, $\leq_{KB}$ an order
among $\mathcal{S}$ that satisfies $(\leq 1)$ to $(\leq 5)$. Let a
contraction operator $\dot{-}$ be defined as: for any sentence
$\alpha$, $Mod(KB\dot{-}\alpha)=Mod(KB)\cup
Min(Mod(\{\neg\alpha\}\cup IC),\leq_{KB})$. Then $\dot{-}$ satisfies
all the rationality postulates for contraction $(\dot{-}1)$ to
$(\dot{-}8)$.

\begin{proof} Follows from Theorem  \ref{T17} and Theorem \ref{T19}
\end{proof}
\end{lemma}

\begin{lemma} \label{l12} Let $KB$ be a Horn knowledge base and
$\dot{-}$ a contraction operator that satisfies all the rationality
postulates for contraction $(\dot{-}1)$ to $(\dot{-}8)$. Then, there
exists an order $\leq_{KB}$ among $\mathcal{S}$, that
satisfies$(\leq 1)$ to $(\leq 5)$, and for any sentence $\alpha$,
$Mod(KB\dot{-}\alpha)$ is given as $Mod(KB)\cup
Min(Mod(\{\neg\alpha\}\cup IC),\leq_{KB})$.

\begin{proof} Follows from Theorem \ref{T18} and Theorem \ref{T19}.
\end{proof}
\end{lemma}

\begin{theorem} \label{T20} Let $KB$ be contracted by $\alpha$,
and $KB\dot{-}\alpha$ be obtained by the construction discussed
above. Then $\dot{-}$ is a contraction operator iff it satisfies all
the rationality postulates $(\dot{-}1)$ to $(\dot{-}8)$.

\begin{proof} Follows from Lemma \ref{l11} and Lemma \ref{l12}

\end{proof}
\end{theorem}

\subsection{Relationship with the coherence approach of AGM}
Given Horn knowledge base $KB=<P,Ab,IC,K>$ represents a belief set
$KB^{\bullet}$ that is closed under $Cn_{P}$. We have defined how
$KB$ can be expanded, revised, or contracted. The question now is:
\it does our foundational approach (with respect to classical first-order logic)
on $KB$ coincide with coherence approach (with respect to our consequence
operator $Cn_{P}$) of $AGM$ on $KB^{\bullet}$?\rm~ There is a
problem in answering this question (similar practical problem Baral
\& Zhang 2005) , since our approach, we require $IC$ to be
immutable, and only the current knowledge $K$ is allowed to change.
On the contrary, $AGM$ approach treat every sentence in
$KB^{\bullet}$ equally, and can throw out sentences from
$Cn_{P}(IC)$. One way to solve this problem is to assume that
sentences in $Cn_{P}(IC)$ are more entrenched than others. However,
one-to-one correspondence can be established, when $IC$ is empty.
The key is our consequence operator $Cn_{P}$, and in the following,
we show that coherence approach of $AGM$ with this consequence
operator, is exactly same as our foundational approach, when $IC$ is
empty.

\subsubsection{Expansion}
\hspace{0.5cm}

Expansion in $AGM$ (Alchourron et al. 1985b) framework is defined as
$KB\#\alpha=Cn_{P}(KB^{\bullet}\cup\{\alpha\})$, is is easy to see
that this is equivalent to our definition of expansion (when $IC$ is
empty), and is formalized below.

\begin{theorem}  \label{T21} Let $KB+\alpha$ be an expansion of $KB$ by $\alpha$. Then $(KB+\alpha)^{\bullet}=KB\#\alpha.$

\begin{proof} By our definition of expansion, $(KB+\alpha)^{\bullet}=Cn_{P}(IC \cup
K\cup\{\alpha\})$, which is clearly the same set as
$Cn_{P}(KB^{\bullet}\cup \{\alpha\})$.
\end{proof}
\end{theorem}

\subsubsection{Revision}\hspace{0.5cm}

$AGM$ puts forward rationality postulates $(*1)$ to $(*8)$ to be
satisfied by a revision operator on $KB^{\bullet}$. reproduced
below:
\begin{enumerate}
\item[(*1)] (\it Closure) \rm $KB^{\bullet}*\alpha$ is a belief set.
\item[(*2)] (\it Success) $\alpha\in KB^{\bullet}*\alpha$.
\item[(*3)] (Expansion 1) $KB^{\bullet}*\alpha \subseteq KB^{\bullet}\# \alpha$.
\item[(*4)] (Expansion 2) \rm If $\neg \alpha \notin KB^{\bullet},$ then
$KB^{\bullet}\#\alpha\subseteq KB^{\bullet}*\alpha$.
\item[(*5)] (\it Consistency)\rm $KB^{\bullet}*\alpha$ is inconsistent iff
$\vdash \neg \alpha$.
\item[(*6)] (\it Preservation) \rm If $\vdash \alpha \leftrightarrow
\beta$, then $KB^{\bullet}*\alpha=KB^{\bullet}*\beta$.
\item[(*7)] (\it Conjunction 1) $KB^{\bullet}*(\alpha\land \beta)\subseteq
(KB^{\bullet}*\alpha)\#\beta$.
\item[(*8)] (Conjunction 2) \rm If $\neg \beta \notin KB^{\bullet}*\alpha$,
then,$(KB^{\bullet}*\alpha)\#\beta\subseteq
KB^{\bullet}*(\alpha\land \beta)$.
\end{enumerate}

The equivalence between our approach and $AGM$ approach is brought
out by the following two theorems.

\begin{theorem}  \label{T22} Let $KB$ a Horn knowledge base with an empty $IC$ and $\dotplus $
be a revision function that satisfies all the rationality postulates
$(\dotplus 1)$ to $(\dotplus 7)$. Let a revision operator $*$ on
$KB^{\bullet}$ be defined as: for any sentence $\alpha$,
$KB^{\bullet}*\alpha=(KB\dotplus \alpha)^{\bullet}$. The revision
operator
*, thus defined satisfies all the $AGM$-postulates for revision
$(*1)$ to $(*8)$.

\begin{proof}
\hspace{0.5cm}
\begin{enumerate}
\item[(*1)] $KB^{\bullet}*\alpha$ is a belief set.\\
This follows immediately, because $(KB\dotplus \alpha)^{\bullet}$ is
closed with respect to $Cn_{P}$.
\item[(*2)] $\alpha\in KB^{\bullet}*\alpha$.\\
This follows from the fact that $\dotplus $ satisfies $(\dotplus
2)$.
\item[(*3)] $KB^{\bullet}*\alpha \subseteq KB^{\bullet}\# \alpha$.\\
\item[(*4)] If $\neg \alpha \notin KB^{\bullet},$ then
$KB^{\bullet}\#\alpha\subseteq KB^{\bullet}*\alpha$.\\
 These two postulates follow
from $(\dotplus 4)$ and theorem \ref{T21}.
\item[(*5)]$KB^{\bullet}*\alpha$ is inconsistent iff
$\vdash \neg \alpha$.\\
This follows from from $(\dotplus 3)$ and our assumption that $IC$
is empty.
\item[(*6)] If $\vdash \alpha \leftrightarrow
\beta$, then $KB^{\bullet}*\alpha=KB^{\bullet}*\beta$.\\
This corresponds to $(\dotplus 5)$.
\item[(*7)] $KB^{\bullet}*(\alpha\land \beta)\subseteq
(KB^{\bullet}*\alpha)\#\beta$. This follows from $(\dotplus 6)$ and
theorem \ref{T21}.
\item[(*8)] If $\neg \beta \notin KB^{\bullet}*\alpha$,
then,$(KB^{\bullet}*\alpha)\#\beta\subseteq
KB^{\bullet}*(\alpha\land \beta)$.\\
This follows from $(\dotplus 7)$ and theorem \ref{T21}
\end{enumerate}
\end{proof}
\end{theorem}

\begin{theorem}  \label{T23} Let $KB$ a Horn knowledge base with an empty $IC$ and
* a revision operator that satisfies all the $AGM$-postulates $(*1)$
to $(*8)$. Let a revision function $+$ on $KB$ be defined as: for
any sentence $\alpha$, $(KB\dotplus
\alpha)^{\bullet}=KB^{\bullet}*\alpha$. The revision function $+$,
thus defined, satisfies all the rationality postulates $(\dotplus
1)$ to $(\dotplus 7)$.

\begin{proof}
\hspace{0.5cm}
\begin{enumerate}
\item[$(\dotplus 1)$] $P,Ab$ and $IC$ do not change.\\
Obvious.
\item[$(\dotplus 2)$] $\alpha$ is accepted in $KB\dotplus \alpha$.\\
Follows from $(^*2)$.
\item[$(\dotplus 3)$] If $\alpha$ is satisfiable and consistent with $IC$,
then $KB\dotplus \alpha$ is consistent.\\
Since we have assumed $IC$ to be empty, this directly corresponds to
$(^*5)$.
\item[$(\dotplus 4)$] If $\neg\alpha$ is not accepted in $KB$, then
$KB\dotplus \alpha\equiv KB+\alpha$.\\ Follows from $(^*3)$ and
$(^*4)$.
\item[$(+5)$] If $KB\equiv KB'$ and $\alpha\equiv \beta$, then
$KB\dotplus \alpha\equiv KB'\dotplus \beta$.\\
Since $KB\equiv KB'$ they represent same belief set, i.e.
$KB^{\bullet}=KB'^{\bullet}$. Now, this postulate follows
immediately from $(^*6)$.
\item[$(\dotplus 6)$] $(KB\dotplus \alpha)+\beta$ implies $KB\dotplus (\alpha\land
\beta)$.\\
Corresponds to $(^*7)$.
\item[$(\dotplus 7)$] If $\neg \beta$ is not accepted in $KB\dotplus \alpha$, then
$KB\dotplus (\alpha\land \beta)$ implies $(KB\dotplus
\alpha)+\beta$.\\
Corresponds to $(^*8)$.
\end{enumerate}
\end{proof}
\end{theorem}

\subsubsection{Contraction}\hspace{0.5cm}

$AGM$ puts forward rationality postulates $(-1)$ to $(-8)$ to be
satisfied by a contraction operator on closed set $KB^{\bullet}$,
reproduced below:

\begin{enumerate}
\item[$(-1)$] (\it Closure) \rm $KB^{\bullet}-\alpha$ is a belief set.
\item[$(-2)$] (\it Inclusion) \rm $KB^{\bullet}-\alpha\subseteq KB^{\bullet}$.
\item[$(-3)$] (\it Vacuity) \rm If $\alpha \notin KB^{\bullet}$, then
$KB^{\bullet}-\alpha=KB^{\bullet}$.
\item[$(-4)$] (\it Success) \rm If $\nvdash \alpha$, then $\alpha
\notin KB^{\bullet}-\alpha$.
\item[$(-5)$] (\it Preservation) \rm If $\vdash \alpha
\leftrightarrow \beta$, then
$KB^{\bullet}-\alpha=KB^{\bullet}-\beta$.
\item[$(-6)$] (\it Recovery) \rm $KB^{\bullet}\subset (KB^{\bullet}-\alpha)+\alpha$.
\item[$(-7)$] (\it Conjunction 1)\rm $KB^{\bullet}-\alpha\cap
KB^{\bullet}-\beta\subseteq KB^{\bullet}-(\alpha\land\beta)$.
\item[$(-8)$] (\it Conjunction 2) \rm If $\alpha \notin
KB^{\bullet}-(\alpha\land\beta)$, then
$KB^{\bullet}-(\alpha\land\beta)\subseteq KB^{\bullet}-\alpha$.
\end{enumerate}

As in the case of revision, the equivalence is brought out by the
following theorems. Since contraction is constructed in terms of
revision, these theorems are trivial.

\begin{corollary}  \label{T24} Let $KB$ be a Horn knowledge base with an empty $IC$ and $\dot{-}$
be a contraction function that satisfies all the rationality
postulates $(\dot{-}1)$ to $(\dot{-}8)$. Let a contraction operator
$-$ on $KB^{\bullet}$ be defined as: for any sentence $\alpha$,
$KB^{\bullet}-\alpha=(KB\dot{-}\alpha)^{\bullet}$. The contraction
operator $-$, thus defined, satisfies all the $AGM$ - postulates for
contraction $(-1)$ to $(-8)$.

%
\begin{proof} Follows from Theorem \ref{T18} and Theorem \ref{T22}

\end{proof}

\end{corollary}

\begin{corollary} \label{T25} Let $KB$ be a Horn knowledge base with an empty $IC$ and $-$ be a
contraction operator that satisfies all the $AGM$- postulates $(-1)$
to $(-8)$. Let a contraction function  $\dot{-}$ on $KB$ be defined
as: for any sentence $\alpha$,
$(KB\dot{-}\alpha)^{\bullet}=KB^{\bullet}-\alpha$. The contraction
function $\dot{-}$, thus defined, satisfies all the rationality
postulates $(\dot{-}1)$ to $(\dot{-}8)$.

%

\begin{proof} Follows from Theorem \ref{T19} and Theorem \ref{T23}

\end{proof}
\end{corollary}

\subsection{Realizing Horn knowledge base dynamics using abductive explanations}
In this section, we explore how belief dynamics can be realized in
practice (see (Aravindan \& Dung 1994), (Aravindan 1995) and
(Bessant et al. 1998)). Here, we will see how revision can be
implemented based on the construction using models of revising
sentence and an order among them. The notion of abduction proves to
be useful and is explained in the sequel.

Let $\alpha$ be a sentence in $\mathcal{L}$. An \it abductive
explanation \rm for $\alpha$ with respect to $KB$ is a set of abductive literals
\footnotemark \footnotetext{An abductive literal is either an
abducible $A$ from $Ab$, or its negation $\neg A$.} $\Delta$ s.t.
$\Delta$ consistent with $IC$ and $\Delta \models_{P}\alpha$ (that
is $\alpha \in Cn_P(\Delta))$. Further $\Delta$ is said to be \it
minimal \rm iff no proper subset of $\Delta$ is an abductive
explanation for $\alpha$.

The basic idea to implement revision of a Horn knowledge base $KB$
by a sentence $\alpha$, is to realize $Mod(\{\alpha\}\cup IC)$ in
terms of abductive explanations for $\alpha$ with respect to $KB$. We first
provide a useful lemma.

\begin{definition} Let $KB$ be a Horn knowledge base, $\alpha$ a sentence, and
$\Delta_{1}$ and $\Delta_{2}$ be two minimal abductive explanations
for $\alpha$ with respect to $KB$. Then, the \it disjunction \rm of $\Delta_{1}$
and $\Delta_{2}$, written as $\Delta_{1}\lor \Delta_{2}$, is given
as:
$$\Delta_{1}\lor\Delta_{2}=(\Delta_{1}\cap \Delta_{2})\cup
\{\alpha\lor\beta| \alpha \in\Delta_{1}\backslash \Delta_{2}~
\text{and}~\beta\in \Delta_{2}\backslash \Delta_{1}\}.$$ Extending
this to $\Delta^{\bullet}$, a set of minimal abductive explanations
for $\alpha$ with respect to $KB$, $\lor\Delta^{\bullet}$ is given by the
disjunction of all elements of $\Delta^{\bullet}$.
\end{definition}

\begin{lemma} \label{l13} Let $KB$ be a Horn knowledge base, $\alpha$ a sentence,and
$\Delta_{1}$ and $\Delta_{2}$ be two minimal abductive explanations
for $\alpha$ with respect to $KB$.
Then,$Mod(\Delta_{1}\lor\Delta_{2})=Mod(\Delta_{1})\cup
Mod(\Delta_{2})$.

\begin{proof} First we show that every model of $\Delta_{1}$ is a model of
$\Delta_{1}\lor\Delta_{2}$. Clearly, a model $M$ of $\Delta_{1}$
satisfies all the sentences in $(\Delta_{1}\cap \Delta_{2})$. The
other sentences in $(\Delta_{1}\lor\Delta_{2})$ are of the form
$\alpha \lor\beta$, where $\alpha$ is from $\Delta_{1}$ and $\beta$
is from $\Delta_{2}$. Since $M$ is a model of $\Delta_{1}$, $\alpha$
is true in $M$, and hence all such sentences are satisfied by $M$.
Hence $M$ is a model of $\Delta_{1}\lor\Delta_{2}$ too. Similarly,
it can be shown that every model of $\Delta_{2}$ is a model of
$\Delta_{1}\lor\Delta_{2}$ too.

Now, it remains to be shown that every model $M$ of
$\Delta_{1}\lor\Delta_{2}$ is either a model of $\Delta_{1}$ or a
model of $\Delta_{2}$. We will now show that if $M$ is not a model
of $\Delta_{2}$, then it must be a model of $\Delta_{1}$. Since $M$
satisfies all the sentences in $(\Delta_{1}\cap\Delta_{2})$, we need
only to show that $M$ also satisfies all the sentences in
$\Delta_{1}\backslash\Delta_{2}$. For every element $\alpha \in
\Delta_{1}\backslash\Delta_{2}$: there exists a subset of
$(\Delta_{1}\lor \Delta_{2})$, $\{\alpha\lor \beta|\beta \in
\Delta_{2}\backslash\Delta_{2}\}$. $M$ satisfies all the sentences
in this subset. Suppose $M$ does not satisfy $\alpha$, then it must
satisfy all $\beta\in \Delta_{1}\backslash\Delta_{2}$. This implies
that $M$ is a model of $\Delta_{2}$, which is a contradictory to our
assumption. Hence $M$ must satisfy $\alpha$, and thus a model
$\Delta_{1}$. Similarly, it can be shown that $M$ must be a model of
$\Delta_{2}$ if it is not a model of $\Delta_{1}$.
\end{proof}
\end{lemma}

As one would expect, all the models of revising sentence $\alpha$
can be realized in terms abductive explanations for $\alpha$, and
the relationship is precisely stated below.

\begin{lemma} \label{l14} Let $KB$ be a Horn knowledge base, $\alpha$ a sentence, and
$\Delta^{\bullet}$ the set of all minimal abductive explanations for
$\alpha$ with respect to $KB$. Then $Mod(\{\alpha\}\cup
IC)=Mod(\lor\Delta^{\bullet})$.

\begin{proof} It can be easily verified that every model $M$ of a minimal
abductive explanation is also a model of $\alpha$. Since every
minimal abductive explanation satisfies $IC$, $M$ is a model of
$\alpha\cup IC$. It remains to be shown that every model $M$ of
$\{\alpha\}\cup IC$ is a model of  one of the minimal abductive
explanations for $\alpha$ with respect to $KB$. This can be verified by
observing that a minimal abductive explanation for $\alpha$ with respect to $KB$
can be obtained from $M$.
\end{proof}
\end{lemma}

Thus, we have a way to generate all the models of $\{\alpha\}\cup
IC$, and we just need to select a subset of this based on an order
that satisfies $(\leq 1)$ to $(\leq 5)$. Suppose we have such an
order that satisfies all the required postulates, then this order
can be mapped to a particular set of abductive explanations for
$\alpha$ with respect to $KB$. This is stated precisely in the following
theorem. An important implication of this theorem is that there is
no need to compute all the abductive explanations for $\alpha$ with respect to
$KB$. However, it does not say which abductive explanations need to
be computed.

\begin{theorem} \label{T26} Let $KB$ be a Horn knowledge base, and $\leq_{KB}$ be an order
among abductive interpretations in $\mathcal{S}$ that satisfies all
the rationality axioms $(\leq 1)$ to $(\leq 5)$. Then, for every
sentence $\alpha$, there exists $\Delta^{\bullet}$ a set of minimal
abductive explanations for $\alpha$ with respect to $KB$, s.t.
$Min(Mod(\{\alpha\}\cup IC),\leq_{KB})$ is a subset of $Mod(\lor
\Delta^{\bullet})$, and this does not hold for any proper subset of
$\Delta^{\bullet}$.

\begin{proof} From Lemma \ref{l13} and Lemma \ref{l14}, it is clear that
$Mod(\{\alpha\}\cup IC)$ is the union of all the models of all
minimal abductive explanations of $\alpha$ with respect to $KB$. $Min$ selects a
subset of this, and the theorem follows immediately. .
\end{proof}
\end{theorem}

The above theorem \ref{T26}, is still not very useful in realizing
revision. We need to have an order among all the interpretations
that satisfies all the required axioms, and need to compute all the
abductive explanations for $\alpha$ with respect to $KB$. The need to compute
all abductive explanations arises from the fact that the converse of
the above theorem does not hold in general. This scheme requires an
universal order $\leq$, in the sense that same order can be used for
any Horn knowledge base. Otherwise, it would be necessary to specify
the new order to be used for further modifying $(KB\dotplus\alpha)$.
However, even if the order can be worked out, it is not desirable to
demand all abductive explanations of $\alpha$ with respect to $KB$ be computed.
So, it is desirable to work out, when the converse of the above
theorem is true. The following theorem says that, suppose $\alpha$
is rejected in $KB$, then revision of $KB$ by $\alpha$ can be worked
out in terms of some abductive explanations for $\alpha$ with respect to $KB$.

\begin{theorem} \label{T27} Let $KB$ be a Horn knowledge base, and a revision function
$\dotplus$ be defined as: for any sentence $\alpha$ that is rejected
in $KB$, $Mod(KB\dotplus\alpha)$ is a non-empty subset of
$Mod(\lor\Delta^{\bullet})$, where $\Delta^{\bullet}$ is a set of
all minimal abductive explanations for $\alpha$ with respect to $KB$. Then,
there exists an order $\leq_{KB}$ among abductive interpretations in
$\mathcal{S}$, s.t. $\leq_{KB}$ satisfies all the rationality axioms
$(\leq 1)$ to $(\leq 5)$ and $Mod(KB\dotplus\alpha)=
Min(Mod(\{\alpha\}\cup IC),\leq_{KB})$.

\begin{proof} Let us construct an order $\leq_{KB}$ among
interpretations in $\mathcal{S}$ as follows: For any two abductive
interpretations $I$ and $I'$ in $\mathcal{S}$, define $I\leq_{KB}I'$
iff either $I\in Mod(\lor\Delta^{\bullet})$ or $I \in Mod(KB\dotplus form(I,I'))$,
where $form(I,I')$ stands for sentence whose only models are $I$ and
$I'$. We will show that $\leq_{KB}$ thus constructed satisfies
$(\leq 1)$ to $(\leq 5)$ and $Min(Mod(\{\alpha\}\cup
IC),\leq_{KB})=Mod(KB\dotplus \alpha)$.

First, we show that $Min(Mod(\{\alpha\}\cup
IC),\leq_{KB})=Mod(KB\dotplus \alpha)$.Suppose $\alpha$ is not
satisfiable, i.e. $Mod(\alpha)$ is empty, or $\alpha$ does not
satisfy $IC$, then there are no abductive models of $\{\alpha\}\cup
IC$, and hence $Min(Mod(\{\alpha\}\cup IC),\leq_{KB})$ is empty.
From $(\dotplus 3)$, we infer that $Mod(KB\dotplus \alpha)$ is also
empty. When $\alpha$ is satisfiable and $\alpha$ satisfies $IC$, the
required result is obtained in two parts:

\begin{enumerate}
\item[1)] If $I\in Min(Mod(\{\alpha\}\cup IC),\leq_{KB})$, then $I\in Mod(KB\dotplus
\alpha)$\\
Since $\alpha$ is satisfiable and consistent with $IC$, $(\dotplus
3)$ implies that there exists at least one model, say $I'$, for
$KB\dotplus \alpha$. From $(\dotplus 1)$, it is clear that $I'$ is a
model of $IC$, from $(\dotplus 2)$ we also get that $I'$ is a model
of $\alpha$, and consequently $I\leq_{KB}I'$ (because $I\in
Min(Mod(\{\alpha\}\cup IC),\leq_{KB})$). Suppose $I\in (Mod(\lor\Delta^{\bullet}))$,
then $(\dotplus 4)$ immediately gives $I \in Mod(KB\dotplus
\alpha)$. If not, from our definition of $\leq_{KB}$, it is clear
that $I \in Mod(KB\dotplus form(I,I'))$. Note that $\alpha \land
form(I,I')\equiv form(I,I')$, since both $I$ and $I'$ are models of
$\alpha$. From $(\dotplus 6)$ and $(\dotplus 7)$, we get
$Mod(KB\dotplus \alpha)\cap\{I,I'\}=Mod(KB\dotplus form(I,I'))$.
Since $I \in Mod(KB\dotplus form(I,I'))$, it immediately follows
that $I\in Mod(KB\dotplus \alpha)$.
\item[2)] If $I\in Mod(KB\dotplus \alpha)$, then $I \in Min(Mod(\{\alpha\}\cup
IC),\leq_{KB})$.\\
From $(\dotplus 1)$ we get $I$ is a model of $IC$, and from
$(\dotplus 2)$, we obtain $I\in Mod(\alpha)$. Suppose $I\in
(Mod(\lor\Delta^{\bullet}))$, then from our definition of $\leq_{KB}$, we get
$I\leq_{KB}I'$, for any other model $I'$ of $\alpha$ and $IC$, and
hence $I \in Min(Mod(\{\alpha\}\cup IC),\leq_{KB})$. Instead, if $I$
is not a model of $KB$, then, to get the required result, we should
show that $I \in Mod(KB\dotplus form(I,I'))$, for every model $I'$
of $\alpha$  and $IC$. As we have observed previously, from
$(\dotplus 6)$ and $(\dotplus 7)$, we get $Mod(KB\dotplus
\alpha)\cap \{I,I'\}=Mod(KB\dotplus form(I,I'))$. Since $I\in
Mod(KB\dotplus \alpha)$, it immediately follows that $I \in
Mod(KB\dotplus form(I,I'))$. Hence $I\leq_{KB} I'$ for any model
$I'$ of $\alpha$ and $IC$, and consequently, $I\in
Min(Mod(\{\alpha\}\cup IC),\leq_{KB})$.
\end{enumerate}

Every model of $Mod(KB\dotplus\alpha)$ is strictly minimal than all other
interpretations. It is easy to verify that such a pre-order
satisfies $(\leq 1)$ to $(\leq 5)$. In particular, since $\alpha$ is
rejected in $KB$, $(\leq 3)$ faithfulness is satisfied, and since
non-empty subset of $Mod(\lor\Delta^{\bullet})$ is selected, $(\leq
4)$ is also satisfied.
\end{proof}
\end{theorem}

An important corollary of this theorem is that, revision of $KB$ by
$\alpha$ can be realized just by computing \it one \rm abductive
explanation of $\alpha$ with respect to $KB$, and is stated below.

\begin{corollary} Let $KB$ be a Horn knowledge base, and a revision
function $\dotplus$ be defined as: for any sentence $\alpha$ that is
rejected in $KB$, $Mod(KB\dotplus\alpha)$ is a non-empty subset of
$Mod(\Delta)$, where $\Delta$ is an abductive explanations for
$\alpha$ with respect to $KB$. Then, there exists an order $\leq_{KB}$ among
abductive interpretations in $\mathcal{S}$, s.t. $\leq_{KB}$
satisfies all the rationality axioms $(\leq 1)$ to $(\leq 5)$ and
$Mod(KB\dotplus\alpha)= Min(Mod(\{\alpha\}\cup IC),\leq_{KB})$.
\end{corollary}

The precondition that $\alpha$ is rejected in $KB$ is not a serious
limitation in various applications such as database updates and
diagnosis, where close world assumption is employed to infer
negative information. For example, in diagnosis it is generally
assumed that all components are functioning normally, unless
otherwise there is specific information against it. Hence, a Horn
knowledge base in diagnosis either accepts or rejects normality of a
component, and there is no "don't know" third state. In other words,
in these applications the Horn knowledge base is assumed to be
complete. Hence, when such a complete Horn knowledge base is revised
by $\alpha$, either $\alpha$ is already accepted in $KB$ or rejected
in $KB$, and so the above scheme works fine.

\section{Related Works}
We begin by recalling previous work on view deletion. Aravindan
(Aravindan \& Dung 1994), (Aravindan 1995), defines a contraction
operator in view deletion with respect to a set of formulae or
sentences using Hansson's (Hansson 1997a) belief change. Similar to
our approach (Delhibabu \& Lakemeyer 2013, Delhibabu \& Behrend, 2014, Delhibabu 2014a, Delhibabu 2014b) he focused on set of formulae or sentences in
knowledge base revision for view update with respect to insertion and deletion
and formulae are considered at the same level. Aravindan proposed
different ways to change knowledge base via only database deletion,
devising particular postulate which is shown to be necessary and
sufficient for such an update process.

Our Horn knowledge base consists of two parts, immutable part and
updatable part, but our focus is on minimal change computations.
The related works are, Eiter (Eiter \& Makino
2007), Langlois (Langlois et al. 2008) and Delgrande (Delgrande \&
Peppas 2011) are focus on Horn revision with different
perspectives like prime implication, logical closure and belief
level. Segerberg (Segerberg 1998) defined a new modeling technique
for belief revision in terms of irrevocability on prioritized
revision. Hansson constructed five types of
non-prioritized belief revision. Makinson (Makinson 1997) developed
dialogue form of revision AGM. Papini (Papini 2000) defined a new
version of knowledge base revision. In this paper, we considered the immutable part as a Horn
clause (Ferm{\'e} $\&$ Hansson 2001 shown shielded contraction similar to immutable part, the success postulate does not hold in general; some non-tautological beliefs are shielded from contraction and cannot be given up. Shielded contraction has close connections with credibility limited revision shown Hansson et al 2001) and the updatable part as an atom (literal). Knowledge bases have a set of integrity constraints. 

Hansson's (Hansson 1997a) kernel change is related to abductive
method. Aliseda's (Aliseda 2006) book on abductive reasoning is one
of the motivation step. Christiansen's (Christiansen \& Dahl 2009)
work on dynamics of abductive logic grammars exactly fits our
minimal change (insertion and deletion). Wrobel's (Wrobel 1995)
definition of first order theory revision was helpful to frame our
algorithm.

On the other hand, we are dealing with view update problem. Keller's
(Keller 1985) thesis is motivation of the view update problem. There
are many papers related to the view update problem (for example, the
recent survey paper on view update by Chen and Liao (Chen \& Liao
2010) and the survey paper on view update algorithms by Mayol and
Teniente (Mayol \& Teniente 1999). More similar to our work is the
paper presented by Bessant (Bessant et al. 1998), which introduces a
local search-based heuristic technique that empirically proves to be
often viable, even in the context of very large propositional
applications. Laurent (Laurent et al. 1998), considers updates in a
deductive database in which every insertion or deletion of a fact
can be performed in a deterministic way.

Furthermore, and at a first sight more related to our work, some
work has been done on "core-retainment" (Hansson 1991) in the model
of language splitting introduced by Parikh (Parikh 1999). More
recently, Doukari (Doukari et al. 2008), \"{O}z\c{c}ep
(\"{O}z\c{c}ep 2012) and Wu (Wu et al. 2011) applied similar ideas
for dealing with knowledge base dynamics. These works represent
motivation step for our future work. Second, we are dealing with how
to change minimally in the theory of "principle of minimal change",
but current focus is on finding second best abductive explanation
(Liberatore \& Schaerf 2004 and 2012), 2-valued minimal hypothesis
for each normal program (Pinto \& Pereira 2011). Our work reflected
in the current trends on Ontology systems and description logics (Qi
and Yang (Qi \& Yang 2008) and Kogalovsky (Kogalovsky 2012)).
Finally, when we presented Horn knowledge base change in abduction
framework, we did not talk about compilability and complexity (see
the works of Liberatore (Liberatore 1997) and Zanuttini (Zanuttini
2003)).

The significance of our work can be summarized in the following:
\begin{itemize}
\item To define a new kind of revision operator on
  Horn knowledge base and obtain axiomatic characterization for it.
  \item To propose new generalized revision algorithm for
  Horn knowledge base dynamics, and study its connections with kernel change and
  abduction procedure.
   \item To develop a new view insertion algorithm for databases.
   \item To design a new view update algorithm for stratifiable Deductive Database (DDB),
 using an axiomatic method based on Hyper tableaux and magic sets.
   \item To study an abductive framework for Horn knowledge base dynamics.
   \item To present a comparative study of view update
algorithms and integrity constraint.
   \item Finally, to shown connection between belief update versus
  database update.
\end{itemize}


\section{Conclusion and remarks}

~~~~~~~~~~The main contribution of this research is to provide a
link between theory of belief dynamics and concrete applications
such as view updates in databases. We argued for generalization of
belief dynamics theory in two respects: to handle certain part of
knowledge as immutable; and dropping the requirement that belief
state be deductively closed. The intended generalization was
achieved by introducing the concept of Horn knowledge base dynamics
and generalized revision for the same. Further, we studied the
relationship between Horn knowledge base dynamics  and abduction
resulting in a generalized algorithm for revision based on abductive
procedures. The successfully demonstrated how Horn knowledge
base dynamics provide an axiomatic characterization for update
an literals to a stratifiable (definite) deductive database.

In bridging the gap between belief dynamics and view updates, we
observe that a balance has to be achieved between computational
efficiency and rationality. While rationally attractive notions of
generalized revision prove to be computationally inefficient, the
rationality behind efficient algorithms based on incomplete trees is
not clear at all. From the belief dynamics point of view, we may
have to sacrifice some postulates, vacuity, to gain computational
efficiency. Further weakening of relevance has to be explored, to
provide declarative semantics for algorithms based on incomplete
trees.

On the other hand, from the database side, we should explore various
ways of optimizing the algorithms that would comply with the
proposed declarative semantics. We believe that partial deduction
and loop detection techniques, will play an important role in
optimizing algorithms. Note that, loop detection could be carried
out during partial deduction, and complete SLD-trees can be
effectively constructed wrt a partial deduction (with loop check) of
a database, rather than wrt database itself. Moreover, we would
anyway need a partial deduction for optimization of query
evaluation.

We have presented two variants of an algorithm for update a view
atom from a definite database. The key idea of this approach is to
transform the given database into a logic program in
such a way that updates can be read off from the models of this
transformed program. We have also shown that this algorithm is
rational in the sense that it satisfies the rationality postulates
that are justified from philosophical angle. In the second variant,
where materialized view is used for the transformation, after
generating a hitting set and removing corresponding $EDB$ atoms, we
easily move to the new materialized view. An obvious way is to
recompute the view from scratch using the new $EDB$ (i.e., compute
the Least Herbrand Model of the new updated database from scratch)
but it is certainly interesting to look for more efficient methods.

Though we have discussed only about view updates, we believe that
Horn knowledge base dynamics can also be applied to other
applications such as view maintenance, diagnosis, and we plan to
explore it further (see works (Biskup 2012) and (Caroprese et al.
2012)). Still, a lot of developments are possible, for improving
existing operators or for defining new classes of change operators. In the relation of Horn KB revision with hitting set as to be describe similar construction for description logic. In particular assuming the T-Box to the hitting set and the rest with the A-Box (Delgrande, JP \& Wassermann 2013). 
We did not talk about complexity (see the works of Liberatore
((Liberatore 1997) and (Liberatore \& Schaerf 2004)), Caroprese
(Caroprese 2012), Calvanese's (Calvanese 2012), and Cong (Cong et
al. 2012)). In this thesis answer impotent question for experimental people that is,\emph{any real life application for AGM in 25 year theory?} (Ferme \&
Hansson 2011). The revision and update are more challenging in
logical view update problem (database theory), we extended the
theory to combine our results similar in the Konieczny's (Konieczny 2011) and Nayak's (Nayak 2011).

\begin{sidewaystable}
\centering
\begin{turn}{180}
\begin{tabular}{|l|l|l|l|l|l|l|l|l|l|l|l|l|l|l|l|l|l|l|l|}\hline
\multirow{3}{*}{\bf Method}&\multicolumn{4}{c}{\bf Problem}&
\multicolumn{4}{|c|}{\bf Database
  schema}&\multicolumn{2}{c}{\bf Update req.}&\multicolumn{3}{|c|}{\bf Mechanism}&\multicolumn{3}{c}{\bf Update Change}&
\multicolumn{3}{|c|}{\bf Solutions}\\
\cline{2-20}

 &\multirow{2}{*}{Type}&View&IC&Run/&Def.&\multirow{2}{*}{View}&IC&Kind of
&\multirow{2}{*} {Mul.}&Update&Tech-&Base&User&\multirow{2}{*}{Type}&~Base~&View&\multirow{2}{*}{Axiom}&\multirow{2}{*}{Sound.}&\multirow{2}{*}{Complete.}\\
&&Update&Enforce.&Comp.&Lang.&&def.&IC&&Operat.&nique&Facts&Part.&&facts&def.&&&\\\hline

\multirow{2}{*}{\bf
\cite{Gue}}&\multirow{2}{*}{N}&\multirow{2}{*}{Yes}&\multirow{2}{*}{Check}&\multirow{2}{*}{Run}
&\multirow{2}{*}{Logic}&\multirow{2}{*}{Yes}&\multirow{2}{*}{Yes}&\multirow{2}{*}{Static}&\multirow{2}{*}{No}&\multirow{2}{*}{
$\iota~\delta$}&\multirow{2}{*}{
SLDNF}&\multirow{2}{*}{No}&\multirow{2}{*}{No}&\multirow{2}{*}{S}&\multirow{2}{*}{Yes}&\multirow{2}{*}{No}&\multirow{2}{*}{1-6,9}&\multirow{2}{*}{No}&Not\\
&&&&&&&&&&&&&&&&&&&proved\\\hline

\multirow{2}{*}{\bf
\cite{Kak}}&\multirow{2}{*}{N}&\multirow{2}{*}{Yes}&\multirow{2}{*}{Maintain}&\multirow{2}{*}{Run}
&\multirow{2}{*}{Logic}&\multirow{2}{*}{Yes}&\multirow{2}{*}{Yes}&\multirow{2}{*}{Static}&\multirow{2}{*}{Yes}&\multirow{2}{*}{
$\iota~\delta$}&\multirow{2}{*}{
SLDNF}&\multirow{2}{*}{No}&\multirow{2}{*}{No}&\multirow{2}{*}{S}&\multirow{2}{*}{Yes}&\multirow{2}{*}{No}&\multirow{2}{*}{---}&\multirow{2}{*}{No}&\multirow{2}{*}{No}\\
&&&&&&&&&&&&&&&&&&&\\\hline

\multirow{2}{*}{\bf
\cite{Kuc}}&\multirow{2}{*}{S}&\multirow{2}{*}{Yes}&\multirow{2}{*}{Check}&\multirow{2}{*}{Run}
&\multirow{2}{*}{Logic}&\multirow{2}{*}{Yes}&\multirow{2}{*}{Yes}&\multirow{2}{*}{Static}&\multirow{2}{*}{Yes}&\multirow{2}{*}{
$\iota~\delta$}&\multirow{2}{*}{
---}&\multirow{2}{*}{Yes}&\multirow{2}{*}{No}&\multirow{2}{*}{SS}&\multirow{2}{*}{Yes}&\multirow{2}{*}{Yes}&\multirow{2}{*}{1-6,7}&Not&Not\\
&&&&&&&&&&&&&&&&&&proved&proved\\\hline

\multirow{2}{*}{\bf
\cite{Moer}}&\multirow{2}{*}{N}&\multirow{2}{*}{No}&\multirow{2}{*}{Maintain}&\multirow{2}{*}{Run}
&\multirow{2}{*}{Logic}&\multirow{2}{*}{Yes}&\multirow{2}{*}{Yes}&\multirow{2}{*}{Static}&\multirow{2}{*}{Yes}&\multirow{2}{*}{
$\iota~\delta$}&\multirow{2}{*}{
---}&\multirow{2}{*}{Yes}&\multirow{2}{*}{No}&\multirow{2}{*}{S}&\multirow{2}{*}{Yes}&\multirow{2}{*}{No}&\multirow{2}{*}{1-6,7}
&No&No\\
&&&&&&&&&&&&&&&&&&proved&proved\\\hline

\multirow{2}{*}{\bf
\cite{Gup}}&\multirow{2}{*}{S}&\multirow{2}{*}{Yes}&Check&Comp.
&Relation.&\multirow{2}{*}{No}&\multirow{2}{*}{No}&\multirow{2}{*}{Static}&\multirow{2}{*}{Yes}&\multirow{2}{*}{
$\iota~\delta~\chi$}&predef.&\multirow{2}{*}{Yes}&\multirow{2}{*}{No}&\multirow{2}{*}{G}&\multirow{2}{*}{Yes}&\multirow{2}{*}{No}&\multirow{2}{*}{---}
&\multirow{2}{*}{No}&\multirow{2}{*}{No}\\
&&&Maintain&Run&Logic&&&&&&Programs&&&&&&&&\\\hline

\multirow{2}{*}{\bf
\cite{Don}}&\multirow{2}{*}{N}&\multirow{2}{*}{Yes}&\multirow{2}{*}{Check}&\multirow{2}{*}{Run}
&\multirow{2}{*}{Logic}&\multirow{2}{*}{Yes}&\multirow{2}{*}{No}&\multirow{2}{*}{Static}&\multirow{2}{*}{Yes}&\multirow{2}{*}{
$\iota~\delta$}&predef&\multirow{2}{*}{Yes}&\multirow{2}{*}{No}&\multirow{2}{*}{S}&\multirow{2}{*}{Yes}&\multirow{2}{*}{No}&\multirow{2}{*}{1-6,7}&Not&\multirow{2}{*}{No}\\
&&&&&&&&&&&Programs&&&&&&&Proved&\\\hline

\multirow{2}{*}{\bf
\cite{Urp}}&\multirow{2}{*}{S}&\multirow{2}{*}{Yes}&Check&\multirow{2}{*}{Run}
&\multirow{2}{*}{Logic}&\multirow{2}{*}{Yes}&\multirow{2}{*}{Yes}&\multirow{2}{*}{Static}&\multirow{2}{*}{Yes}&\multirow{2}{*}{
$\iota~\delta~\chi$}&\multirow{2}{*}{
SLDNF}&\multirow{2}{*}{No}&\multirow{2}{*}{No}&\multirow{2}{*}{SS}&\multirow{2}{*}{Yes}&\multirow{2}{*}{No}&\multirow{2}{*}{1-6,7}&\multirow{2}{*}{Yes}&\multirow{2}{*}{No}\\
&&&Maintain&&&&&&&&&&&&&&&&\\\hline

\multirow{2}{*}{\bf
\cite{Gup1}}&\multirow{2}{*}{N}&\multirow{2}{*}{Yes}&\multirow{2}{*}{Maintain}&\multirow{2}{*}{Run}
&\multirow{2}{*}{Logic}&\multirow{2}{*}{Yes}&\multirow{2}{*}{No}&\multirow{2}{*}{Static}&\multirow{2}{*}{Yes}&\multirow{2}{*}{
$\iota~\delta$}&\multirow{2}{*}{Unfold}&\multirow{2}{*}{Yes}&\multirow{2}{*}{No}&\multirow{2}{*}{SS}&\multirow{2}{*}{Yes}&\multirow{2}{*}{No}&\multirow{2}{*}{1-6,9}&\multirow{2}{*}{Yes}&\multirow{2}{*}{Yes}\\
&&&&&&&&&&&&&&&&&&&\\\hline

\multirow{2}{*}{\bf
\cite{May}}&\multirow{2}{*}{N}&\multirow{2}{*}{Yes}&\multirow{2}{*}{Maintain}&Comp.
&\multirow{2}{*}{Logic}&\multirow{2}{*}{Yes}&\multirow{2}{*}{Yes}&Static&\multirow{2}{*}{Yes}&\multirow{2}{*}{
$\iota~\delta~\chi$}&\multirow{2}{*}{
SLDNF}&\multirow{2}{*}{Yes}&\multirow{2}{*}{No}&\multirow{2}{*}{S}&\multirow{2}{*}{Yes}&\multirow{2}{*}{No}&\multirow{2}{*}{1-6,9}&Not&Not\\
&&&&Run&&&&Dynamic&&&&&&&&&&proved&proved\\\hline

\multirow{2}{*}{\bf
\cite{Wut}}&\multirow{2}{*}{S}&\multirow{2}{*}{Yes}&\multirow{2}{*}{Maintain}&\multirow{2}{*}{Run}
&\multirow{2}{*}{Logic}&\multirow{2}{*}{Yes}&\multirow{2}{*}{Yes}&\multirow{2}{*}{Static}&\multirow{2}{*}{Yes}&\multirow{2}{*}{
$\iota~\delta$}&\multirow{2}{*}{
Unfold.}&\multirow{2}{*}{No}&\multirow{2}{*}{Yes}&\multirow{2}{*}{S}&\multirow{2}{*}{Yes}&\multirow{2}{*}{No}&\multirow{2}{*}{1-6,7}&Not&\multirow{2}{*}{No}\\
&&&&&&&&&&&&&&&&&&proved&\\\hline

\multirow{2}{*}{\bf
\cite{Arav1}}&\multirow{2}{*}{H}&\multirow{2}{*}{Yes}&\multirow{2}{*}{Check}&\multirow{2}{*}{Run}
&\multirow{2}{*}{Logic}&\multirow{2}{*}{Yes}&\multirow{2}{*}{Yes}&\multirow{2}{*}{Static}&\multirow{2}{*}{Yes}&\multirow{2}{*}{
$\iota~\delta$}&\multirow{2}{*}{
SLD}&\multirow{2}{*}{Yes}&\multirow{2}{*}{No}&\multirow{2}{*}{S}&\multirow{2}{*}{Yes}&\multirow{2}{*}{No}&\multirow{2}{*}{1-6,9}&\multirow{2}{*}{Yes}&\multirow{2}{*}{Yes}\\
&&&&&&&&&&&&&&&&&&&\\\hline

\multirow{2}{*}{\bf
\cite{Cer}}&\multirow{2}{*}{N}&\multirow{2}{*}{No}&\multirow{2}{*}{Maintain}&Comp
&Relation&\multirow{2}{*}{Yes}&\multirow{2}{*}{Limited}&\multirow{2}{*}{Static}&\multirow{2}{*}{Yes}&\multirow{2}{*}{
$\iota~\delta~\chi$}&\multirow{2}{*}{
Active}&\multirow{2}{*}{Yes}&\multirow{2}{*}{Yes}&\multirow{2}{*}{S}&\multirow{2}{*}{Yes}&\multirow{2}{*}{No}&\multirow{2}{*}{---}&\multirow{2}{*}{No}&\multirow{2}{*}{No}\\
&&&&Run&Logic&&&&&&&&&&&&&&\\\hline

\multirow{2}{*}{\bf
\cite{Ger}}&\multirow{2}{*}{N}&\multirow{2}{*}{No}&\multirow{2}{*}{Maintain}&Comp
&Relation&\multirow{2}{*}{No}&Flat&Static&\multirow{2}{*}{Yes}&\multirow{2}{*}{
$\iota~\delta~\chi$}&\multirow{2}{*}{
Active}&\multirow{2}{*}{Yes}&\multirow{2}{*}{Yes}&\multirow{2}{*}{S}&\multirow{2}{*}{Yes}&\multirow{2}{*}{No}&\multirow{2}{*}{---}&\multirow{2}{*}{No}&\multirow{2}{*}{No}\\
&&&&Run&Logic&&Limited&Dynamic&&&&&&&&&&&\\\hline

\multirow{2}{*}{\bf
\cite{Chen1}}&\multirow{2}{*}{H}&\multirow{2}{*}{Yes}&Check&Comp.
&O-O&Class
&\multirow{2}{*}{Limited}&\multirow{2}{*}{Static}&\multirow{2}{*}{Yes}&\multirow{2}{*}{
$\iota~\delta$}&\multirow{2}{*}{
Active}&\multirow{2}{*}{Yes}&\multirow{2}{*}{No}&\multirow{2}{*}{SS}&\multirow{2}{*}{Yes}&\multirow{2}{*}{No}&\multirow{2}{*}{1-6,9}&\multirow{2}{*}{No}&\multirow{2}{*}{Yes}\\
&&&Maintain&Run&&Att.&&&&&&&&&&&&&\\\hline

\multirow{2}{*}{\bf
\cite{Con}}&\multirow{2}{*}{N}&\multirow{2}{*}{Yes}&\multirow{2}{*}{Maintain}&\multirow{2}{*}{Run}
&\multirow{2}{*}{Logic}&\multirow{2}{*}{Yes}&Flat&\multirow{2}{*}{Static}&\multirow{2}{*}{Yes}&\multirow{2}{*}{
$\iota~\delta$}&\multirow{2}{*}{
Unfold.}&\multirow{2}{*}{Yes}&\multirow{2}{*}{No}&\multirow{2}{*}{S}&\multirow{2}{*}{Yes}&\multirow{2}{*}{No}&\multirow{2}{*}{1-6,9}&Not&\multirow{2}{*}{Yes}\\
&&&&&&&Limited&&&&&&&&&&&proved&\\\hline

\multirow{2}{*}{\bf
\cite{Lu1}}&\multirow{2}{*}{N}&\multirow{2}{*}{Yes}&\multirow{2}{*}{Maintain}&\multirow{2}{*}{Run}
&\multirow{2}{*}{Logic}&\multirow{2}{*}{Yes}&\multirow{2}{*}{Limited}&\multirow{2}{*}{Static}&\multirow{2}{*}{No}&\multirow{2}{*}{
$\iota~\delta$}&\multirow{2}{*}{
Active}&\multirow{2}{*}{Yes}&\multirow{2}{*}{No}&\multirow{2}{*}{SS}&\multirow{2}{*}{Yes}&\multirow{2}{*}{No}&\multirow{2}{*}{1-6,7}&\multirow{2}{*}{Yes}&Not\\
&&&&&&&&&&&&&&&&&&&proved\\\hline

\multirow{2}{*}{\bf
\cite{Ten}}&\multirow{2}{*}{N}&\multirow{2}{*}{Yes}&\multirow{2}{*}{Maintain}&Comp
&\multirow{2}{*}{Logic}&\multirow{2}{*}{Yes}&\multirow{2}{*}{Yes}&Static&\multirow{2}{*}{Yes}&\multirow{2}{*}{
$\iota~\delta$}&\multirow{2}{*}{
SLDNF}&\multirow{2}{*}{Yes}&\multirow{2}{*}{No}&\multirow{2}{*}{S}&\multirow{2}{*}{Yes}&\multirow{2}{*}{No}&\multirow{2}{*}{1-6,9}&\multirow{2}{*}{Yes}&\multirow{2}{*}{Yes}\\
&&&&Run&&&&Dynamic&&&&&&&&&&&\\\hline

\multirow{2}{*}{\bf
\cite{Sch}}&\multirow{2}{*}{S}&\multirow{2}{*}{Yes}&\multirow{2}{*}{Maintain}&\multirow{2}{*}{Comp}
&\multirow{2}{*}{Logic}&\multirow{2}{*}{No}&Flat&\multirow{2}{*}{Static}&\multirow{2}{*}{Yes}&\multirow{2}{*}{
$\iota~\delta$}&
predef&\multirow{2}{*}{---}&\multirow{2}{*}{Yes}&\multirow{2}{*}{G}&\multirow{2}{*}{No}&\multirow{2}{*}{Yes}&\multirow{2}{*}{---}&\multirow{2}{*}{No}&Not\\
&&&&&&&Limited&&&&Programs&&&&&&&&proved\\\hline

\end{tabular}
\end{turn}
\begin{turn}{180}
\text{\bf Tab. 1.\rm~Summary of view-update and integrity constraint
with our axiomatic method}
\end{turn}
\begin{turn}{180}
\textbf{Appendix A}
\end{turn}
\end{sidewaystable}

\begin{sidewaystable}
\centering
\begin{tabular}{|l|l|l|l|l|l|l|l|l|l|l|l|l|l|l|l|l|l|l|l|}\hline
\multirow{3}{*}{\bf Method}&\multicolumn{4}{c}{\bf Problem}&
\multicolumn{4}{|c|}{\bf Database
  schema}&\multicolumn{2}{c}{\bf Update req.}&\multicolumn{3}{|c|}{\bf Mechanism}&\multicolumn{3}{c}{\bf Update Change}&
\multicolumn{3}{|c|}{\bf Solutions}\\
\cline{2-20}

 &\multirow{2}{*}{Type}&View&IC&Run/&Def.&\multirow{2}{*}{View}&IC&Kind of
&\multirow{2}{*} {Mul.}&Update&Tech-&Base&User&\multirow{2}{*}{Type}&~Base~&View&\multirow{2}{*}{Axiom}&\multirow{2}{*}{Sound.}&\multirow{2}{*}{Complete.}\\
&&Update&Enforce.&Comp.&Lang.&&def.&IC&&Operat.&nique&Facts&Part.&&facts&def.&&&\\\hline

\multirow{2}{*}{\bf
\cite{Sta}}&\multirow{2}{*}{N}&\multirow{2}{*}{No}&\multirow{2}{*}{Maintain}&\multirow{2}{*}{Comp}
&\multirow{2}{*}{Logic}&\multirow{2}{*}{Yes}&\multirow{2}{*}{Limited}&\multirow{2}{*}{Static}&\multirow{2}{*}{Yes}&\multirow{2}{*}{
$\iota~\delta~\chi$}&
Predef&\multirow{2}{*}{Yes}&\multirow{2}{*}{No}&\multirow{2}{*}{S}&\multirow{2}{*}{Yes}&\multirow{2}{*}{No}&\multirow{2}{*}{1-6,7}&\multirow{2}{*}{Yes}&\multirow{2}{*}{No}\\
&&&&&&&&&&&Program&&&&&&&&\\\hline

\multirow{2}{*}{\bf
\cite{Arav2}}&\multirow{2}{*}{H}&\multirow{2}{*}{Yes}&\multirow{2}{*}{Check}&\multirow{2}{*}{Run}
&\multirow{2}{*}{Logic}&\multirow{2}{*}{Yes}&\multirow{2}{*}{Limited}&\multirow{2}{*}{Static}&\multirow{2}{*}{Yes}&\multirow{2}{*}{
$\iota~\delta$}&\multirow{2}{*}{
SLD}&\multirow{2}{*}{Yes}&\multirow{2}{*}{No}&\multirow{2}{*}{S}&\multirow{2}{*}{Yes}&\multirow{2}{*}{No}&\multirow{2}{*}{1-6,9}&\multirow{2}{*}{Yes}&\multirow{2}{*}{Yes}\\
&&&&&&&&&&&&&&&&&&&\\\hline

\multirow{2}{*}{\bf
\cite{Dec}}&\multirow{2}{*}{N}&\multirow{2}{*}{Yes}&\multirow{2}{*}{Maintain}&\multirow{2}{*}{Run}
&\multirow{2}{*}{Logic}&\multirow{2}{*}{Yes}&\multirow{2}{*}{Yes}&\multirow{2}{*}{Static}&\multirow{2}{*}{Yes}&\multirow{2}{*}{
$\iota~\delta$}&\multirow{2}{*}{
SLDNF}&\multirow{2}{*}{No}&\multirow{2}{*}{No}&\multirow{2}{*}{S}&\multirow{2}{*}{Yes}&\multirow{2}{*}{No}&\multirow{2}{*}{1-6,7}&\multirow{2}{*}{No}&Not\\
&&&&&&&&&&&&&&&&&&&Proved\\\hline

\multirow{2}{*}{\bf
\cite{Lobo}}&\multirow{2}{*}{N}&\multirow{2}{*}{Yes}&\multirow{2}{*}{Maintain}&\multirow{2}{*}{Run}
&\multirow{2}{*}{Logic}&\multirow{2}{*}{Yes}&Flat&\multirow{2}{*}{Static}&\multirow{2}{*}{Yes}&\multirow{2}{*}{
$\iota~\delta$}&\multirow{2}{*}{
Unfold}&\multirow{2}{*}{No}&\multirow{2}{*}{Yes}&\multirow{2}{*}{G}&\multirow{2}{*}{Yes}&\multirow{2}{*}{No}&\multirow{2}{*}{1-6,7}&Not&\multirow{2}{*}{No}\\
&&&&&&&Limited&&&&&&&&&&&proved&\\\hline

\multirow{2}{*}{\bf
\cite{Yang}}&\multirow{2}{*}{H}&\multirow{2}{*}{No}&\multirow{2}{*}{Maintain}&Comp.
&\multirow{2}{*}{Relation}&\multirow{2}{*}{Yes}&\multirow{2}{*}{Limited}&Static&\multirow{2}{*}{Yes}&\multirow{2}{*}{
$\iota~\delta~\chi$}&\multirow{2}{*}{
Unfold}&\multirow{2}{*}{Yes}&\multirow{2}{*}{No}&\multirow{2}{*}{S}&\multirow{2}{*}{Yes}&\multirow{2}{*}{No}&\multirow{2}{*}{1-6,7}&Not&Not\\
&&&&Run&&&&Dynamic&&&&&&&&&&proved&proved\\\hline

\multirow{2}{*}{\bf
\cite{Maa}}&\multirow{2}{*}{N}&\multirow{2}{*}{No}&Maintain&Comp
&\multirow{2}{*}{Logic}&\multirow{2}{*}{No}&Flat&Static&\multirow{2}{*}{Yes}&\multirow{2}{*}{
$\iota~\delta$}&\multirow{2}{*}{
Active}&\multirow{2}{*}{Yes}&\multirow{2}{*}{No}&\multirow{2}{*}{G}&\multirow{2}{*}{No}&\multirow{2}{*}{No}&\multirow{2}{*}{---}&\multirow{2}{*}{No}&\multirow{2}{*}{No}\\
&&&Restore&Run&&&Limited&Dynamic&&&&&&&&&&&\\\hline

\multirow{2}{*}{\bf
\cite{Sch1}}&\multirow{2}{*}{N}&\multirow{2}{*}{No}&\multirow{2}{*}{Maintain}&Comp
&\multirow{2}{*}{Relation}&\multirow{2}{*}{No}&Flat&\multirow{2}{*}{Static}&\multirow{2}{*}{Yes}&\multirow{2}{*}{
$\iota~\delta$}&\multirow{2}{*}{
Active}&\multirow{2}{*}{Yes}&\multirow{2}{*}{No}&\multirow{2}{*}{S}&\multirow{2}{*}{No}&\multirow{2}{*}{No}&\multirow{2}{*}{---}&\multirow{2}{*}{No}&\multirow{2}{*}{No}\\
&&&&Run&&&Limited&&&&&&&&&&&&\\\hline

\multirow{2}{*}{\bf
\cite{Lu}}&\multirow{2}{*}{N}&\multirow{2}{*}{Yes}&\multirow{2}{*}{Check}&\multirow{2}{*}{Run}
&\multirow{2}{*}{Logic}&\multirow{2}{*}{Yes}&\multirow{2}{*}{Limited}&\multirow{2}{*}{Static}&\multirow{2}{*}{Yes}&\multirow{2}{*}{
$\iota~\delta$}&\multirow{2}{*}{
SLD}&\multirow{2}{*}{Yes}&\multirow{2}{*}{No}&\multirow{2}{*}{S}&\multirow{2}{*}{Yes}&\multirow{2}{*}{No}&\multirow{2}{*}{1-6,9}&\multirow{2}{*}{Yes}&\multirow{2}{*}{Yes}\\
&&&&&&&&&&&&&&&&&&&\\\hline

\multirow{2}{*}{\bf
\cite{Agr}}&\multirow{2}{*}{O}&\multirow{2}{*}{No}&\multirow{2}{*}{Maintain}&\multirow{2}{*}{Run}
&\multirow{2}{*}{Logic}&\multirow{2}{*}{Yes}&\multirow{2}{*}{Limited}&\multirow{2}{*}{Static}&\multirow{2}{*}{Yes}&\multirow{2}{*}{
$\iota~\delta$}&\multirow{2}{*}{
---}&\multirow{2}{*}{Yes}&\multirow{2}{*}{No}&\multirow{2}{*}{S}&\multirow{2}{*}{Yes}&\multirow{2}{*}{No}&\multirow{2}{*}{---}&\multirow{2}{*}{No}&\multirow{2}{*}{No}\\
&&&&&&&&&&&&&&&&&&&\\\hline

\multirow{2}{*}{\bf
\cite{Sch2}}&\multirow{2}{*}{N}&\multirow{2}{*}{No}&\multirow{2}{*}{Maintain}&\multirow{2}{*}{Comp}
&\multirow{2}{*}{Relation}&\multirow{2}{*}{No}&\multirow{2}{*}{Limited}&\multirow{2}{*}{Static}&\multirow{2}{*}{Yes}&\multirow{2}{*}{
$\iota~\delta$}&Predef&\multirow{2}{*}{No}&\multirow{2}{*}{No}&\multirow{2}{*}{G}&\multirow{2}{*}{No}&\multirow{2}{*}{No}&\multirow{2}{*}{---}&\multirow{2}{*}{No}&\multirow{2}{*}{No}\\
&&&&&&&&&&&Program&&&&&&&&\\\hline

\multirow{2}{*}{\bf
\cite{Hal}}&\multirow{2}{*}{N}&\multirow{2}{*}{No}&\multirow{2}{*}{Maintain}&\multirow{2}{*}{Comp}
&\multirow{2}{*}{Logic}&\multirow{2}{*}{Yes}&\multirow{2}{*}{Limited}&\multirow{2}{*}{Static}&\multirow{2}{*}{Yes}&\multirow{2}{*}{
$\iota~\delta$}&\multirow{2}{*}{
---}&\multirow{2}{*}{No}&\multirow{2}{*}{No}&\multirow{2}{*}{S}&\multirow{2}{*}{No}&\multirow{2}{*}{No}&\multirow{2}{*}{---}&\multirow{2}{*}{No}&\multirow{2}{*}{No}\\
&&&&&&&&&&&&&&&&&&&\\\hline

\multirow{2}{*}{\bf
\cite{Chi}}&\multirow{2}{*}{N}&\multirow{2}{*}{No}&\multirow{2}{*}{Maintain}&Comp.
&\multirow{2}{*}{Relation}&\multirow{2}{*}{Yes}&\multirow{2}{*}{Limited}&Static&\multirow{2}{*}{Yes}&\multirow{2}{*}{
$\iota~\delta$}&\multirow{2}{*}{
---}&\multirow{2}{*}{Yes}&\multirow{2}{*}{No}&\multirow{2}{*}{S}&\multirow{2}{*}{Yes}&\multirow{2}{*}{No}&\multirow{2}{*}{---}&Not&Not\\
&&&&Run&&&&Dynamic&&&&&&&&&&proved&proved\\\hline

\multirow{2}{*}{\bf
\cite{Dom}}&\multirow{2}{*}{H}&\multirow{2}{*}{Yes}&\multirow{2}{*}{Check}&\multirow{2}{*}{Run}
&\multirow{2}{*}{Logic}&\multirow{2}{*}{Yes}&\multirow{2}{*}{Yes}&\multirow{2}{*}{Static}&\multirow{2}{*}{Yes}&\multirow{2}{*}{
$\iota~\delta$}&
Predef&\multirow{2}{*}{Yes}&\multirow{2}{*}{No}&\multirow{2}{*}{S}&\multirow{2}{*}{Yes}&\multirow{2}{*}{No}&\multirow{2}{*}{1-6,7}&\multirow{2}{*}{Yes}&Not\\
&&&&&&&&&&&Programs&&&&&&&&proved\\\hline

\multirow{2}{*}{\bf
\cite{Heg1}}&\multirow{2}{*}{O}&\multirow{2}{*}{Yes}&\multirow{2}{*}{Check}&\multirow{2}{*}{Run}
&\multirow{2}{*}{Relation}&\multirow{2}{*}{Yes}&\multirow{2}{*}{Limited}&\multirow{2}{*}{Static}&\multirow{2}{*}{No}&\multirow{2}{*}{
$\iota~\delta$}&\multirow{2}{*}{
Unfold}&\multirow{2}{*}{Yes}&\multirow{2}{*}{No}&\multirow{2}{*}{S}&\multirow{2}{*}{Yes}&\multirow{2}{*}{No}&\multirow{2}{*}{1-6,9}&\multirow{2}{*}{Yes}&\multirow{2}{*}{Yes}\\
&&&&&&&&&&&&&&&&&&&\\\hline

\multirow{2}{*}{\bf
\cite{Baue}}&\multirow{2}{*}{O}&\multirow{2}{*}{Yes}&\multirow{2}{*}{Check}&\multirow{2}{*}{Run}
&\multirow{2}{*}{Relation}&\multirow{2}{*}{Yes}&\multirow{2}{*}{Limited}&\multirow{2}{*}{Static}&\multirow{2}{*}{No}&\multirow{2}{*}{
$\iota~\delta$}&\multirow{2}{*}{
Unfold}&\multirow{2}{*}{Yes}&\multirow{2}{*}{No}&\multirow{2}{*}{S}&\multirow{2}{*}{Yes}&\multirow{2}{*}{No}&\multirow{2}{*}{1-6,9}&\multirow{2}{*}{Yes}&\multirow{2}{*}{Yes}\\
&&&&&&&&&&&&&&&&&&&\\\hline

\multirow{2}{*}{\bf
\cite{Far}}&\multirow{2}{*}{N}&\multirow{2}{*}{Yes}&\multirow{2}{*}{Check}&\multirow{2}{*}{Run}
&\multirow{2}{*}{Logic}&\multirow{2}{*}{Yes}&\multirow{2}{*}{Yes}&\multirow{2}{*}{Static}&\multirow{2}{*}{Yes}&\multirow{2}{*}{
$\iota~\delta$}&\multirow{2}{*}{
SLDNF}&\multirow{2}{*}{Yes}&\multirow{2}{*}{No}&\multirow{2}{*}{S}&\multirow{2}{*}{Yes}&\multirow{2}{*}{No}&\multirow{2}{*}{1-6,9}&\multirow{2}{*}{Yes}&\multirow{2}{*}{Yes}\\
&&&&&&&&&&&&&&&&&&&\\\hline

\multirow{2}{*}{\bf
\cite{Saha}}&\multirow{2}{*}{N}&\multirow{2}{*}{No}&\multirow{2}{*}{Maintain}&\multirow{2}{*}{Run}
&\multirow{2}{*}{Logic}&\multirow{2}{*}{Yes}&\multirow{2}{*}{Limited}&\multirow{2}{*}{Static}&\multirow{2}{*}{Yes}&\multirow{2}{*}{
$\iota~\delta$}&\multirow{2}{*}{
Predef}&\multirow{2}{*}{Yes}&\multirow{2}{*}{No}&\multirow{2}{*}{S}&\multirow{2}{*}{Yes}&\multirow{2}{*}{No}&\multirow{2}{*}{1-6,7}&\multirow{2}{*}{Yes}&Not\\
&&&&&&&&&&&Programs&&&&&&&&proved\\\hline

\multirow{2}{*}{\bf
\cite{Hor}}&\multirow{2}{*}{O}&\multirow{2}{*}{No}&\multirow{2}{*}{Maintain}&\multirow{2}{*}{Comp}
&\multirow{2}{*}{Relation}&\multirow{2}{*}{Yes}&\multirow{2}{*}{Limited}&\multirow{2}{*}{Static}&\multirow{2}{*}{Yes}&\multirow{2}{*}{
$\iota~\delta$}&\multirow{2}{*}{
---}&\multirow{2}{*}{Yes}&\multirow{2}{*}{No}&\multirow{2}{*}{S}&\multirow{2}{*}{Yes}&\multirow{2}{*}{No}&\multirow{2}{*}{---}&\multirow{2}{*}{No}&\multirow{2}{*}{No}\\
&&&&&&&&&&&&&&&&&&&\\\hline

\end{tabular}
\end{sidewaystable}

\begin{sidewaystable}
\centering
\begin{turn}{180}
\begin{tabular}{|l|l|l|l|l|l|l|l|l|l|l|l|l|l|l|l|l|l|l|l|}\hline
\multirow{3}{*}{\bf Method}&\multicolumn{4}{c}{\bf Problem}&
\multicolumn{4}{|c|}{\bf Database
  schema}&\multicolumn{2}{c}{\bf Update req.}&\multicolumn{3}{|c|}{\bf Mechanism}&\multicolumn{3}{c}{\bf Update Change}&
\multicolumn{3}{|c|}{\bf Solutions}\\
\cline{2-20}

 &\multirow{2}{*}{Type}&View&IC&Run/&Def.&\multirow{2}{*}{View}&IC&Kind of
&\multirow{2}{*} {Mul.}&Update&Tech-&Base&User&\multirow{2}{*}{Type}&~Base~&View&\multirow{2}{*}{Axiom}&\multirow{2}{*}{Sound.}&\multirow{2}{*}{Complete.}\\
&&Update&Enforce.&Comp.&Lang.&&def.&IC&&Operat.&nique&Facts&Part.&&facts&def.&&&\\\hline

\multirow{2}{*}{\bf
\cite{Sak3}}&\multirow{2}{*}{N}&\multirow{2}{*}{Yes}&\multirow{2}{*}{Check}&\multirow{2}{*}{Run}
&\multirow{2}{*}{Logic}&\multirow{2}{*}{No}&\multirow{2}{*}{Limited}&\multirow{2}{*}{Static}&\multirow{2}{*}{Yes}&\multirow{2}{*}{
$\iota~\delta$}&\multirow{2}{*}{
SLDNF}&\multirow{2}{*}{Yes}&\multirow{2}{*}{No}&\multirow{2}{*}{S}&\multirow{2}{*}{Yes}&\multirow{2}{*}{No}&\multirow{2}{*}{1-6,9}&\multirow{2}{*}{Yes}&\multirow{2}{*}{Yes}\\
&&&&&&&&&&&&&&&&&&&\\\hline

\multirow{2}{*}{\bf
\cite{Far1}}&\multirow{2}{*}{N}&\multirow{2}{*}{Yes}&\multirow{2}{*}{Check}&\multirow{2}{*}{Run}
&\multirow{2}{*}{Logic}&\multirow{2}{*}{Yes}&\multirow{2}{*}{Yes}&\multirow{2}{*}{Static}&\multirow{2}{*}{No}&\multirow{2}{*}{
$\iota~\delta$}&\multirow{2}{*}{
SLDNF}&\multirow{2}{*}{No}&\multirow{2}{*}{No}&\multirow{2}{*}{S}&\multirow{2}{*}{Yes}&\multirow{2}{*}{No}&\multirow{2}{*}{1-6,9}&Not&\multirow{2}{*}{No}\\
&&&&&&&&&&&&&&&&&&&proved\\\hline

\multirow{2}{*}{\bf
\cite{Mar}}&\multirow{2}{*}{N}&\multirow{2}{*}{Yes}&Check&\multirow{2}{*}{Run}
&\multirow{2}{*}{Logic}&\multirow{2}{*}{Yes}&\multirow{2}{*}{Yes}&\multirow{2}{*}{Static}&\multirow{2}{*}{No}&\multirow{2}{*}{
$\iota~\delta~\chi$}&\multirow{2}{*}{
SLD}&\multirow{2}{*}{Yes}&\multirow{2}{*}{No}&\multirow{2}{*}{S}&\multirow{2}{*}{Yes}&\multirow{2}{*}{No}&\multirow{2}{*}{---}&\multirow{2}{*}{No}&\multirow{2}{*}{No}\\
&&&Maintain&&&&&&&&&&&&&&&&\\\hline

\multirow{2}{*}{\bf
\cite{Bra}}&\multirow{2}{*}{N}&\multirow{2}{*}{Yes}&Check&Run
&\multirow{2}{*}{Logic}&\multirow{2}{*}{Yes}&\multirow{2}{*}{Yes}&\multirow{2}{*}{Static}&\multirow{2}{*}{No}&\multirow{2}{*}{
$\iota~\delta~\chi$}&\multirow{2}{*}{
SLD}&\multirow{2}{*}{Yes}&\multirow{2}{*}{No}&\multirow{2}{*}{SS}&\multirow{2}{*}{Yes}&\multirow{2}{*}{No}&\multirow{2}{*}{1-6,7}&\multirow{2}{*}{Yes}&\multirow{2}{*}{Not}\\
&&&Maintain&Comp&&&&&&&&&&&&&&&proved\\\hline

\multirow{2}{*}{\bf
\cite{Chris2}}&\multirow{2}{*}{N}&\multirow{2}{*}{Yes}&Check&\multirow{2}{*}{Run}
&\multirow{2}{*}{Logic}&\multirow{2}{*}{Yes}&\multirow{2}{*}{Yes}&Static&\multirow{2}{*}{Yes}&\multirow{2}{*}{
$\iota~\delta~\chi$}&\multirow{2}{*}{
Predef}&\multirow{2}{*}{Yes}&\multirow{2}{*}{No}&\multirow{2}{*}{S}&\multirow{2}{*}{Yes}&\multirow{2}{*}{No}&\multirow{2}{*}{1-6,9}&\multirow{2}{*}{Yes}&\multirow{2}{*}{Yes}\\
&&&Maintain&&&&&Dynamic&&&Program&&&&&&&&\\\hline

\multirow{2}{*}{\bf
\cite{Car}}&\multirow{2}{*}{N}&\multirow{2}{*}{Yes}&\multirow{2}{*}{Check}&\multirow{2}{*}{Comp}
&\multirow{2}{*}{Logic}&\multirow{2}{*}{Yes}&\multirow{2}{*}{Yes}&\multirow{2}{*}{Dynamic}&\multirow{2}{*}{Yes}&\multirow{2}{*}{
$\iota~\delta$}&\multirow{2}{*}{
Predef}&\multirow{2}{*}{Yes}&\multirow{2}{*}{No}&\multirow{2}{*}{S}&\multirow{2}{*}{Yes}&\multirow{2}{*}{No}&\multirow{2}{*}{---}&Not&\multirow{2}{*}{No}\\
&&&&&&&&&&&Programs&&&&&&&Proved&proved\\\hline

\multirow{2}{*}{\bf
\cite{Chris1}}&\multirow{2}{*}{N}&\multirow{2}{*}{Yes}&Check&\multirow{2}{*}{Run}
&\multirow{2}{*}{Logic}&\multirow{2}{*}{Yes}&\multirow{2}{*}{Yes}&Static&\multirow{2}{*}{Yes}&\multirow{2}{*}{
$\iota~\delta~\chi$}&\multirow{2}{*}{
Predef}&\multirow{2}{*}{Yes}&\multirow{2}{*}{No}&\multirow{2}{*}{S}&\multirow{2}{*}{Yes}&\multirow{2}{*}{No}&\multirow{2}{*}{1-6,9}&\multirow{2}{*}{Yes}&\multirow{2}{*}{Yes}\\
&&&Maintain&&&&&Dynamic&&&Program&&&&&&&&\\\hline

\multirow{2}{*}{\bf
\cite{Coe}}&\multirow{2}{*}{N}&\multirow{2}{*}{No}&\multirow{2}{*}{Maintain}&\multirow{2}{*}{Comp}
&\multirow{2}{*}{Logic}&\multirow{2}{*}{Yes}&\multirow{2}{*}{No}&\multirow{2}{*}{---}&\multirow{2}{*}{Yes}&\multirow{2}{*}{
$\iota~\delta$}&\multirow{2}{*}{
---}&\multirow{2}{*}{Yes}&\multirow{2}{*}{No}&\multirow{2}{*}{S}&\multirow{2}{*}{Yes}&\multirow{2}{*}{No}&\multirow{2}{*}{---}&\multirow{2}{*}{No}&\multirow{2}{*}{No}\\
&&&&&&&&&&&&&&&&&&&\\\hline

\multirow{2}{*}{\bf
\cite{Zho}}&\multirow{2}{*}{N}&\multirow{2}{*}{No}&\multirow{2}{*}{Maintain}&\multirow{2}{*}{Run}
&\multirow{2}{*}{Relation}&\multirow{2}{*}{Yes}&\multirow{2}{*}{No}&\multirow{2}{*}{---}&\multirow{2}{*}{Yes}&\multirow{2}{*}{
$\iota~\delta~\chi$}&\multirow{2}{*}{
Unfold}&\multirow{2}{*}{Yes}&\multirow{2}{*}{No}&\multirow{2}{*}{SS}&\multirow{2}{*}{No}&\multirow{2}{*}{No}&\multirow{2}{*}{---}&Not&Not\\
&&&&&&&&&&&&&&&&&&proved&proved\\\hline

\multirow{2}{*}{\bf
\cite{Heg}}&\multirow{2}{*}{O}&\multirow{2}{*}{No}&\multirow{2}{*}{Maintain}&Comp.
&\multirow{2}{*}{Logic}&\multirow{2}{*}{Yes}&\multirow{2}{*}{No}&\multirow{2}{*}{---}&\multirow{2}{*}{Yes}&\multirow{2}{*}{
$\iota~\delta$}&\multirow{2}{*}{
----}&\multirow{2}{*}{Yes}&\multirow{2}{*}{No}&\multirow{2}{*}{G}&\multirow{2}{*}{No}&\multirow{2}{*}{No}&\multirow{2}{*}{---}&\multirow{2}{*}{Yes}&Not\\
&&&&Run&&&&&&&&&&&&&&&proved\\\hline

\multirow{2}{*}{\bf
\cite{Beh}}&\multirow{2}{*}{S}&\multirow{2}{*}{Yes}&\multirow{2}{*}{Check}&\multirow{2}{*}{Run}
&\multirow{2}{*}{Logic}&\multirow{2}{*}{Yes}&Flat&\multirow{2}{*}{Static}&\multirow{2}{*}{Yes}&\multirow{2}{*}{
$\iota~\delta$}&\multirow{2}{*}{
SLDNF}&\multirow{2}{*}{Yes}&\multirow{2}{*}{No}&\multirow{2}{*}{S}&\multirow{2}{*}{Yes}&\multirow{2}{*}{No}&\multirow{2}{*}{1-6,9}&\multirow{2}{*}{Yes}&Not\\
&&&&&&&Limited&&&&&&&&&&&&proved\\\hline

\multirow{2}{*}{\bf
\cite{Dom1}}&\multirow{2}{*}{N}&\multirow{2}{*}{Yes}&\multirow{2}{*}{Check}&\multirow{2}{*}{Run}
&\multirow{2}{*}{Logic}&\multirow{2}{*}{Yes}&\multirow{2}{*}{Yes}&\multirow{2}{*}{Static}&\multirow{2}{*}{Yes}&\multirow{2}{*}{
$\iota~\delta$}&\multirow{2}{*}{
---}&\multirow{2}{*}{Yes}&\multirow{2}{*}{No}&\multirow{2}{*}{S}&\multirow{2}{*}{Yes}&\multirow{2}{*}{No}&\multirow{2}{*}{---}&\multirow{2}{*}{No}&\multirow{2}{*}{No}\\
&&&&&&&&&&&&&&&&&&&\\\hline

\multirow{2}{*}{\bf
\cite{Cha}}&\multirow{2}{*}{O}&\multirow{2}{*}{No}&\multirow{2}{*}{Maintain}&\multirow{2}{*}{Run}
&\multirow{2}{*}{Relation}&\multirow{2}{*}{Yes}&\multirow{2}{*}{No}&\multirow{2}{*}{Static}&\multirow{2}{*}{Yes}&\multirow{2}{*}{
$\iota~\delta$}&\multirow{2}{*}{
SLD}&\multirow{2}{*}{Yes}&\multirow{2}{*}{Yes}&\multirow{2}{*}{G}&\multirow{2}{*}{No}&\multirow{2}{*}{No}&\multirow{2}{*}{---}&Not&Not\\
&&&&&&&&&&&&&&&&&&proved&proved\\\hline

\multirow{2}{*}{\bf
\cite{Bell}}&\multirow{2}{*}{O}&\multirow{2}{*}{No}&\multirow{2}{*}{Maintain}&\multirow{2}{*}{Comp}
&\multirow{2}{*}{Relation}&\multirow{2}{*}{Yes}&\multirow{2}{*}{No}&\multirow{2}{*}{Static}&\multirow{2}{*}{Yes}&\multirow{2}{*}{
$\iota~\delta~\chi$}&\multirow{2}{*}{
---}&\multirow{2}{*}{Yes}&\multirow{2}{*}{No}&\multirow{2}{*}{SS}&\multirow{2}{*}{Yes}&\multirow{2}{*}{No}&\multirow{2}{*}{---}&\multirow{2}{*}{No}&\multirow{2}{*}{No}\\
&&&&&&&&&&&&&&&&&&&\\\hline

\multirow{2}{*}{\bf
\cite{Alex}}&\multirow{2}{*}{O}&\multirow{2}{*}{No}&\multirow{2}{*}{Maintain}&Comp.
&\multirow{2}{*}{Relation}&\multirow{2}{*}{No}&\multirow{2}{*}{Limited}&Static&\multirow{2}{*}{Yes}&\multirow{2}{*}{
$\iota~\delta$}&\multirow{2}{*}{
---}&\multirow{2}{*}{Yes}&\multirow{2}{*}{No}&\multirow{2}{*}{G}&\multirow{2}{*}{Yes}&\multirow{2}{*}{No}&\multirow{2}{*}{---}&\multirow{2}{*}{No}&\multirow{2}{*}{No}\\
&&&&Run&&&&Dynamic&&&&&&&&&&&\\\hline

\multirow{2}{*}{\bf
\cite{Mam1}}&\multirow{2}{*}{N}&\multirow{2}{*}{No}&\multirow{2}{*}{Maintain}&\multirow{2}{*}{Comp}
&\multirow{2}{*}{Relation}&\multirow{2}{*}{No}&\multirow{2}{*}{Yes}&\multirow{2}{*}{Static}&\multirow{2}{*}{Yes}&\multirow{2}{*}{
$\iota~\delta~\chi$}&\multirow{2}{*}{
Unfold}&\multirow{2}{*}{No}&\multirow{2}{*}{Yes}&\multirow{2}{*}{SS}&\multirow{2}{*}{No}&\multirow{2}{*}{No}&\multirow{2}{*}{---}&\multirow{2}{*}{No}&\multirow{2}{*}{No}\\
&&&&&&&&&&&&&&&&&&&\\\hline

\multirow{2}{*}{\bf
\cite{Sal}}&\multirow{2}{*}{N}&\multirow{2}{*}{No}&\multirow{2}{*}{Check}&\multirow{2}{*}{Comp}
&\multirow{2}{*}{Logic}&\multirow{2}{*}{No}&\multirow{2}{*}{Yes}&\multirow{2}{*}{Static}&\multirow{2}{*}{Yes}&\multirow{2}{*}{
$\iota~\delta$}&\multirow{2}{*}{
Active}&\multirow{2}{*}{Yes}&\multirow{2}{*}{No}&\multirow{2}{*}{G}&\multirow{2}{*}{Yes}&\multirow{2}{*}{No}&\multirow{2}{*}{---}&\multirow{2}{*}{No}&\multirow{2}{*}{No}\\
&&&&&&&&&&&&&&&&&&&\\\hline

\multirow{2}{*}{\bf
\cite{Del}}&\multirow{2}{*}{N}&\multirow{2}{*}{Yes}&\multirow{2}{*}{Check}&\multirow{2}{*}{Run}
&\multirow{2}{*}{Logic}&\multirow{2}{*}{Yes}&\multirow{2}{*}{Yes}&\multirow{2}{*}{Static}&\multirow{2}{*}{Yes}&\multirow{2}{*}{
$\iota~\delta$}&\multirow{2}{*}{
SLD}&\multirow{2}{*}{Yes}&\multirow{2}{*}{No}&\multirow{2}{*}{S}&\multirow{2}{*}{Yes}&\multirow{2}{*}{No}&\multirow{2}{*}{1-6,9}&
\multirow{2}{*}{Yes}&\multirow{2}{*}{Yes}\\
&&&&&&&&&&&&&&&&&&&\\\hline

\end{tabular}
\end{turn}
\end{sidewaystable}

\section*{Acknowledgement}
I would like to thanks Chandrabose Aravindan and Gerhard Lakemeyer both my Indian and Germany PhD supervisor, give encourage to write the paper. I owe my deepest gratitude to Ramaswamy Ramanujam from Institute of Mathematical Sciences and Ulrich Furbach from University Koblenz-Landau for my PhD thesis member. I thank to my PhD thesis examiner's Eduardo  Ferm{\'e} from University of Madeira, Mohua Banerjee from Indian Institute of Technology Kanpur and Arindama Singh form Indian Institute of Technology Madras. The author acknowledges the support of SSN College of Engineering research funds and RWTH Aachen, where he is visiting scholar with an Erasmus Mundus External Cooperation Window India4EU by the European Commission.



\end{document}